\documentclass[hidelinks,onefignum,onetabnum]{siamart251104}

\usepackage{lipsum}
\usepackage{amsfonts}
\usepackage{graphicx}
\usepackage{epstopdf}
\usepackage{makecell}
\usepackage{subfig}
\usepackage{algorithmic}
\usepackage{placeins}
\usepackage{titlesec}
\usepackage{amssymb}

\ifpdf
  \DeclareGraphicsExtensions{.eps,.pdf,.png,.jpg}
\else
  \DeclareGraphicsExtensions{.eps}
\fi

\newsiamremark{remark}{Remark}
\newsiamremark{hypothesis}{Hypothesis}
\crefname{hypothesis}{Hypothesis}{Hypotheses}
\newsiamthm{claim}{Claim}
\newsiamremark{fact}{Fact}
\crefname{fact}{Fact}{Facts}

\headers{}{C. Zhong, Z. Li, S. Xu, X. Li, L. Zhang, and J. Yuan}

\title{A Globally Convergent Variational Framework for Mode Number Detection via Spectral Cutting Curves\thanks{Chenjie Zhong and Zhipeng Li are co-first authors who contributed equally to this article.
\funding{This study is supported by the National Natural Science Foundation of China under Grant No. 61773290 and the Fundamental Research Funds for the Central Universities (22120230311) and Tongji University Medicine-X Interdisciplinary Research Initiative (Grant No. 2025-0554-YB-11).}}}

\author{
Chenjie Zhong\textsuperscript{\textdagger} 
\and 
Zhipeng Li\textsuperscript{\textdagger,*} 
\and 
Shangzhi Xu\textsuperscript{\textdagger}
\\
Xiaohu Li\textsuperscript{\textdaggerdbl}
\and
Luodan Zhang\textsuperscript{\textdaggerdbl}
\and
Jianjun Yuan\textsuperscript{\textdagger}
}

\footnotetext[2]{State Key Laboratory of High-speed Maglev Transportation Technology, and Department of Information and Communication Engineering, Tongji University, Shanghai 201804, China}
\footnotetext[3]{NingBo Communication Construction Engineering Testing Center Co. Ltd}
\footnotetext[1]{Corresponding author: \email{lizhipeng@tongji.edu.cn}}

\usepackage{amsopn}

\ifpdf
\hypersetup{
  pdftitle={A Globally Convergent Variational Framework for Mode Number Detection via Spectral Cutting Curves},
  pdfauthor={Chenjie Zhong, Zhipeng Li, Shangzhi Xu, Xiaohu Li, Luodan Zhang, and Jianjun Yuan}
}
\fi

\begin{document}

\maketitle

\begin{abstract}
Automatically determining the number of intrinsic mode functions (IMFs) and their center frequencies in Variational Mode Decomposition (VMD) remains an open mathematical challenge.  Existing methods often rely on heuristic prior settings, trial-and-error strategies, or recursive extraction procedures that are computationally inefficient and prone to error accumulation, and generally lack theoretical guarantees on convergence. In this article, we propose a novel variational framework that endogenously determines the number of modes. We observe that any curve lying below the spectral amplitude naturally divides the spectral domain into connected regions where the spectrum rises above the curve, the number of which defines the modal number \(K[g]\)---a topological functional induced by the cutting curve. Since \(K[g]\) is discontinuous and intractable for direct optimization, we instead seek the optimal cutting curve itself as a continuous variational surrogate: the optimal curve best separates distinct spectral peaks into individual connected regions above it, while merging noise-induced fragments below. This surrogate adversarially maximizes the integral of \(g\) (encouraging it to rise and support significant peaks) while penalizing its curvature (suppressing excessive undulations that would fragment the spectrum), transforming the problem into iteratively solving a fourth-order boundary value problem via Lagrangian duality. We establish a rigorous proof of global convergence for the dual ascent algorithm in function space.   Comprehensive numerical experiments on artificial and real-world signals including electrocardiogram (ECG) data show that our method can provide accurate estimates of IMFs and center frequencies, and comparison with methods like Successive VMD also shows our advantages in avoiding redundant modes while ensuring the recovery of necessary components, indicating that we have provided a robust, theoretically grounded initialization routine for VMD.
\end{abstract}


\clearpage

\section{Introduction}

Variational Mode Decomposition (VMD) decomposes original signals into multiple Intrinsic Mode Functions (IMF) with limited bandwidth with the help of variational optimization.\cite{Dragomiretskiy}   The core of VMD is to minimize the sum of the bandwidths with respect to the estimated IMFs as well as to ensure the sum of all the IMFs is equal to the original signal.  In VMD, it is assumed the frequency of each mode is almost compact around a center pulsation and the associated analytic signal is computed by means of Hilbert Transform to get the unilateral frequency spectrum.  Then each mode's frequency spectrum was shifted into baseband by multiplying a frequency modulation factor with respect to the estimated center frequency, transforming each band-pass mode into low pass, and the bandwidth can be estimated by $H^1$ Gaussian smoothness.  Afterwards, they decompose the original signal into each modes by minimizing the sum of the bandwidth of each mode with the constraints that the summation of all modes is the original signal.  In short, their key variation optimization can be listed as following:

\begin{equation}
    \min\limits_{u_k,\omega_k}\left \{ 
        \sum\limits_k \left \Vert 
            \partial_t \left[ \left( \delta(t) + \frac{\text{j}}{\pi t} \right) * u_k(t) \right] e^{-\text{j}\omega_k t} 
        \right \Vert_2^2 
    \right \}
\end{equation}
\begin{equation}
    \text{s.t.}\quad \sum\limits_k u_k = f 
\end{equation}
so that each mode can be altered into an associated analytical signal by using Hilbert Transform and shifted into base-band for the following estimation of bandwidth using $H^1$ Gaussian smoothness with the shifted demodulated components.  To solve the problem, one can construct an augmented Lagrangian Function by introducing both quadratic penalty and Lagrangian Multipliers as follows:
\begin{equation}
\begin{aligned}
\mathcal{L}(u_k, \omega_k, \lambda) &= \alpha \sum_{k} \left\Vert \partial_t \left[ \left( \delta(t) + \frac{j}{\pi t} \right) * u_k(t) \right] e^{-j \omega_k t} \right\Vert_2^2\\
&\quad\quad\quad\quad\quad\quad\quad + \left\Vert f - \sum u_k \right\Vert_2^2 + \langle \lambda, f - \sum u_k \rangle
\end{aligned}
\end{equation}
This problem can be solved by using ADMM(Alternate Direction Method of Multipliers), in which each component can be updated by solving the equivalent minimization problem:
\begin{equation}
u_k^{n+1} = \underset{u_k \in \mathbb{R}}{\arg\min} \left\{ \alpha \left\Vert \partial_t \left[ \left( \delta(t) + \frac{j}{\pi t} \right) * u_k(t) \right] e^{-j\omega_k t} \right\Vert_2^2 + \left\Vert f - \sum u_i + \frac{\lambda}{2} \right\Vert_2^2 \right\}.
\end{equation}
One can rewrite it in frequency domain form, so that convolution becomes multiplication for the convenience of the calculation in subsequent variational operations:
\begin{equation}
\hat{u}_k^{n+1} =\underset{\hat{u}_k,\, \hat{u}_k = \hat{u}_k^T}{\arg\min} \left\{ \alpha \left\Vert j\omega \left[ 1 + \text{sgn}(\omega + \omega_k) \right] \hat{u}_k(\omega + \omega_k) \right\Vert_2^2 + \left\Vert \hat{f} - \sum \hat{u}_i + \frac{\lambda}{2} \right\Vert_2^2 \right\}.
\end{equation}
and by substitution of $\omega \rightarrow \omega + \omega_k $ to perform a translation and considering the Hermitian symmetry of spectrum with respect to the real signals, both terms can be transformed into half-space integrals with non-negative frequencies:
\begin{equation}
\hat{u}_k^{n+1} = \underset{u_k,\, \hat{u}_k = \hat{u}_k^T}{\arg\min} \left\{ \int_0^\infty 4\alpha (\omega - \omega_k)^2 | \hat{u}_k(\omega) |^2 + 2 \left( \hat{f} - \sum \hat{u}_i + \frac{\hat{\lambda}}{2} \right)^2 \text d\omega \right\}
\end{equation}
which can be evaluated by vanishing the first ordered variation of positive frequencies for both signal components and center frequencies
\begin{equation}
\hat{u}_k^{n+1} = \left( \hat{f} - \sum_{i \neq k} \hat{u}_i + \frac{\hat{\lambda}}{2} \right) \frac{1}{1 + 2\alpha (\omega - \omega_k)^2}
\end{equation}
\begin{equation}
\omega_k^{n+1} = \frac{\displaystyle\int_0^\infty \omega \left| \hat{u}_k(\omega) \right|^2 \text d\omega}{\displaystyle\int_0^\infty \left| \hat{u}_k(\omega) \right|^2 \text d\omega}
\end{equation}

Since its establishment, VMD is always playing an important role in modern signal decomposition in many fields\cite{Sharma,ShangZhang,Li,Kumaraguruparan,Wen-Chao,Group,Gianmarco}. Although VMD has provided flexibility in decomposing signals, there are factors that hinder its further applications in that the number of intrinsic mode function must be determined manually as a prior in coordination with penalty factor without a common principle\cite{Yang,Yunqian,Xia,Wu}.  In order to alleviate the problem, the subsequent research either needs the knowledge to the range of the number of modes\cite{Lian}, retrieves the IMFs recursively\cite{Mojtaba}, or highly dependent on a predefined window-size to retrieve the peak and valley of spectra \cite{Feng}.  All these methods either require prior or rely on schemes devoid of convergence guarantees, making the results contains spurious or missed modes.  The absence of a well-posed, convergent paradigm for automatically finding the number of IMFs is a crucial factor that hinders the development of VMD.

To close this gap, we introduce a fundamentally different approach for finding the IMFs.  Instead of retrieving modes in an trial-and-error way or recursively extracting modes, we focus on the spectrum amplitude and reformulate the problem by introducing Cutting Curve, a lower-enveloping curve of a signal's spectrum which can support the spectrum as tightly as possible while keeping itself as smooth as possible, and transform the problem for finding IMFs into first evaluating the Cutting Curve by adversarially maximizing its integral while penalizing its curvature, then retrieving the separated peaks of the original spectrum above the Cutting Curve.

Our main contribution can be summarized as 4 folds.  First, we give a new perspective for determining the number of IMFs by introducing Cutting Curve of spectrum and solve the problem in real field.  Second, we derive its optimality conditions, establishing a rigorous equivalence between this variational problem and a fourth-order boundary value problem, and establish rigorous mathematical proof for the global convergence to the course of dual ascent iteration in function space.  Thirdly, we develop an efficient iteration scheme for the algorithm by converting the corresponding  difference equation and boundary conditions into matrix format and expanding compatible rule to the component-wise multiplication between matrices and vectors for Hadamard product.  Finally, we give experiments on both artificial and real signals, and comparison with SVMD to show the robustness of our algorithm and experiments show that our method can quickly get accurate estimation about the number of IMFs and the corresponding initial center frequencies as prior to the following VMD procedure.

This article will be organized as follows:  In the following section we propose the main framework of our algorithm and give the proof of the convergence.  In section 3 we introduce the implementation detail about the whole workflow.  In Section 4 we design a series of experiments to evaluate the performance of our algorithm.  In the final section we show our conclusion about the main contribution of the paper.

\section{Methodology and Convergence Proof}
According to the theory of Variational Mode Decomposition, each intrinsic mode function should manifest in the spectral domain as a prominent convex bulge, or alternatively as a cluster of adjacent minor bulges. In an ideal, flat spectrum, identifying such bulges would be straightforward. However, real-world spectra invariably exhibit an underlying trend—a slowly varying floor upon which the modal peaks sit. This trend obscures the boundaries between modes, rendering the extraction of bulges far from trivial.

In this section, we develop a principled framework to address this difficulty. Our strategy proceeds as follows. First, we formalize the intuitive conditions that an optimal modal partition should satisfy: each detected mode must accumulate sufficient residual energy with respect to the trend of the spectrum; adjacent weak peaks that individually lack saliency should be merged; and distinct modes should be separated by sufficiently deep valleys.  Second, let $f(x)$ be the spectrum of any investigated signal, we recognize that any curve $g(x)$ lying below $f(x)$ will introduce a partition with respect to the frequency-axis into connected intervals where $f(x) - g(x) > 0$.  The number of each connected intervals defines the modal number $K(g(x))$ which is a general function of $g(x)$, and the problem of determining the optimal number $K$ is then recast to selecting an optimal curve $g(x)$ to cut the spectrum, or the optimal Cutting Curve $g^*(x)$.

Since $K(g(x))$ is a discontinuous functional that is hard to be optimized, we construct an surrogate continuous functional which can automatically yields a cutting curve satisfying our assumption that mentioned above.  The surrogate functional defines a conditional variational problem that adversarially maximizing the integral of the curve and minimizing the curvature.  By finding the optimal curve defined by the surrogate function satisfying the assumption of optimal $K$, one can evaluate the number of modes in a continuous functional framework, which is much more easier than imposed in the discontinuous circumstances.

Additionally, we can prove that the optimization of the surrogate functional is a convex problem which can be transformed into an variational problem, and be transformed into a fourth-ordered boundary-value-problem, and a projected dual-ascent algorithm is developed to get the optimal solution.  We can also establish the global convergence to the unique minimizer in the function space.  The remainder of the part provides the mathematical detail of the proposed method.

\subsection{Modal Number as a Functional of the Cutting Curve} In this section we introduce Cutting Curve of a Spectrum.

\noindent \textbf{ Definition 1: Cutting Curve and Support Interval Sets}

Let $f(x) \ge 0$ be any amplitude spectrum of a signal defined on a closed interval $\Omega \subset \mathbb R$, for any continuous function $g(x) \in C(\Omega)$ satisfying $0 \le g(x) \le f(x)$, we define the supporting intervals $\mathcal S(g(x))$ as 
\begin{equation}
\mathcal S(g(x)) = \{x \in \Omega | 0 \le g(x) \le f(x)\}
\end{equation}
and $g(x)$ as a cutting curve, then $\mathcal S(g(x))$ can be uniquely decomposed into a series countably many disjoint open sets
\begin{equation}
\mathcal S(g(x)) = \bigcup_{k=1}^{K(g(x))}C_k
\end{equation}
where 
\begin{equation}
C_i \cap C_j = \varnothing, 1 \le i \le j \le K(g(x))
\end{equation}
Then for any cutting function $g(x)$, $K$ is the uniquely functional with respect to $g(x)$ which can be defined as the number of connected components of the set $\mathcal S$.

\noindent \textbf{Definition 2: Optimal Modal Number $K^*$}

The best number of modal should satisfy the following conditions:

(a) Each component $C_k$ of $\mathcal S(g(x))$ should support sufficient energy of the spectrum, that is the residual $\int_{C_k} f(x) - g(x) \text dx$ should be sufficiently significant, with respect to the whole trend of the spectrum.

(b) Between distinct modal components, the cutting curve should remain sufficiently smooth, avoiding spurious oscillations induced by noise or minor fluctuations.

(c) Adjacent weak peaks lacking individual saliency should not constitute separate modes. They are to be either aggregated into a combined mode or, where a dominant peak is present, assimilated as subsidiary ripples of that principal mode.

As for assumption (a), one should find a curve that approaches the infinitum $f(x)$, yet maximizing the residual $\int f(x)-g(x) \text dx$, which in turn encourage the evaluated curve to isolate different modes.  However, purely encouraging the curve to approach $f(x)$ everywhere will result in nothing to be retrieved, so that the curve should remain at the bottom of $f(x)$ at the interval where peaks exist.  By considering (b) and (c), the penalty to the curve is a good choice to satisfy the assumption.  Thus by encouraging the integral of $g(x)$ and discouraging the curvature of $g(x)$, evaluating the optimal $g^*(x)$ is implicitly evaluating the best mode number $K$ under our basic assumption.

However, directly optimization to $K(g(x))$ is difficult since this is not a continuous functional with respect to $g(x)$, so in the following sections we instead optimize $g(x)$ to approach the target.

\begin{figure}[!hbtp]
    \centering
    \includegraphics[width=\columnwidth]{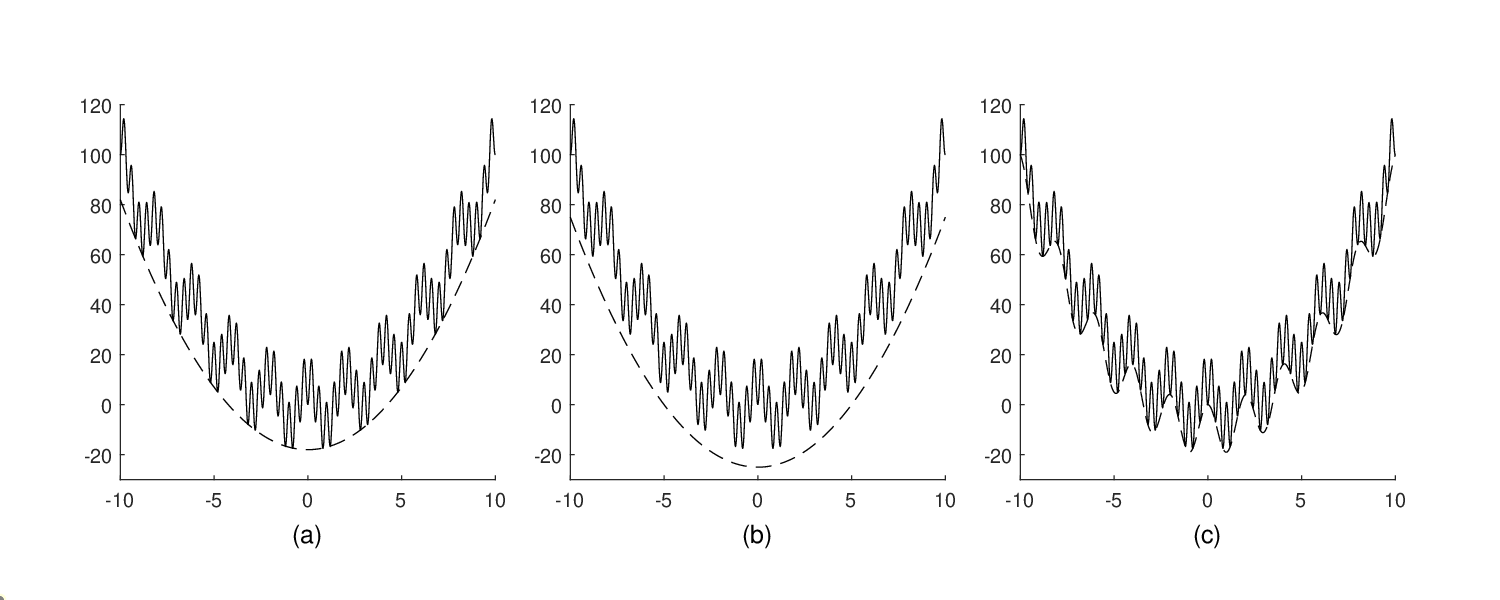}
    \caption{Illustration to the concept of best Cutting Curve: In (a)-(c) the solid line is $y=10\cos\pi x -10\cos 5\pi x + x^2$. The dashed line in (a) is a proper Cutting Curve since it actually fits the lower bound of the solid line as well as ignoring high frequent details of the solid line. The dashed line in (b) has a gap to the bottom of the solid line so it is not suitable to be a good cutting curve. In contrast, the dashed line in (c) fits too tightly to the lower bound of the solid line, capturing too much detail of the solid line so that it can be reckoned as the lower envelope rather than a proper Cutting Curve.}
\label{diag1}
\end{figure}

\subsection{Best Cutting Curve of a Function}
Without loss of generality we first introduce the Cutting Curve of a function.  For any function $f(x) > 0$, the best Cutting Curve is a curve that tightly fits the lower bound of the function while reflects the overall trend of the lower bound.  Fig.\ref{diag1} gives an intuitive illustration about the conception of Cutting Curve.  As illustrated in Fig.\ref{diag1}, intuitively the Cutting Curve should be tightly close to the lower bound of the function, approaching the infimum of functions as close as possible as well as maintain its own smoothness in order to capture the global trend of the lower bound rather than too much details, as capturing too much details will in turn reduce the robustness for separating the modes since there might be some noise added to the signal in either time or frequency domain.

\subsection{Variational Methods for Finding the Best Cutting Curve}

To solve the problem for finding the best Cutting Curve, we use optimization method mainly based on variation and convex optimization\cite{Dazhong,Boyd}.  Let $f(x) \ge 0$ be a known function, the problem of finding the Cutting Curve curve of $f(x)$ is equivalent to finding a function $g(x)$ that tightly approaching the point-wise infimum of $f(x)$ while maintaining the global trend of its lower bound.  This can be viewed as the following optimization problem:
\begin{equation} 
    \left\{
  \begin{aligned}
  \underset{g(x)}{\max}\int_{\Omega}g(x)\text dx\\
  \underset{g(x)}{\min}\int_{\Omega}g''^2(x)\text dx \\
    \end{aligned}
    \right.\text{w.r.t}
  \begin{cases}
  g(x)\le f(x) \\
  g(x)\ge 0 \\
    \end{cases}
  \label{opt}
\end{equation}
Converting (\ref{opt}) into standard form, we get:
\begin{equation} 
    \left\{
  \begin{aligned}
  \underset{g(x)}{\min}\int_{\Omega}-g(x)\text dx\\
  \underset{g(x)}{\min}\int_{\Omega}g''^2(x)\text dx \\
    \end{aligned}
    \right.\text{w.r.t}
  \begin{cases}
  g(x) - f(x) \le 0\\
  - g(x)\le 0 \\
    \end{cases}
  \label{std}
\end{equation}
This is an optimization problem with inequality constraints.  To solve the problem, one can construct an Lagrange Slack Function by adding functional-styled weights to each factor and constraints
\begin{align}
L(g,\lambda,\mu,\alpha,\beta) &= \int_{\Omega}\alpha(x) g''^2(x)\text{d}x + \int_{\Omega}\beta(x)[-g(x)]\text{d}x \notag \\ &+ \int_\Omega \lambda(x)[g(x)-f(x)]\text{d}x + \int_\Omega \mu(x)[-g(x)]\text{d}x
\label{proposal}
\end{align}
where
\begin{equation}
\begin{cases}
  \lambda(x)=0, f(x)\ge g(x), \forall x \\
  \lambda(x)>0, g(x) > f(x), \forall x \\
  \end{cases}
\end{equation}
\begin{equation}
\begin{cases}
   \mu(x)=0, g(x)\ge 0, \forall x \\
   \mu(x)>0, g(x) < 0, \forall x \\
  \end{cases}
\end{equation}
and $\alpha(x), \beta(x) > 0$.
The solution of the problem can be found by solving a variation problem, namely, if we set
\begin{equation}
F=\alpha(x) g''^2(x)-\beta(x)g(x) +\lambda(x)[g(x)-f(x)]-\mu(x)g(x)
\end{equation}
then the above Lagrange Slack Function is like the form $\displaystyle J[g(x)]=\int_{\Omega} F(x,g,g'')\text dx$ so that if we assume the variation is zero at both ends of the curve, one can easily know the optimal solution of the above function happen to satisfy the simplified form of Euler-Poisson equation if we let the variation vanish,\cite{Dazhong} so it can be written as
\begin{equation}
F_g+\cfrac {\text d^2F_{g''}}{\text dx^2}=0
\label{EP}
\end{equation}
where 
\begin{equation}
F_g = -\beta(x)+\lambda(x)-\mu(x)
\label{Fg}
\end{equation}
\begin{equation}
F_{g''} = 2\alpha(x)g''(x)
\end{equation}
It is easily to get that
\begin{equation}
\cfrac {\text d^2 F_{g''}} {\text dx^2}  = 2\alpha''(x)g''(x)+4\alpha'(x)g^{(3)}(x)+2\alpha(x)g^{(4)}(x)
\label{F2g}
\end{equation}
By substituting (\ref{Fg}) and (\ref{F2g}) into (\ref{EP}), we can get the following ordinary differential equation
\begin{equation}
2\alpha''(x)g''(x)+4\alpha'(x)g^{(3)}(x)+2\alpha(x)g^{(4)}(x) =\beta(x)-\lambda(x)+\mu(x)
\label{final}
\end{equation}
This is a 4th-ordered Ordinary Differential Equation (ODE), so that normally four constraints are needed to determine a definite solution.  Recall that we assume the variation is zero at both ends of the curve, so that we can typically set four constraints at both ends of the curve with respect to the function value and its first-order derivative as follows:
\begin{equation}
g(0)=f(0),g(n)=f(n),g'(0)=f'(0),g'(n)=f'(n)
\label{con}
\end{equation} where $0$ stands for the left end and $n$ the right end of the curve.
To solve the boundary value problem of ODE, we use Finite Difference Method\cite{Dehao}.  Note that for a differential function $g(x)$, it's derivative of first to fourth order can be expressed as the linear combination of $g(x)$ itself in the corresponding differential format and we can take similar operation to the function $\alpha(x)$ (See Supplementary Materials \ref{Matrices}), so that different orders of derivative with respect to $\alpha(x)$ and $g(x)$ can be expressed as the linear combination of themselves, one can easily transform (\ref{final}) into the matrix-form as follows:
\begin{align}
(2\boldsymbol{A^{(2)}\alpha}) \odot (\boldsymbol{G^{(2)}g}) + (4\boldsymbol{A^{(1)}\alpha}) \odot (\boldsymbol{G^{(3)}g}) +(2\boldsymbol{\alpha}) \odot (\boldsymbol {G^{(4)}g}) = \boldsymbol{\beta - \lambda + \mu}
\label{ode_equation}
\end{align} 
where $\boldsymbol{A}^{(n)}$ and $\boldsymbol{G}^{(n)}$ is the conversion matrices (See Supplementary Materials \ref{Matrices}) that transform the n$^\text{th}$ derivatives of $\alpha(x)$ and $g(x)$ into linear expression of the corresponding function, $\odot$ stands for component-wise multiplication between matrices.  The conversion matrices $\boldsymbol{A}^{(n)}$ and $\boldsymbol{G}^{(n)}$ has the same structure.  $\boldsymbol{g}$, $\boldsymbol{\alpha}$, $\boldsymbol{\beta}$, $\boldsymbol{\lambda}$, $\boldsymbol{\mu}$ are column vectors obtained by respectively discretizing the corresponding functions, in which each element is a uniformly discrete sampled value on the function $g(x)$, $\alpha(x)$, $\beta(x)$, $\lambda(x)$, $\mu(x)$, respectively.  The boundary condition is handled at the last 4 line in $(2\pmb A^{(2)}\pmb \alpha) \odot \pmb G^{(2)}+(4\pmb A^{(1)}\pmb \alpha)\odot \pmb G^{(3)} +2\pmb \alpha \odot \pmb G^{(4)}$ (See Supplementary \ref{Matrices} for details).  Algorithm \ref{alg:algorithm} shows the details of our key steps for evaluating the Cutting Curve.

\begin{algorithm}[H]
\caption{The Key Steps for Evaluating Optimal Cutting Curve}
\label{alg:algorithm}
\begin{algorithmic}[1]
\STATE Input $\{f(x)\}$
\STATE Initialize $\alpha(x), \beta(x), \lambda(x), \mu(x), g(x)$
\STATE Pre-process by extending $f(x)$ at both ends.
\REPEAT
    \STATE $n \leftarrow n + 1$
    \STATE Update $g(x)$:
    \STATE \quad Solve the equation:
    \STATE \quad $(2\boldsymbol{A^{(2)}\alpha}) \odot (\boldsymbol{G^{(2)}g}) + (4\boldsymbol{A^{(1)}\alpha}) \odot (\boldsymbol{G^{(3)}g}) +(2\boldsymbol{\alpha}) \odot (\boldsymbol {G^{(4)}g}) = \boldsymbol{\beta - \lambda + \mu}$
    \STATE Dual ascent:
    \STATE \quad $\lambda^{n+1} \leftarrow \max\{0, \lambda^n + \theta (g - f)\}$
    \STATE \quad $\mu^{n+1} \leftarrow \max\{0, \mu^n + \gamma (-g)\}$
\UNTIL{convergence: $\|g^{n} - g^{n-1}\|_2^2 / \|g^{n}\|_2^2 < \epsilon$}
\end{algorithmic}
\end{algorithm}

\subsection{Proof of the Convergence}
One of the core contribution of this work is to provide a deterministic, global convergence guarantees for the problem of automatic mode detection, which is typically unaddressed in the pioneer works.  we prove that the algorithm generates sequences converging to an optimal solution of the original problem. The proof is built upon several folds:

\textbf{Convexity and Duality}: Proof of the problem's convexity and strong duality.

\textbf{Well-posedness of Iterations}: Demonstration that each iterative step solves a well-posed subproblem with a unique solution.

\textbf{Convergence in Dual Space}: Impose functional functional to dual ascent method and analyse its convergence under both unique and non-unique optimal dual solutions.

\textbf{Convergence in Primal Space}: Establish the convergence of the primal sequence via the continuous dependence of solutions on parameters.

For the conciseness of the proof, we leave the detail of the the convexity and duality into Appendix and focus on the well-posedness of iterations and convergence in dual and primal spaces here.

\subsubsection{Convexity and Duality}
In this part we state the main results about the convexity and duality, and we put the details in Appendix \ref{Proof_of_convexity} and \ref{Proof_of_dual}.

\begin{lemma}
The objection functional of (\ref{std}) is convex, and the feasible set is also convex.
\end{lemma}

\begin{lemma}
The dual of objection functional of (\ref{std}) is concave, with the dual-gap zero.
\end{lemma}

\subsubsection{Well-posedness of Iterations}
Before analyzing the convergence of the iterative algorithm, we must ensure that the subproblem solved at each iteration is well-posed.  Specifically, as in each iteration, our primal solution $g_{k+1}(x)$ is given by dual variables $(\lambda_k(x),\mu_k(x))$ at the k-th step and we transform the variational problem into a linear fourth-order boundary value problem (BVP), the well-posedness of this BVP encompassing the existence and uniqueness of the solution is crucial to guarantee that our iterative scheme is not only well-defined but also stable against small numerical perturbations. Since $\alpha(x)$ and $\beta(x)$ in our algorithm can be considered as weights function that are always positive, we can establish the existence and uniqueness of our solution to (\ref{final}) below.

\begin{theorem}
Assume $\alpha(x) > 0$, then the solution of (\ref{final}) with 
$g(0)=0,\\g(n)=0,g'(0)=0,g'(n)=0$ is $g(x) \equiv 0$.
\label{exist_unique}
\end{theorem}
\begin{proof}
Recall that when dealing with the homogeneous condition, we exactly have
\begin{equation}
g(x)=\int_a^x\left (\int_a^x \cfrac {C_1t+C_2}{\alpha(t)} \text dt \right ) \text dt +C_3x+C_4
\end{equation}
where
\begin{equation}
g(a)=g(b)=g'(a)=g'(b)=0
\end{equation}
It is obvious that $C_3$ and $C_4$ are zero by $g(a)$ and $g'(a)$, and $g'(b)=0$ is equivalent to
\begin{equation}
\int_a^b\cfrac {C_1t+C_2}{\alpha(t)}\text dt = 0
\label{linear}
\end{equation}
Since $\alpha(t) > 0$ and $\alpha(t) \in C[a,b]$, since $C_1t + C_2$ is a linear function so that if the zero point is not in $[a,b]$ then it is not possible to make (\ref{linear}) hold.  If there exists $t \in [a,b]$ that makes $C_1t+C_2$ zero, then the upper limit integral function $\Phi(x)=\displaystyle\int_a^x\cfrac {C_1t+C_2}{\alpha(t)}\text dt$ has no zero points in $(a,b)$, thus makes $g(b)=\displaystyle\int_a^b\Phi(t)\text dt$ non-zero, which is paradoxical to our assumption if not both $C_1$ and $C_2$ are zero.  Then according to Fredholm Alternative Theorem, we can assert that during each iteration, there always exist a unique solution to $g(x)$\cite{Collins2006,kress2013linear,keener2019principles}.
\end{proof}

\begin{corollary}
For each iteration of (\ref{ode_equation}), there always exists a unique solution.
\end{corollary}

\subsubsection{Convergence in Dual Space}

To facilitate the convergence analysis in the dual space, we introduce the following regularity assumptions on the primal problem:

The target spectrum function $f(x)$ is Lipschitz continuous on $\mathcal D$.

The solution $g(x)$ to the subproblem and the iterates $g_k(x)$ reside in $C^2(\mathcal D)$ with uniformly bounded second derivatives.

These assumptions are standard for analyzing variational problems involving curvature penalties. In practice, for discrete and potentially noisy spectral data, $f(x)$ can be considered as a smooth interpolation of the observed points. Crucially, the core object of interest—the Cutting Curve—is intrinsically a smooth function due to the curvature penalty term in the objective. Therefore, the analysis under these smoothness assumptions captures the essential behavior of the algorithm.

We begin the proof by first restating the iteration in Algorithm \ref{alg:algorithm}.  Notice that during the iteration, we have

\begin{align}
  \left\{ \begin{alignedat}{2} 
    g_{k+1}(x)&= \underset{g(x)}{\arg\min}\mathcal L(g(x), \lambda_k(x), \mu_k(x)) \\
    \lambda_{k+1}(x) &= \max[0,\lambda_k(x) + \theta(g_{k+1}(x)-f(x))] \\
    \mu_{k+1}(x) &= \max[0, \mu_k(x)+\gamma(-g_{k+1}(x))]
  \end{alignedat} \right.
\label{iteration}
\end{align}

To show the projection property for the $\max$ operation we first prove the following lemma:
\begin{lemma}
Let $m(x), x \in \mathcal D$ be any continuous function that has both positive and negative values, Let $\Pi$ be an operator that makes all the negative values for $m(x)$ zero and non-negative unchanged, or $\Pi m(x)=\max(0, m(x))$, let $M(x), x \in \mathcal D$ be any non-negative function, then 
\begin{equation}
\int_{\mathcal{D}} \bigl[ m(x) - \Pi m(x) \bigr] \bigl[ M(x) - \Pi m(x) \bigr] dx \le 0
\label{mx}
\end{equation}
\label{lemma_mx}
\end{lemma}
\begin{proof}
By expanding the equation (\ref{mx}) we should have the following:
\begin{equation}
\begin{split}
&\ \ \ \ \int_{\mathcal{D}} \bigl[ m(x) - \Pi m(x) \bigr] \bigl[ M(x) - \Pi m(x) \bigr] dx \\
&= \int_{\mathcal{D}} m(x)M(x) dx 
   - \int_{\mathcal{D}} \Pi m(x)M(x) dx - \int_{\mathcal{D}} m(x)\Pi m(x) dx 
   + \int_{\mathcal{D}} \bigl[ \Pi m(x) \bigr]^2 dx \\
&= \int_{\mathcal{D}} m(x)M(x) dx 
   - \int_{\mathcal{D}} \Pi m(x)M(x) dx  - \int_{\mathcal{D}} \bigl[ \Pi m(x) \bigr]^2 dx 
   + \int_{\mathcal{D}} \bigl[ \Pi m(x) \bigr]^2 dx \\
&\le \int_{\mathcal{D}} m(x)M(x) dx 
   - \int_{\mathcal{D}} m(x)M(x) dx - \int_{\mathcal{D}} \bigl[ \Pi m(x) \bigr]^2 dx 
   + \int_{\mathcal{D}} \bigl[ \Pi m(x) \bigr]^2 dx = 0
\end{split}
\label{int}
\end{equation}
Moreover, if $m(x)$ is in feasible set so that $m(x)=\Pi m(x)$, then (\ref{int}) will become zero.
\end{proof}
Now we prove that in iteration (\ref{iteration}) the multiplier $\lambda_k(x)$, $\mu_k(x)$ always terminate at one of the optimum.  To say in detail, we should have the following theorem:
\begin{theorem}
Let $\mathcal G(\lambda(x), \mu(x)) = \underset{g(x)}{\arg\min}\mathcal L(g(x), \lambda(x), \mu(x))$ and $\lambda_k(x)$, $\mu_k(x)$ be equicontinuous for all $k=1,2,3,\cdots$, then if $\mathcal G$ satisfies 
\begin{align}
||D[\mathcal{G}(\lambda_1(x), \mu_1(x))&-\mathcal{G}(\lambda_2(x), \mu_2(x))]||\\ \nonumber &\le L||[\lambda_1(x)-\lambda_2(x), \mu_1(x)-\mu_2(x)]||
\label{grad}
\end{align}
for any $\lambda_1(x)$, $\lambda_2(x)$, $\mu_1(x)$, $\mu_2(x)$, where $D$ stands for the variation, and the step size $\eta < \frac 1 L$, then as $k \rightarrow +\infty$, the iteration series of $(\lambda_k(x), \mu_k(x))$ must terminate at the point as close as possible to one of the optimum $(\lambda^*(x), \mu^*(x))$ for $\mathcal G$.
\label{lambda_mu}
\end{theorem}
\begin{proof}
For conciseness, simplifying the integral by introducing matrix-form expression, where
\begin{align}
\left [ \begin{array}{cc}
a(x) & b(x)
\end{array} \right ]
\left [ \begin{array}{cc}
c(x) \\
d(x)
\end{array} \right ]
= \int_{\mathcal D} a(x)c(x) \text dx + \int_{\mathcal D} b(x)d(x) \text dx
\end{align}
 we have
\begin{align}
\left [ \begin{array}{cc}
\tilde\lambda_{k+1}(x) - \lambda_{k+1}(x) & \tilde\mu_{k+1}(x) - \mu_{k+1}(x)
\end{array} \right ]
\left [ \begin{array}{cc}
\lambda_{k}(x) - \lambda_{k+1}(x) \\
\mu_{k}(x) - \mu_{k+1}(x)
\end{array} \right ] \le 0
\label{judge}
\end{align}
by considering all possible locations for $(\tilde\lambda_{k+1}(x),\tilde\mu_{k+1}(x))$, where $(\tilde\lambda, \tilde\mu)$ stands for general points before assigning the operator $\Pi$, or $\max(0,\cdot)$ compared with $(\lambda, \mu)$.  Expand (\ref{judge}) by setting 
\begin{equation}
\begin{bmatrix}
\tilde\lambda_{k+1}(x) \\
\tilde\mu_{k+1}(x)
\end{bmatrix}=
\begin{bmatrix}
\tilde\lambda_{k}(x) \\
\tilde\mu_{k}(x)
\end{bmatrix}+\eta D\mathcal G(\lambda_k(x), \mu_k(x))
\end{equation}
as iteration step and rearrange, we have
\begin{align}
D\mathcal G_{(\lambda_{k}(x),\mu_{k}(x))}(\lambda_{k}(x),\mu_{k}(x))
\left [ \begin{array}{cc}
\lambda_{k+1}(x) - \lambda_{k}(x) \\
\mu_{k+1}(x) - \mu_{k}(x)
\end{array} \right ]\ge
\cfrac 1 \eta 
\left |\left| \left [ \begin{array}{cc}
\lambda_{k+1}(x) - \lambda_{k}(x) \\
\mu_{k+1}(x) - \mu_{k}(x)
\end{array} \right ]\right |\right|_2^2
\label{to_combine}
\end{align}
On the other side, for any functional we have
\begin{align}
&\mathcal{G}(\lambda_{k+1}(x),\mu_{k+1}(x)) 
= \mathcal{G}(\lambda_{k}(x),\mu_{k}(x)) \notag \\
&+ D G_{(\lambda_k(x), \mu_k(x))}(\lambda(x), \mu(x))^\top
  \begin{bmatrix}
    \lambda_{k+1}(x)-\lambda_k(x) \\
    \mu_{k+1}(x)-\mu_k(x)
  \end{bmatrix} \notag \\
&+ \frac{1}{2!}
  \begin{bmatrix}
    \lambda_{k+1}(x)-\lambda_k(x) & \mu_{k+1}(x)-\mu_k(x)
  \end{bmatrix} \notag \\
&\quad D^2 G_{(\lambda(x), \mu(x))}(\varepsilon_k(x), \zeta_k(x))
  \begin{bmatrix}
    \lambda_{k+1}(x)-\lambda_k(x) \\
    \mu_{k+1}(x)-\mu_k(x)
  \end{bmatrix}
\label{second_order}
\end{align}
where $(\varepsilon_k(x), \zeta_k(x))$ locates on general line segment from $(\lambda_{k}(x),\mu_{k}(x))$ to 

\noindent$(\lambda_{k+1}(x),\mu_{k+1}(x))$.  If $\mathcal G$ is second-order differentiable and satisfies (\ref{grad}) so that $||D^2\mathcal G||$ is bounded with $[-L,L]$.  Due to the concavity of $\mathcal G$ and the left side of the bound, and combine (\ref{to_combine}) we get
\begin{align}
\mathcal G(\lambda_{k+1}(x),\mu_{k+1}(x))-\mathcal G(\lambda_{k}(x),\mu_{k}(x)) \ge \left (\cfrac 1 \eta - \cfrac L 2 \right )
\left |\left| \left [ \begin{array}{cc}
\lambda_{k+1}(x) - \lambda_{k}(x) \\
\mu_{k+1}(x) - \mu_{k}(x)
\end{array} \right ]\right |\right|_2^2 \ge 0
\label{monotonous}
\end{align}
if $\eta=\max\{\theta, \gamma\} =\frac 1 L < \frac 2 L$, then $\mathcal G$ is a monotonously increasing functional series with upper bound due to its concavity, so that $\mathcal G$ converges and the difference between adjacent term of both $\mathcal G$ and $(\lambda_{k}(x),\mu_{k}(x))$ is approaching zero.

Now we prove $\{g_k(x)\}$ is guaranteed to approach one of the optimum when the algorithm terminates even if $\mathcal G$ is not strongly concave.   By introducing $l(\pmb y;\pmb x)=\mathcal G(\pmb x)+D\mathcal G(\pmb x)(\pmb y-\pmb x)$ and due to the concavity the iteration on $\mathcal G$ should satisfy \cite{BertsekasCOA}
\begin{align}
&l\left( 
\begin{bmatrix} \lambda_{k+1}(x) \\ \mu_{k+1}(x) \end{bmatrix}; 
\begin{bmatrix} \lambda_{k}(x) \\ \mu_{k}(x) \end{bmatrix} 
\right) - \cfrac 1 {2\eta}\left\| 
\begin{bmatrix}
\lambda_{k+1}(x) - \lambda_{k}(x) \\
\mu_{k+1}(x) - \mu_{k}(x)
\end{bmatrix} 
\right\|_2^2 \nonumber \\
&\ge  l\left( 
\begin{bmatrix} \lambda^*(x) \\ \mu^*(x) \end{bmatrix}; 
\begin{bmatrix} \lambda_{k}(x) \\ \mu_{k}(x) \end{bmatrix} 
\right) - \cfrac 1 {2\eta} \left\| 
\begin{bmatrix}
\lambda^*(x) - \lambda_k(x) \\
\mu^*(x) - \mu_k(x)
\end{bmatrix} 
\right\|_2^2 \nonumber \\
&+ \cfrac 1 {2\eta}\left\| 
\begin{bmatrix}
\lambda^*(x) - \lambda_{k+1}(x) \\
\mu^*(x) - \mu_{k+1}(x)
\end{bmatrix} 
\right\|_2^2 
\label{expand_inequality}
\end{align}
since $-\mathcal G$ is convex. Due to the concavity of $\mathcal G$, we have 
\begin{align}
\mathcal G(\lambda_{k+1}(x), \mu_{k+1}(x)) \ge l \left (
\begin{bmatrix} \lambda_{k+1}(x) \\ \mu_{k+1}(x) \end{bmatrix}; 
\begin{bmatrix} \lambda_{k}(x) \\ \mu_{k}(x) \end{bmatrix}
\right ) - \cfrac 1 {2\eta}\left\| 
\begin{bmatrix}
\lambda_{k+1}(x) - \lambda_{k}(x) \\
\mu_{k+1}(x) - \mu_{k}(x)
\end{bmatrix} 
\right\|_2^2
\label{neq_1}
\end{align}
and
\begin{align}
l\left( 
\begin{bmatrix} \lambda^*(x) \\ \mu^*(x) \end{bmatrix}; 
\begin{bmatrix} \lambda_{k}(x) \\ \mu_{k}(x) \end{bmatrix} 
\right) \ge \mathcal G(\lambda^*(x), \mu^*(x))
\label{neq_2}
\end{align}

Substituting (\ref{neq_1}) and (\ref{neq_2}) into (\ref{expand_inequality}), then we have
\begin{align}
\left\| 
\begin{bmatrix}
\lambda^*(x) - \lambda_{k+1}(x) \\
\mu^*(x) - \mu_{k+1}(x)
\end{bmatrix} 
\right\|_2^2 &\le
2\eta[\mathcal G(\lambda_{k+1}(x), \mu_{k+1}(x)) \notag \\ 
&- \mathcal G(\lambda^*(x), \mu^*(x))] + \left\| 
\begin{bmatrix}
\lambda^*(x) - \lambda_k(x) \\
\mu^*(x) - \mu_k(x)
\end{bmatrix} 
\right\|_2^2
\label{to_sum}
\end{align}

Sum up (\ref{to_sum}) for each $k$ and eliminate the repeated terms that occur at the both ends of the inequality, we have
\begin{align}
\left\| 
\begin{bmatrix}
\lambda^*(x) - \lambda_{k+1}(x) \\
\mu^*(x) - \mu_{k+1}(x)
\end{bmatrix} 
\right\|_2^2 &\le
2\eta \sum\limits_{i=0}^k[\mathcal G(\lambda_{i+1}(x), \mu_{i+1}(x)) \notag \\
&- \mathcal G(\lambda^*(x), \mu^*(x))] + \left\| 
\begin{bmatrix}
\lambda^*(x) - \lambda_0(x) \\
\mu^*(x) - \mu_0(x)
\end{bmatrix} 
\right\|_2^2
\end{align}
As the left side is non-negative, and recall in (\ref{monotonous}) we know that $\{\mathcal G_k\}$ is non-decreasing and $\mathcal G^*$ is the supreme of $\mathcal G$, hence $\mathcal G(\lambda_{k+1}(x), \mu_{k+1}(x)) - \mathcal G(\lambda^*(x), \mu^*(x))$ is the largest among all the terms involved in the summation, by applying reasonable inequality scaling we must have
\begin{align}
0 \ge \mathcal G(\lambda_{k+1}(x)&, \mu_{k+1}(x)) - \mathcal G(\lambda^*(x), \mu^*(x)) \ge
-\cfrac 1 {2\eta(k+1)}\left\| 
\begin{bmatrix}
\lambda^*(x) - \lambda_0(x) \\
\mu^*(x) - \mu_0(x)
\end{bmatrix}
\right\|_2^2
\label{constrants}
\end{align} 
which shows the iteration is approaching the optimum.  From (\ref{to_sum}) and by utilizing the monotonicity of the functional $\mathcal G$ during iteration we have
\begin{equation}
\left\| 
\begin{bmatrix}
\lambda^*(x) - \lambda_{k+1}(x) \\
\mu^*(x) - \mu_{k+1}(x)
\end{bmatrix} 
\right\|_2^2 \le \left\| 
\begin{bmatrix}
\lambda^*(x) - \lambda_k(x) \\
\mu^*(x) - \mu_k(x)
\end{bmatrix} 
\right\|_2^2
\label{decreasing}
\end{equation} which shows the distance between the iterating points and the optimum is always decreasing, since the distance is always non-negative so that it has lower bound, $(\lambda_k(x), \mu_k(x))-(\lambda^*(x), \mu^*(x))$ is a convergent sequence so that $(\lambda_k(x), \mu_k(x))$ is bounded.  According to Arzelà-Ascoli Theorem, we can only assert that there exists some sub-sequences $(\lambda_{k_i}(x),\mu_{k_i}(x))$ that converges, which does not indicate that the whole iterating sequence is necessarily convergent in all cases and needs careful discussion.  There are four cases when the iterating sequence is approaching the optimum: 

CASE I: If there is only one optima in $\mathcal G$, then it is trivial to show that all these sub-sequences converges to the unique optimum by the constraints in (\ref{constrants}).

CASE II: If the optimum of $\mathcal G$ forms a convex set, and during the iteration the sequence $(\lambda_k(x), \mu_k(x))$ is approaching and never fall into the optimal set, then from the middle parts of (\ref{constrants}) and (\ref{monotonous}) we know that the value of $\mathcal G$ is approaching the maxima and the distance between adjacent term of the sequence $(\lambda_k(x), \mu_k(x))$ and $(\lambda_{k+1}(x), \mu_{k+1}(x))$ is yielding zero.  Thus we know in this case the iterating sequence in dual space can always terminate to the point as close as possible to one of the optimum, and the iterating sequence is of asymptotic regularity.

CASE III: If the iterating sequence has approached to the unique optima during some steps, then due to (\ref{decreasing}) the next iteration will still stay at the unique optima.

CASE IV: If the optimum of $\mathcal G$ forms a convex set, and the iterating sequence has approached to one of the optimum $(\lambda^*_1(x), \mu^*_1(x))$ early during some iterating step, then in this case the left side of (\ref{monotonous}) must be zero since under reasonable step size the value of functional $\mathcal G$ cannot decrease, so that the middle part of (\ref{monotonous}) is zero, indicating the iterating sequences will stay still.

In all cases we show $(\lambda_k(x), \mu_k(x))$ can always terminate at the point as close as possible to at least one of the optimum and stay still as long as it arrives the optimum during the iteration even if there are more than one optimum in $\mathcal G$.
\end{proof}

\subsubsection{Convergence in Primal Space}
\begin{theorem}
Under the convergence of dual sequence as stated in Theorem \ref{lambda_mu}, the primal sequence $\{g_k(x)\}$ generated by Algorithm \ref{alg:algorithm} converges to the optimal solution $g^*(x)$, specifically, we have $\lim\limits_{k\rightarrow+\infty}||g_k(x)-g^*(x)||=0$.
\end{theorem}
\begin{proof}
As the solution of (\ref{std}) is continuously dependent on its coefficients, so we have $\displaystyle\lim_{k\rightarrow +\infty}g_k(x)=g^*(x)$ when $\{(\lambda_k(x), \mu_k(x))\}$ when in Case I, III and IV, and $||g_{k+1}(x)-g_k(x)||$ is approaching zero in Case II respectively.  In all cases we can terminate algorithm and guarantee the $\{g_k(x)\}$ terminates to one of the optimum as close as possible.
\end{proof}

\section{Implementation}
In this chapter, we focus on the detail about the implementation and numerical scheme about the algorithm for evaluating problem (\ref{std}).  Our workflow consists of 4 interconnected stages: (1) We discretize (\ref{final}) by using Finite Difference Method\cite{Dehao}, transforming it into a linear system.  (2) Crucially, by expressing the derivations of both solution $g(x)$ and the weight function $\alpha(x)$ as linear combination of their nodal values, we recast the entire ODE into a compact matrix form by introducing component-wise multiplication with a mild extension of Hadamard production.  By imposing well-posed boundary conditions, we make the coefficient matrices full-rank and the whole problem invertible so that the evaluation can be significantly accelerated by modern software toolkits via matrix operation.  (3) To further accelerate the convergence and make algorithm effective when handling high or low-passed signals, a smooth extrapolation technique is applied to the input spectrum at the both ends.  (4) Finally, after obtaining the optimal cutting curve, a statistically grounded threshold is determined via kernel density estimation (KDE) on the residual spectrum to filter out background noise. The resulting significant peaks are then aggregated using DBSCAN to merge spatially adjacent minor peaks into a single coherent intrinsic mode, thereby completing the mode extraction.  The following subsections detail each of these components.

\subsection{Spatial Discretization  Scheme}
We discretize the solution domain $\Omega=[a,b]$ into $N+1$ equally spaced nodes $x_0, x_1, \dots, x_n$, where $x_0=a$ and $x_n=b$, with the grid spacing $h=\dfrac {b-a} N$, thus the value of any continuous function $p(x)$ at the nodes are recorded as a column vector $\pmb p=[p(x_0),p(x_1),\dots,p(x_N)]^\top \in \mathbb R^{N+1}$.  Then we employ Five-Point Central Difference Scheme to transform different orders of derivatives into the linear combination with respect to the primal nodal values.  For instance, for the fourth-order derivative $g^{(4)}(x)$ at interior nodes $x_i, i=2,3,\dots,N-2$ we have
\begin{equation}
g^{(4)}(x_i)\approx \cfrac{g(x_{i+2})-4g(x_{i+1})+6g(x_i)-4g(x_{i-1})+g(x_{i-2})}{h^4}
\label{m_front}
\end{equation}
and the corresponding boundary conditions can be approximately expressed as
\begin{align}
  \left\{ \begin{alignedat}{2} 
    g(0)&=g(x_0) \\
    g(n)&=g(x_N) \\
    g'(0)&\approx\cfrac {g(x_1)-g(x_0)} h \\
    g'(n)&\approx\cfrac {g(x_N)-g(x_{N-1})} h
  \end{alignedat} \right.
\label{m_rear}
\end{align}

Note that (\ref{m_front}) and (\ref{m_rear}) can be transformed into compact matrix form via a pentadiagonal matrix $\pmb G^{(4)} \in \mathbb R^{(N+1)\times(N+1)}$, thus we have
\begin{equation}
\boldsymbol{G}^{(4)}=\begin{bmatrix}
1 & -4 & 6 & -4 & 1 & 0 & \cdots & \cdots & \cdots & 0 \\
0 & 1 & -4 & 6 & -4 & 1  & 0 & \cdots & \cdots & 0 \\
\vdots &\vdots &\vdots &\vdots &\vdots &\vdots &\vdots &\vdots &\vdots &\vdots \\
0 & \cdots & \cdots & 0 & 1 & -4 & 6 & -4 & 1  & 0 \\
0 & 0 & \cdots & \cdots & 0 & 1 & -4 & 6 & -4 & 1 
\end{bmatrix}
\end{equation}
thus $\boldsymbol{G}^{(4)}\boldsymbol{g}$ is the fourth-ordered derivative of $\pmb g$ at each node and the non-zero entry stands for the corresponding coefficients of difference and boundary conditions.  $\boldsymbol{G}^{(3)}$, $\boldsymbol{G}^{(2)}$, $\boldsymbol{G}^{(1)}$ and $\boldsymbol{A}^{(4)}$ to $\boldsymbol{A}^{(1)}$ have the similar structure as $\boldsymbol{G}^{(4)}$, we give the detail of each pentadiagonal matrix in \ref{Matrices}.

By integrating both the difference strategy and the boundary condition (\ref{final}) is transformed into (\ref{ode_equation}), where $\odot$ stands for component-wise multiplication.  However, in (\ref{ode_equation}) $\pmb g$ is inseparable so that we cannot resort to effective evaluation.  To solve this, our $\odot$ is just an extension of standard Hadamard production by adding some compatible broadcasting rules in dealing with column vector, that is $\forall \pmb A_{m\times n}, \pmb v_{m\times 1}$, we have
\begin{equation}
\pmb v \odot \pmb A = \pmb A \odot \pmb v = \pmb A * \pmb V_{m\times n}
\end{equation}
where * stands for standard Hadamard production and each column of $\pmb V$ is just a copy of $\pmb v$.  By applying this rule we have $(\pmb A\pmb a) \odot (\pmb B \pmb b) = ((\pmb A \pmb a)\odot \pmb B) \pmb b$, we can transform (\ref{ode_equation}) into 
\begin{equation}
((2\boldsymbol{A^{(2)}\alpha}) \odot \boldsymbol{G^{(2)}})\pmb g + ((4\boldsymbol{A^{(1)}\alpha}) \odot \boldsymbol{G^{(3)}})\pmb g +((2\boldsymbol{\alpha}) \odot \boldsymbol {G^{(4)}})\pmb g = \pmb \beta - \pmb \lambda + \pmb \mu
\end{equation}
so that we can treat (\ref{ode_equation}) as common matrix equation and get the solution by simply using matrix inversion with modern mathematical software tookits.

\begin{figure}[!htbp]
\centering
\subfloat[]{%
    \includegraphics[width=0.45\linewidth]{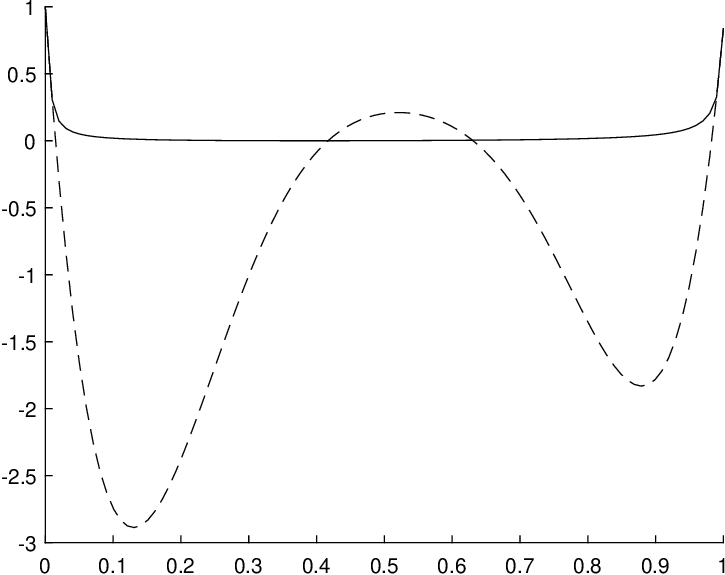}
}\hfill
\subfloat[]{%
    \includegraphics[width=0.45\linewidth]{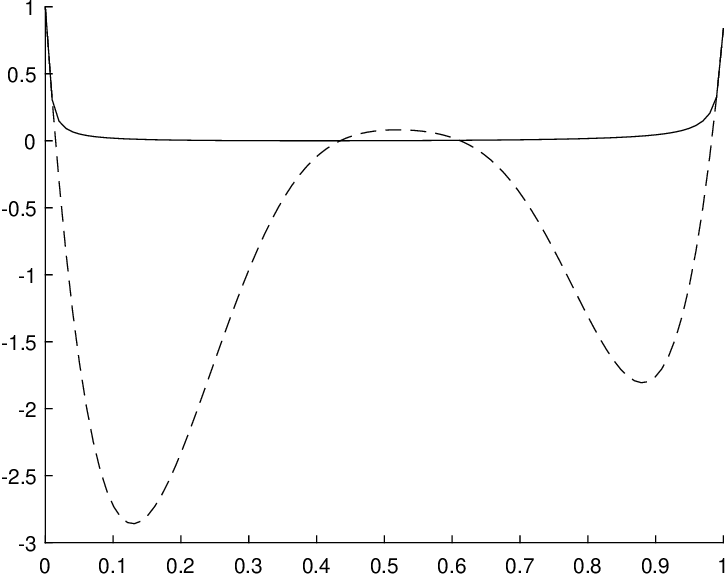}
}\\
\subfloat[]{%
    \includegraphics[width=0.45\linewidth]{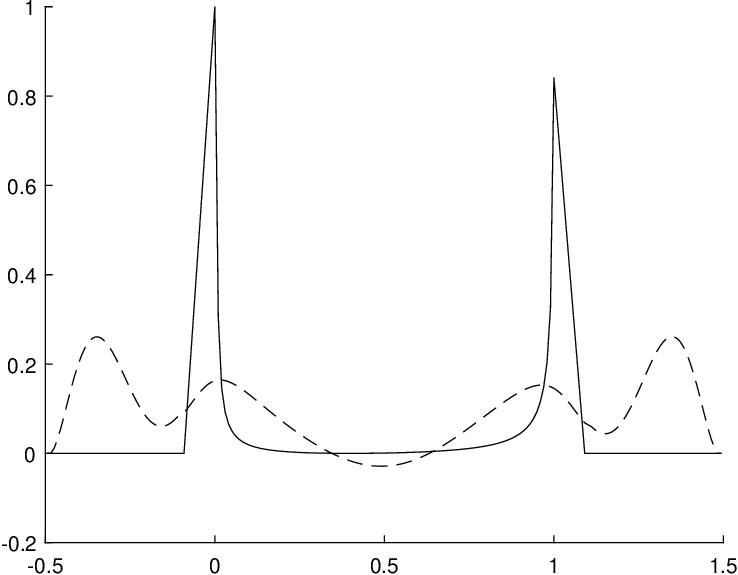}
}\hfill
\subfloat[]{%
    \includegraphics[width=0.45\linewidth]{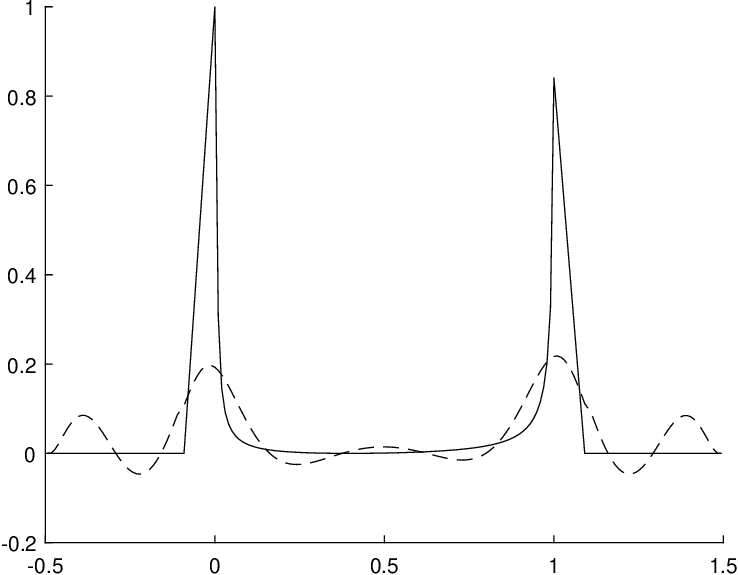}
}
\caption{Iterations for finding the Cutting Curve without ((a) and (b)) /with ((c) and (d)) smooth extrapolation to the spectrum when iterated to 500 ((a) and (c)) and 1500 ((b) and (d)) steps for the same spectrum. The solid line is the original spectrum and the dotted line is the evaluated Cutting Curve.  With smooth extrapolation to the original spectrum, the speed of convergence has significant improvement.}
\label{diag_19}
\end{figure}

\subsection{Treatment of Boundary Conditions}
The primal problem requires the Cutting Curve $g(x)$ confirming with the value of $f(x)$ and $f'(x)$ at both ends as boundary conditions, this conducts linear constraints directly onto the first and last elements to the solution vector of $\pmb g$.  For low and high pass signals of which the spectrum value are high and drop abruptly at either of the ends, this will make $g(x)$ inherits the trend of $f(x)$ at both ends at some interval so that convergence in this case will be very slow with reasonable grid space.  To alleviate this effect, we adopt a smooth extrapolation strategy by symmetrically appending gently decaying segments at both ends to make the extended value align with the minimum of $f(x)$, and the extended derivative smoothly approach to zero within the extended region.  This strategy significantly improves the condition number of the coefficient matrix of the discrete system, ensuring the numerical stability of the iterative solution, without any altering the primal spectrum within the original interval.  Fig. \ref{diag_19} shows the illustrations.

\subsection{Postprocessing: From Cutting Curve to Mode Identification}
Upon the above course when the algorithm converges, we obtain the optimal discrete vector $\pmb g^*$ above which significant intrinsic modes corresponds to the parts of the original spectrum $\pmb f$ that protrude distinctly.  However, there are some residual background fluctuations at the bottom that hinders one to objectively distinguish those modes.  To solve this problem, we assume that the residual consists of two components: first, a background fluttering component originate from noise or minor undulations, of which the amplitude concentrates within a very small range; second, significant protrusions from original modes, which are larger compared with the former and more scattered.  In this view, when investigating the value of the residual spectrum value distribution, we can intuitively find that most values concentrate near the level where fluttering component resides.  This inspires us to employ Kernel Density Estimation (KDE) to non-parametrically model the spectrum value distribution as a probability density function:
\begin{equation}
\phi(t)=\cfrac 1 {mh}\sum\limits_{i=1}^m K\left (\cfrac {t-r_j} {h} \right )
\end{equation}
where $K(\cdot)$ is Gaussian function, with $h$ the bandwidth and $r_j$ the discrete nodes in the residual vector $\pmb r$, and we evaluate the best threshold to eliminate the fluttering noise as
\begin{equation}
t^* = \underset{t}{\arg\max}\phi(t)
\end{equation}
so that we can isolate each mode in the function $f(x)-g^*(x)-t^*$.

After that, we use DBSCAN to merge adjacent peaks with low energy at the bottom for $f(x)-g^*(x)-t$.

\begin{figure}
  \centering
  \subfloat[]{
    \includegraphics[width=0.45\columnwidth]{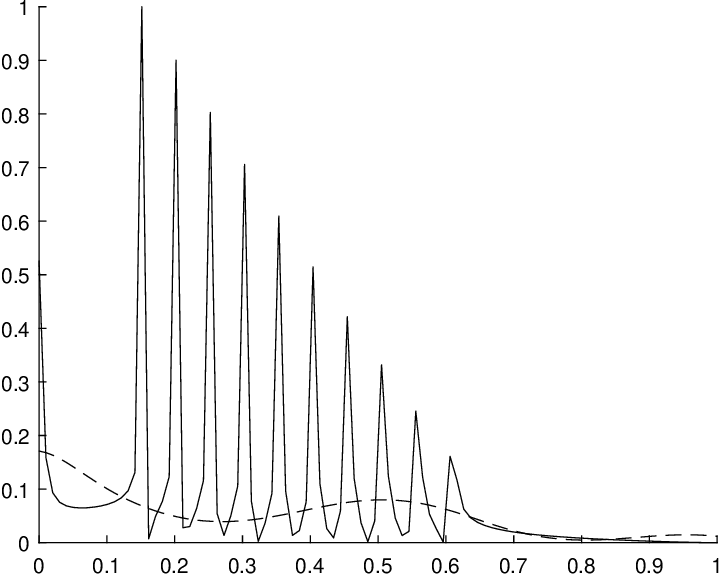}
    \label{fig:3a}
  }
  \subfloat[]{
    \includegraphics[width=0.45\columnwidth]{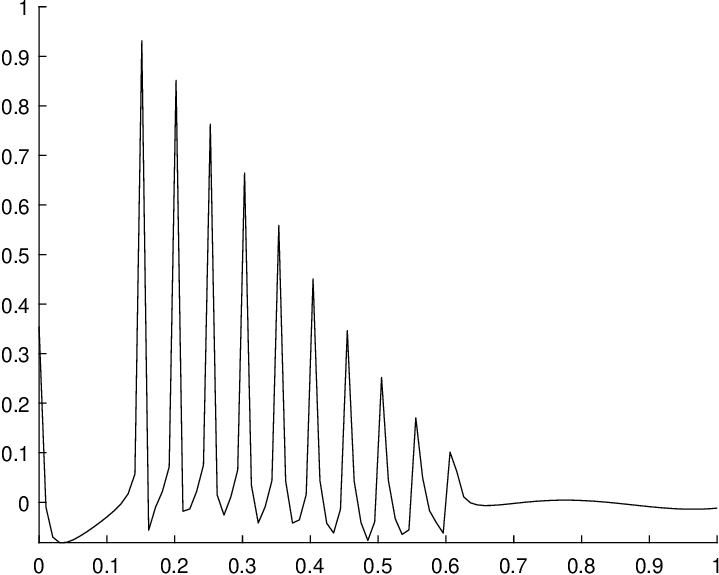}
    \label{fig:3b}
  }
\\
  \subfloat[]{
    \includegraphics[width=0.45\columnwidth]{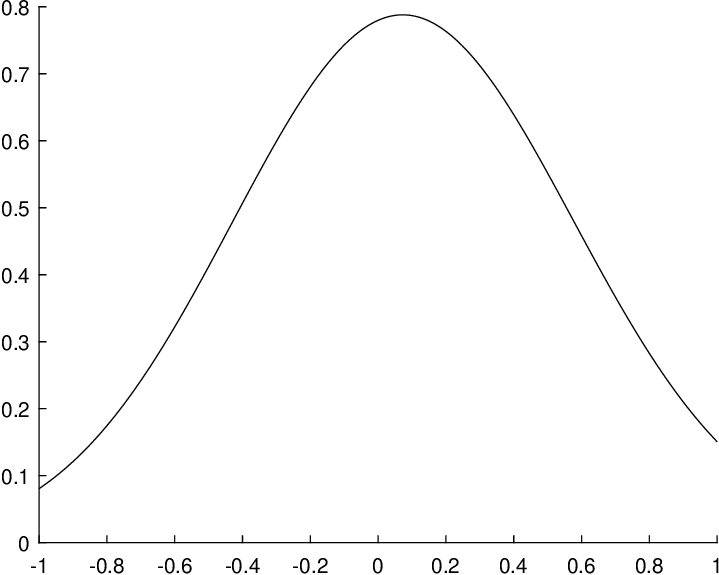}
  \label{fig:3c}
  }
  \subfloat[]{
    \includegraphics[width=0.45\columnwidth]{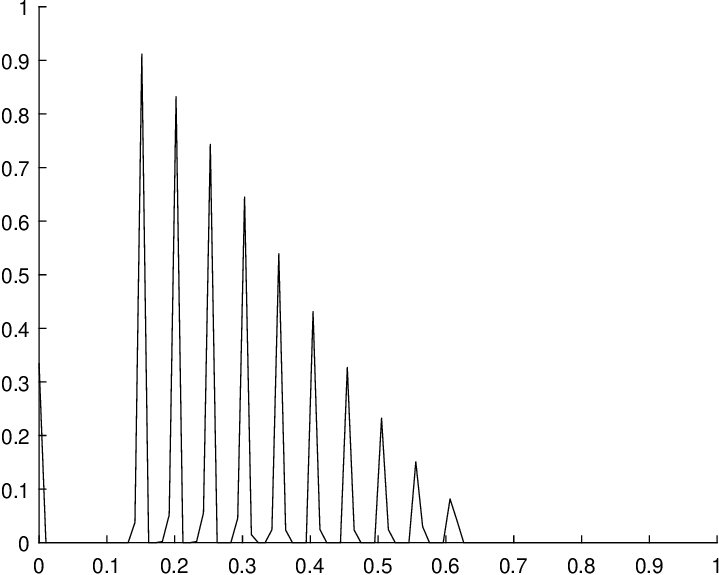}
    \label{fig:3d}
  }
\caption{Illustration of the non-consistence of the error between the converged optimal Cutting Curve and the original function.  Fig. \ref{fig:3a} shows the converged $g(x)$ and Fig. \ref{fig:3b} shows $f(x)-g(x)$.  Note that since $g(x)$ does have gap to the lower bound of $f(x)$ so that the bottom of $f(x)-g(x)$ is almost not on a horizontal line everywhere, which indicates there are some non-zero gaps.  These non-zero gaps will have a significant impact on the stability of the subsequent modal number solution.  \ref{fig:3c} shows the estimate kernel density pertaining the gap at the bottom.  \ref{fig:3d} shows the estimated $f(x)-g(x)$ after removal of the gap at the bottom.} 
\label{diag3}
\end{figure}

\subsection{The Whole Workflow}
Based on the above discrete framework, each step of Algorithm \ref{alg:algorithm} can be implemented by specific linear algebraic operations.  The whole workflow can be summarized as Algorithm \ref{alg:workflow}.  We first input the original spectrum $f(x)$, then we initialize weight functions $\alpha(x), \beta(x)$, penalty functions $\lambda(x), \mu(x)$ and the Cutting Curve function $g(x)$.  After that we smoothly extend $f(x)$ at both ends and begin Algorithm \ref{alg:algorithm} until it converges.  Then we do Kernel Density Estimation to figure out the threshold to eliminate the fluttering noise, and aggregate adjacent minor peaks into single coherent intrinsic mode by DBSCAN.

\begin{algorithm}[H]
\caption{The Whole Workflow}
\label{alg:workflow}
\begin{algorithmic}[1]
\STATE Input $\{f(x)\}$
\STATE Initialize $\alpha(x), \beta(x), \lambda(x), \mu(x), g(x)$
\STATE Smoothly Extending $f(x)$ at both ends.
\REPEAT
   	\STATE Run the kernel process of Algorithm \ref{alg:algorithm}
\UNTIL{convergence: $\|g^{n} - g^{n-1}\|_2^2 / \|g^{n}\|_2^2 < \epsilon$}
\STATE Estimate KDE for the computed $g(x)$.
\STATE Compute best threshold $t_{best} = \arg\max \text{KDE}_{g(x)}(t)$.
\STATE Subtract noise-gap threshold from $g(x)$ by $g(x) - t_{\text{best}}$.
\STATE Use DBSCAN to merge spatially adjacent minor peaks into a single coherent intrinsic mode.
\STATE Get $k$, $\{\omega_k\}$ from non-zero intervals.
\STATE Begin VMD procedure \cite{Dragomiretskiy} with $k$ and $\{\omega_k\}$
\STATE Output $\{u_k\}$
\end{algorithmic}
\end{algorithm}

\section{Experiments}
In this section, we show some experiments to validate the applicability of our methods.  Our experiments focus on 3 goals: (1) To validate the convergence and robustness of our method under synthetic signals like single-modality signal, dual-modality signal, segmented signal, comb-spectrum signal, dense-modality signal to demonstrate the extensive adaptability of the model across modes with various distribution.  (2) To compare our method against the most challenge-able method Successive Variational Mode Decomposition (SVMD), another recursive, adaptive method determining the intrinsic modes, exhibiting  the advantage of our method on not generating redundant modes or eliminating necessary modes.  (3) To understand the applicability of our method in real-world signals like electrocardiogram (ECG) signals, showing the accuracy of decomposition by investigation of the signals reconstructed from the decomposed modes.

All experiments are conducted under the same parameter settings unless otherwise stated.  The stopping criterion is set as
$\dfrac {\|g_n-g_{n-1}\|} {\|g_n\|} < 10^{-4}$.  To adaptively determine the effective frequency band, we compute the cumulative energy distribution of the single-sided power spectrum.  The cut-off frequency $f_c$ is defined such that
\begin{equation}
\displaystyle \cfrac{\int_0^{f_c} S(f) \text df}{\int_0^{+\infty} S(f) \text df} = 0.95, 
\end{equation}
which is found by scanning from the Nyquist frequency downward until the energy condition is met. The resulting $f_c$ effectively isolates the dominant spectral components while discarding high-frequency noise tails.  The truncated spectrum is then down‑sampled to 100 points, and the optimal cutting curve is determined on this down‑sampled spectrum.

\begin{figure}[!htbp]
  \centering
  \subfloat[]{
  	\includegraphics[width=0.45\columnwidth]{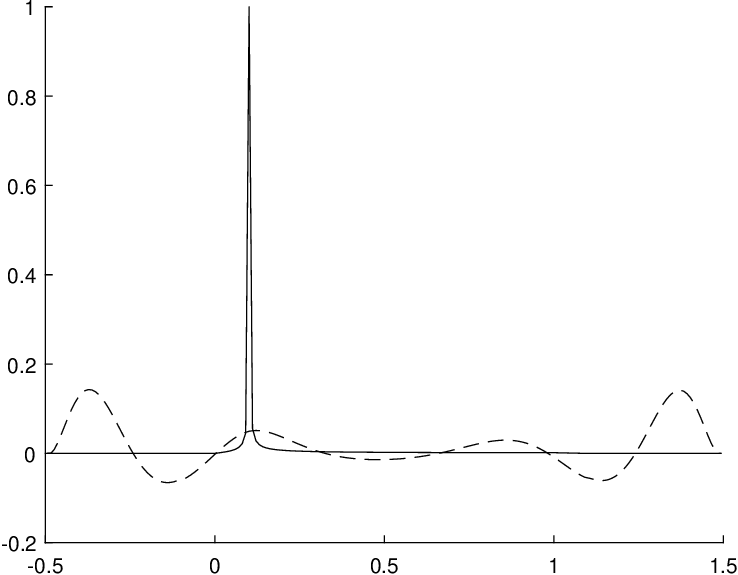}
  }
  \subfloat[]{
  	\includegraphics[width=0.45\columnwidth]{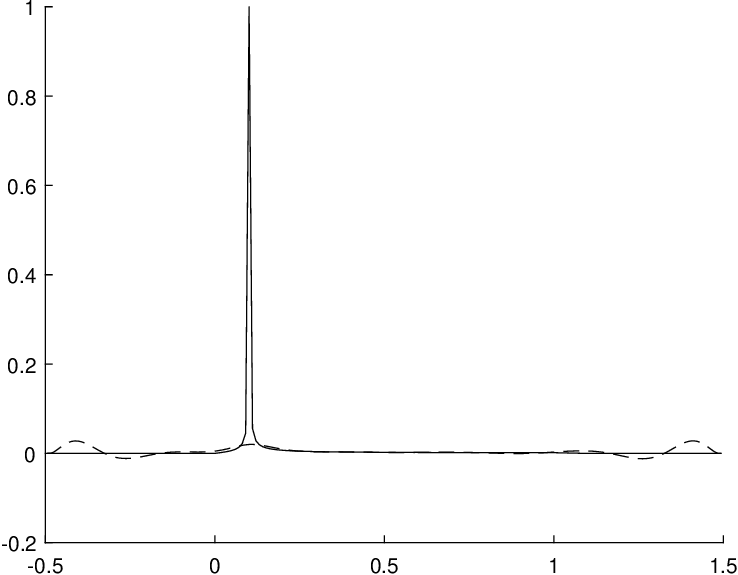}
  }
  \vspace{3pt}
  \subfloat[]{
  	\includegraphics[width=0.55\columnwidth]{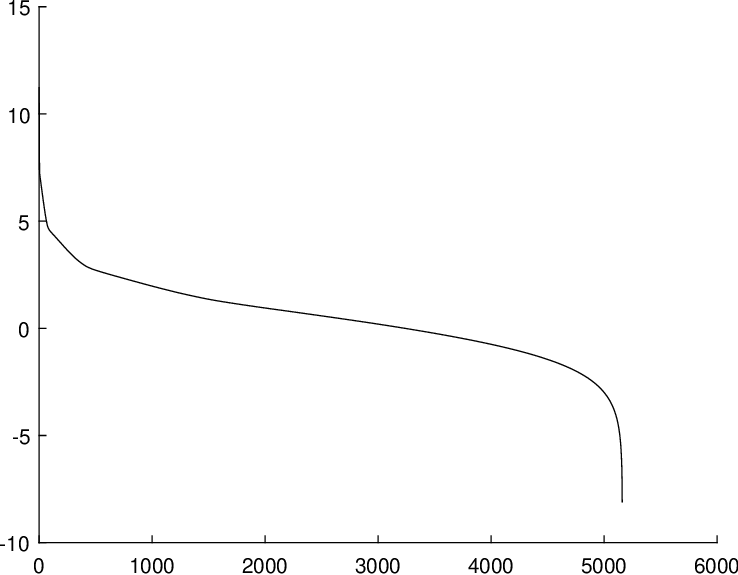}
  }
  
  \caption{The initial (a), final convergent state (b) and convergence curve in logarithmic coordinate system (c) with respect to the error for finding the optimum Cutting Curve of FFT with respect to time series $y(t)=100\sin 20\pi t$.  The iteration lasts for 5161 steps and costs 1.88 seconds.}
  \label{diag4}
\end{figure}

\subsection{Experiments on Finding the Cutting Curve}

We choose five typical signals including single mode signal, multi-mode signal, piecewise continuous signal, signal that are not necessarily narrow-banded, composite signal in which the center frequency of its components are very close.  We only illustrate the initial and final state of the iteration for the evaluation of the cutting curve, as well as the convergent curve with respect to the error in logarithm to indicate the number of iteration steps in the caption, and leave the detail of decomposition result in Appendix \ref{Experiment}.

Experiment 1:  The first example is a single mode signal, on the other words, the signal is narrow-banded and only contains one center frequency.  Here we choose $y(t)=100\sin 20\pi t$.  As we see the normalized center frequency is $0.1$.  The algorithm takes 5161 steps and 1.88 seconds to converge.

Experiment 2:  The second example is a multi-mode signal, $y(t)=10\cos(10\pi t)+20\sin(20\pi t)$, which has two modes, exactly.  The normalized center frequency is 0.05 and 0.1, which is relatively very close to each other.  However, our algorithm can still survive where center frequencies of different modes are very close, since from 6000th iteration and on, the cutting curve forms a small bulge at the close frequency points, which just separates the two frequency points.

\begin{figure}[!htbp]
  \centering
  \subfloat[]{
  	\includegraphics[width=0.45\columnwidth]{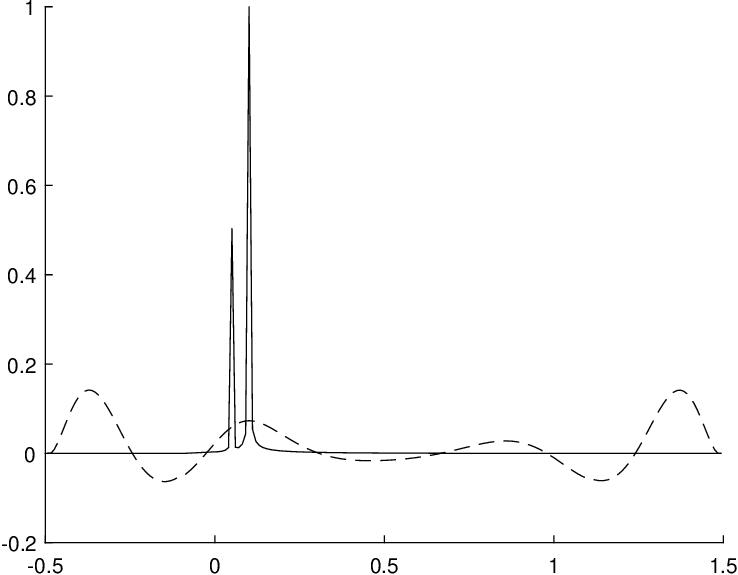}
  }
  \subfloat[]{
  	\includegraphics[width=0.45\columnwidth]{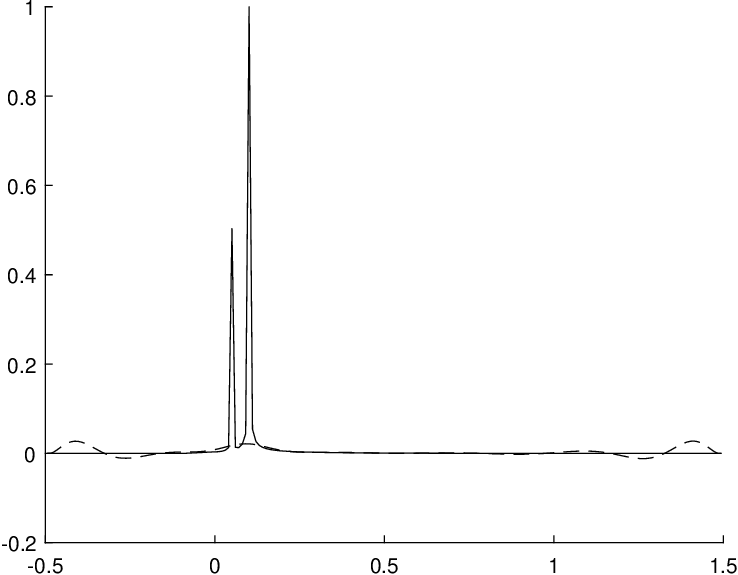}
  }
  \vspace{3pt}
  \subfloat[]{
  	\includegraphics[width=0.55\columnwidth]{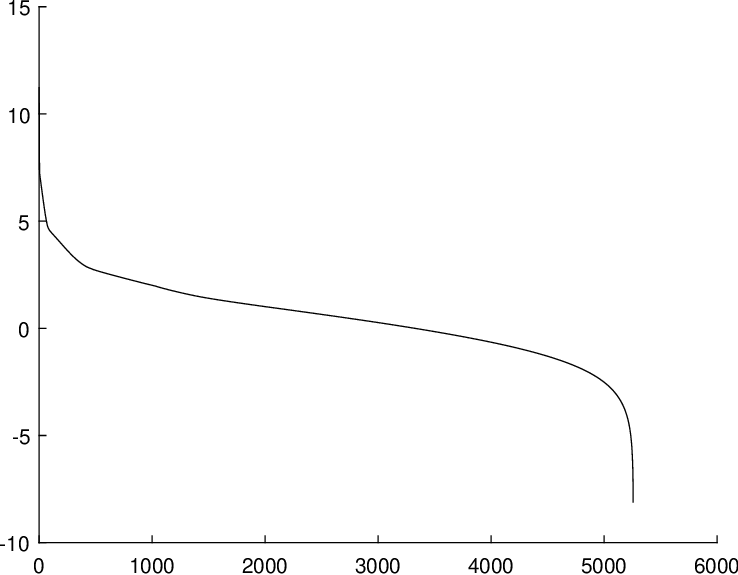}
  }
  \caption{The initial (a), final convergent state (b) and convergence curve in logarithmic coordinate system (c) for finding the Cutting Curve of FFT with respect to FFT of $y(t)=10\cos 10\pi t+20\sin 20\pi t$.  The iteration lasts for 5258 steps and costs 1.89 seconds.}
  \label{diag7}
\end{figure}

Experiment 3:  The next example is a little complicated, it has four modes totally, we choose
\begin{equation}
\footnotesize
    y(t) = \left\{
    \begin{aligned}
        6t^2 + \cos(10\pi t + 10\pi t^2) + \cos(60\pi t), &\quad t \in [0,0.5] \\
        6t^2 + \cos(10\pi t + 10\pi t^2) + \cos(80\pi t - 10\pi), &\quad t \in (0.5,1]
    \end{aligned}
    \right.
\label{seg_func}
\end{equation}
as our signal.  It has a low frequency component, $6t^2$, two pure harmonic components, $\cos 60\pi t, \cos (80\pi t-10\pi)$, in different intervals that does not intersect with each other, respectively, and one narrow banded component, $\cos(10\pi t+10\pi t^2)$, in the whole interval.  To our surprise, the iteration converges even more quickly than those signals that has fewer modes.  In our cases, it takes about 3026 iterations to achieve the convergence.

\begin{figure}[!htbp]
  \centering
  \subfloat[]{
  	\includegraphics[width=0.45\columnwidth]{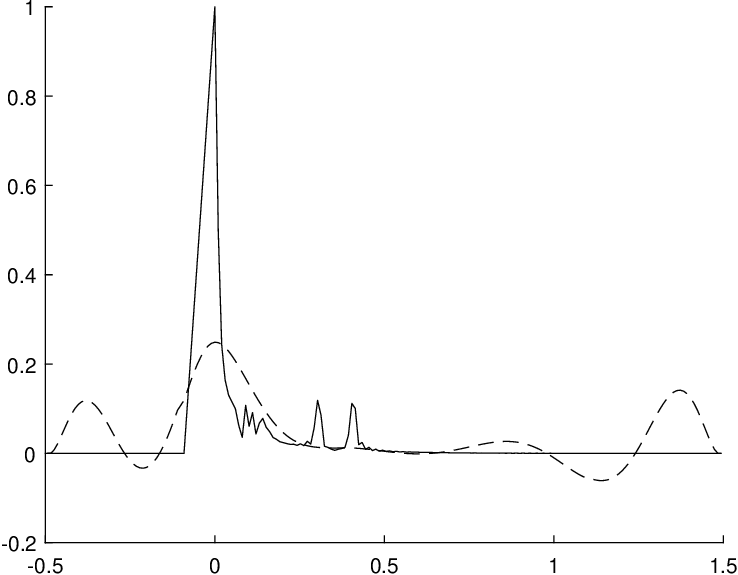}
  }
  \subfloat[]{
  	\includegraphics[width=0.45\columnwidth]{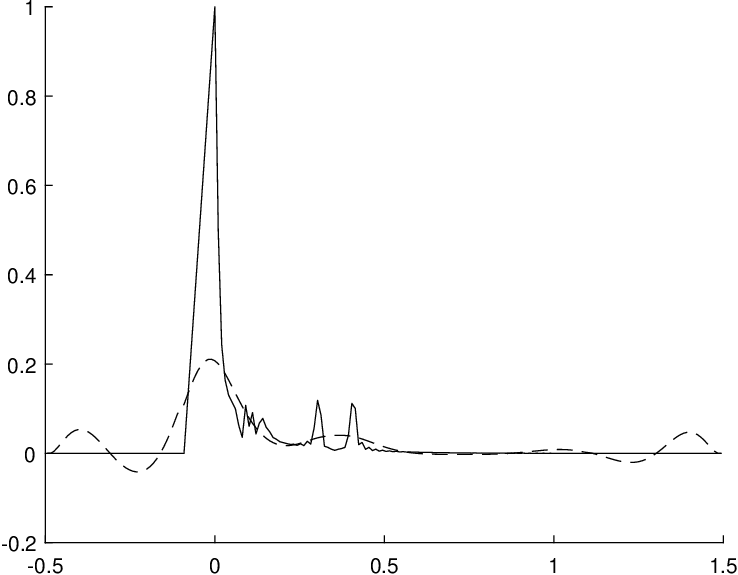}
  }
  \vspace{3pt}
  \subfloat[]{
  	\includegraphics[width=0.55\columnwidth]{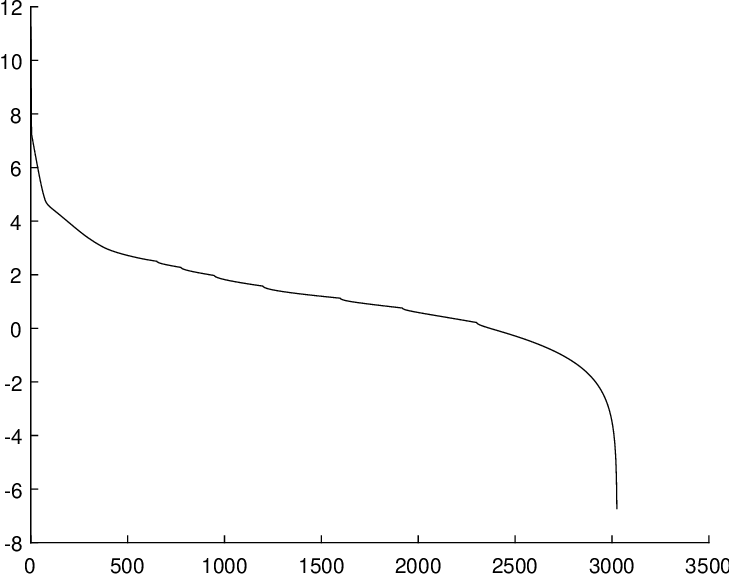}
  }
  \caption{The initial (a), final convergent state (b) and convergence curve in logarithmic coordinate system (c) for finding the Cutting Curve of FFT with respect to time series (\ref{seg_func}).  Note that although the spectrum of second mode, $\cos 10\pi t$, is affected by that of the first mode, $6t^2$, the computed Cutting Curve precisely captured the trend of the lower bound so that all the modes can be filtered out.  The iteration lasts for 3026 steps and costs 1.07 seconds.}
  \label{diag10}
\end{figure}

Experiment 4:  Now we show some examples on signals that conflict with the narrow-band assumption.  We take $\cfrac 1 {1.2+\cos 2\pi t} + \cfrac {\cos (32\pi t +0.2\cos 64\pi t)}{1.5+\sin(2\pi t)}$.  As we see that, the spectrum of the signal exhibits two main peaks, however the signal is not a completely narrow-band signal since it exhibits comb-shaped spectrum with respect to harmonic components, although the energy of those harmonic components are relatively very small(Fig. \ref{diag13}).  Our algorithm extracts three components of the signal, despite that in the mainstream view there should be only two IMFs.   However, since it is not a narrow-band signal, the difficulty of decomposition objectively exists since one can not even provide a specific standard to determine whether these small harmonic components should be considered as an intrinsic mode.  We also observe that, even in \cite{Dragomiretskiy,Mojtaba} when carefully zoomed in, there does exist small oscillations in the first decomposed component which indicates the non-purity of it.  Despite the fact that our results has three components, every separate component exhibits better purity than the pioneer works.  On the other side, from the perspective of reconstruction error, it is shown in this example that, results generated by algorithm in \cite{Dragomiretskiy} with both 2 and 3 modes has reconstruction error as twice as large as ours, which also indicates that our algorithm has better performance in this case.

\begin{figure}[!htbp]
  \centering
  \subfloat[]{
  	\includegraphics[width=0.45\columnwidth]{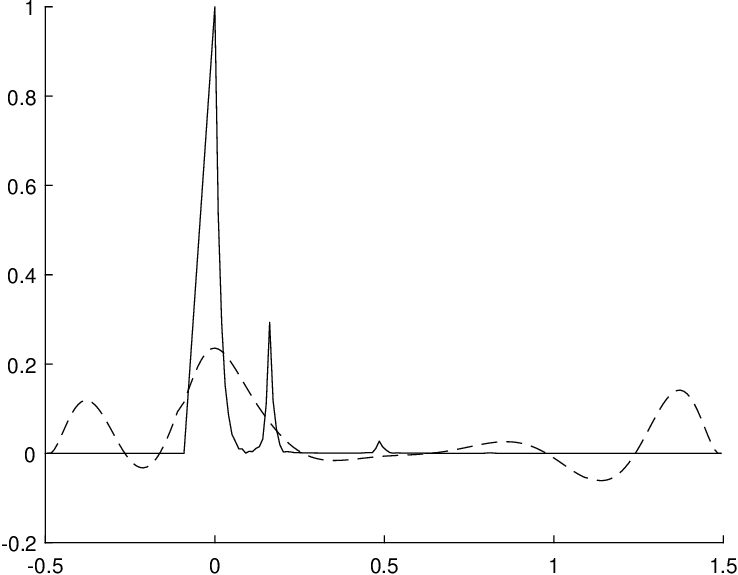}
  }
  \subfloat[]{
  	\includegraphics[width=0.45\columnwidth]{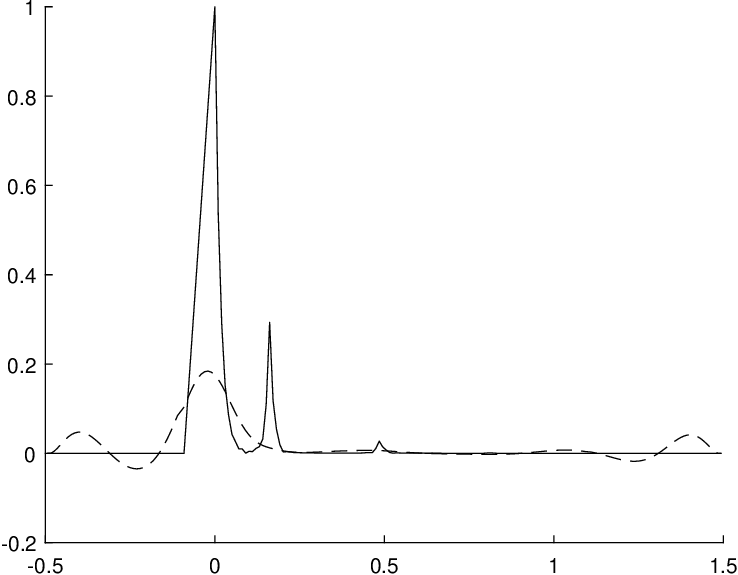}
  }
  \vspace{3pt}
  \subfloat[]{
  	\includegraphics[width=0.55\columnwidth]{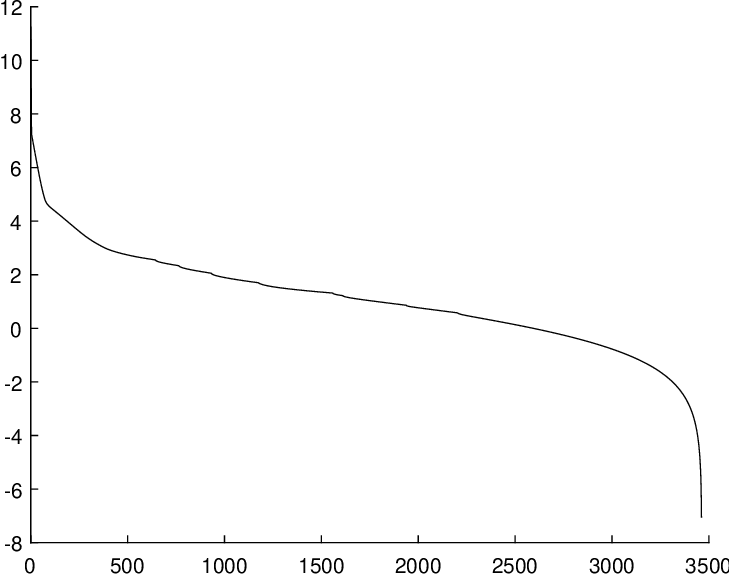}
  }
  \caption{The initial (a), final convergent state (b) and convergence curve in logarithmic coordinate system (c) for finding the cutting curve of function $\frac 1 {1.2+\cos 2\pi t} + \frac {\cos (32\pi t +0.2\cos 64\pi t)}{1.5+\sin(2\pi t)}$.  The iteration lasts for 3463 steps and costs 1.25 seconds.}
  \label{diag13}
\end{figure}

\begin{figure}[!htbp]
\centering
\subfloat[]{
    \includegraphics[width=0.95\columnwidth]{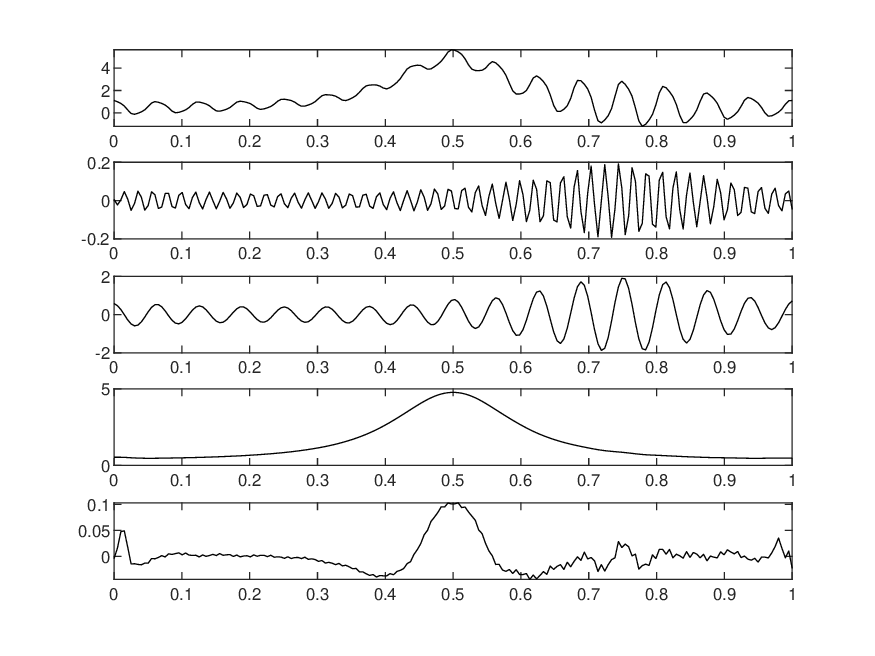}
}
\vspace{3pt}
\subfloat[]{
    \includegraphics[width=0.95\columnwidth]{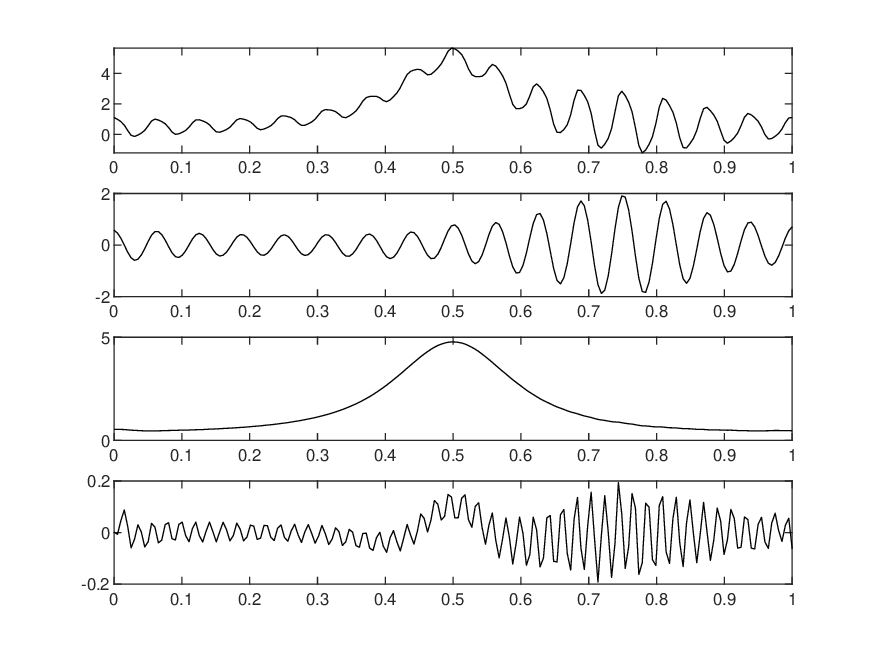}
}
\caption{to be continued}
\label{diag15ab_article}
\end{figure}

\begin{figure}[!htbp]
\centering
\ContinuedFloat
\subfloat[]{
    \includegraphics[width=0.95\columnwidth]{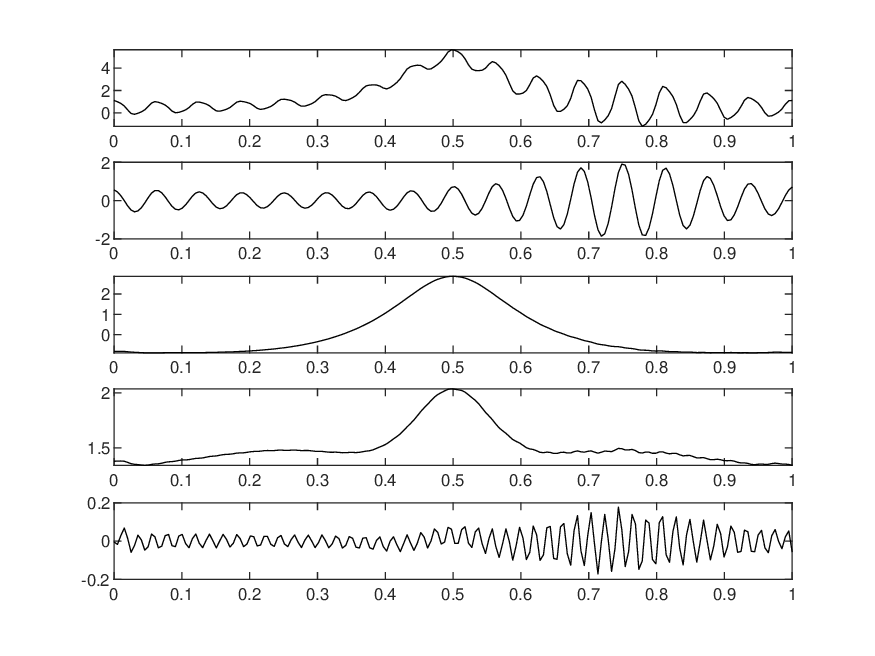}
\label{diag15c_article}
}
\caption{The decomposition result (a) of $\cfrac 1 {1.2+\cos 2\pi t} + \cfrac {\cos (32\pi t +0.2\cos 64\pi t)}{1.5+\sin(2\pi t)}$ from the subsequent VMD, using the found number of modes and center frequency as initialization, compared with original VMD results with IMFS number 2 (b) and 3 (c) respectively.  The sub-figure from top to bottom are respectively the source signal, the decomposed modes(there are three) and the residual.  As compared with (b) and (c), for original VMD, both the results has reconstruction error twice as large as ours.}
\end{figure}

Experiment 5:  Next we show even when the central frequencies are very close, our algorithm can also work well.  We take signals as $y(t)=6t + \displaystyle\sum\limits_{i=1}^{10} (13-i)\cos[(20+10i)\pi t]$.  The signal has one low-frequency component and 10 alternating current components with different amplitudes, of which the minimum difference between frequency gaps is only 0.04 after normalization and our algorithm can still survive in figuring out all the components. 
\setcounter{subfigure}{0}
\begin{figure}[!htbp]
  \centering
  \subfloat[]{
  	\includegraphics[width=0.45\columnwidth]{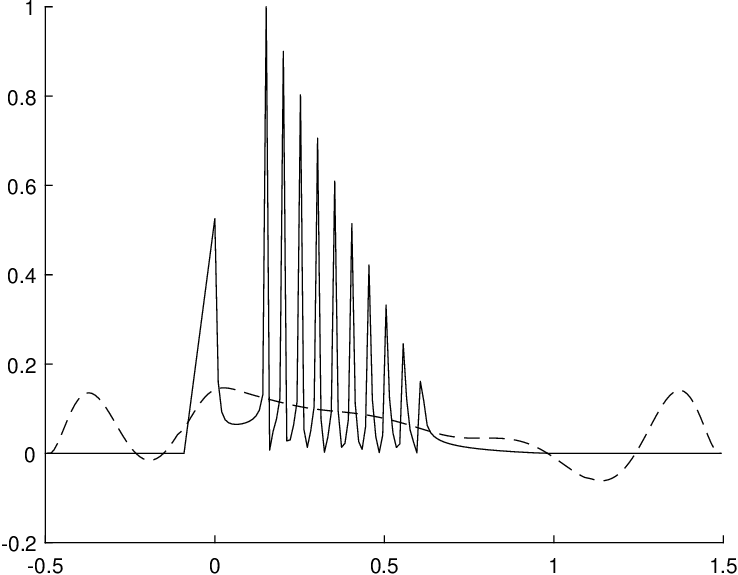}
  }
  \subfloat[]{
  	\includegraphics[width=0.45\columnwidth]{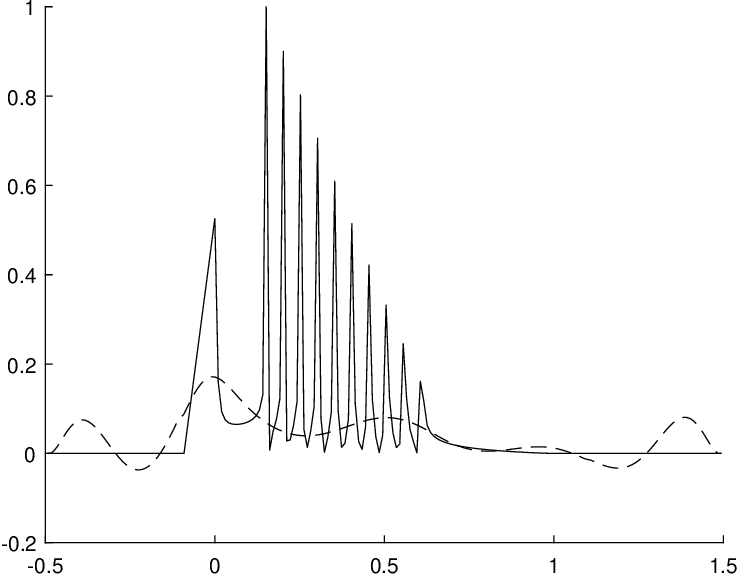}
  }
  \vspace{3pt}
  \subfloat[]{
  	\includegraphics[width=0.55\columnwidth]{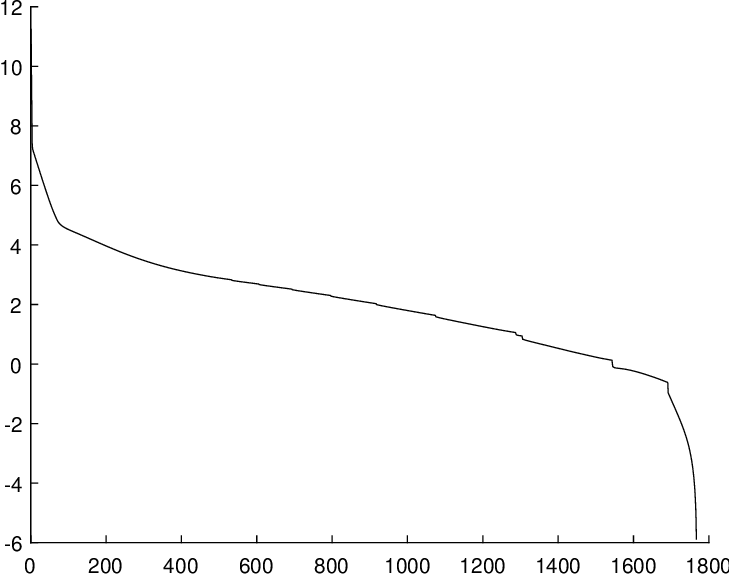}
  }
  \caption{The initial (a), final convergent state (b) and convergence curve in logarithmic coordinate system (c) for finding the cutting curve of function 
 $y(t)=6t + \displaystyle\sum\limits_{i=1}^{10} (13-i)\cos[(20+10i)\pi t]$, The iteration lasts for 1768 steps and costs 0.64 seconds.}
  \label{diag16}
\end{figure}
\clearpage
\subsection{Experiments on Central Frequencies}
In order to prove that our algorithm can find frequencies close to the practical central frequencies for each IMF component, we compare our results from the above Example 1 to Example 5, with the results of the original VMD process with the same number of IMF, to keep away the impact to the performance of VMD from the initial frequency value that we provide to the following VMD procedure.  We first get the frequency estimation from our algorithm and the result of pure VMD procedure with the same number of IMFs respectively, and then sort them in order, to observe the difference between the results of ours and the corresponding values from VMD.

\begin{figure}[hbtp]
\centering
\includegraphics[scale=0.65]{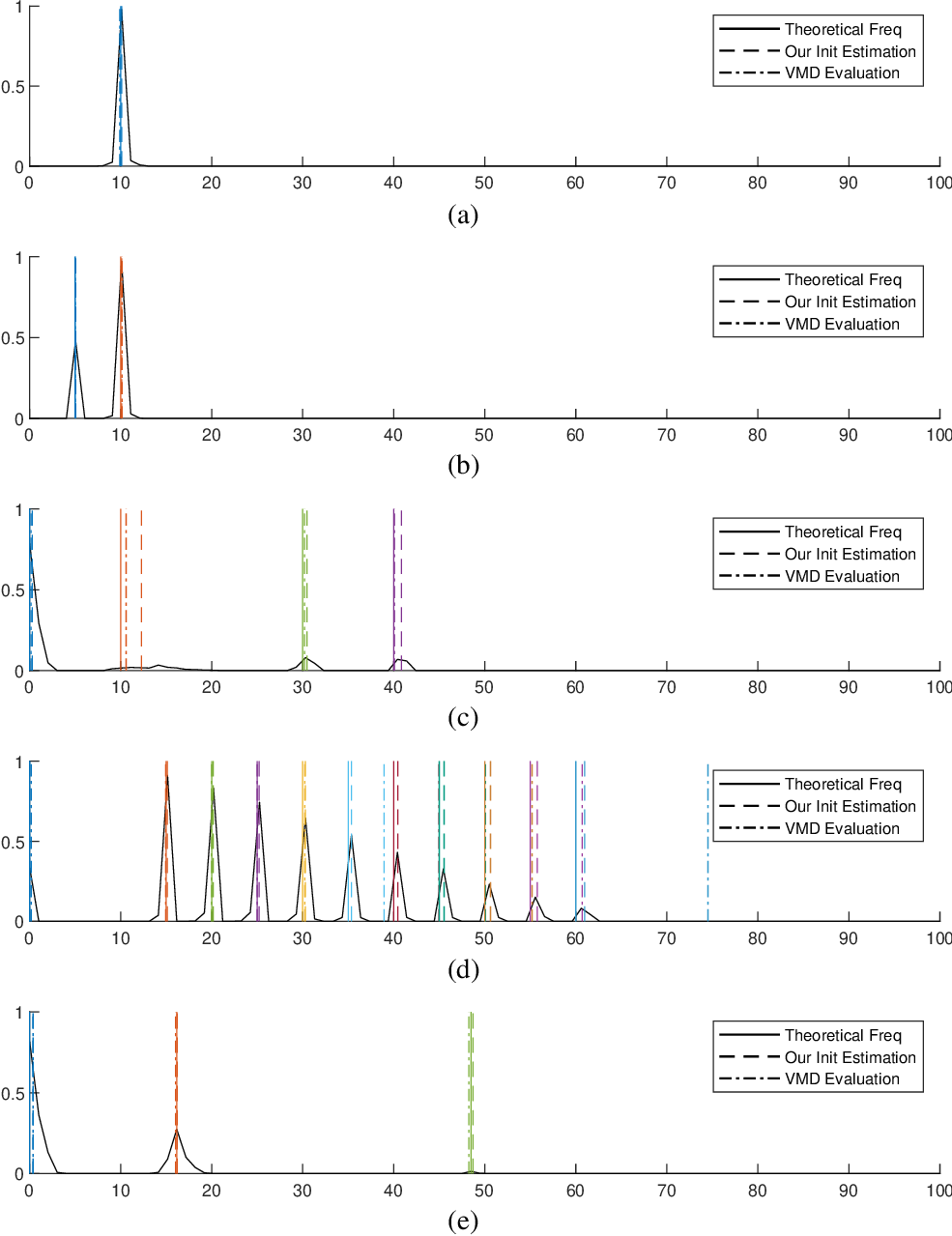}
\caption{The theoretical frequencies (solid line), our estimated frequencies (dashed line) and the VMD-evaluated frequencies (dash-dotted line) after iteration with respect to Experiment 1-5 (a)-(e).}
\label{rows_5_col_1}
\end{figure}

Fig. \ref{rows_5_col_1} shows the theoretical center frequencies, our estimated frequencies and VMD-evaluated frequencies with respect to the experimental synthetic signals in order, where the solid line stands for the theoretical frequency of each mode, the dashed line stands for our estimation and the dotted-dash line stands for the ones after VMD's evaluation, and different colors stands for different decomposed modalities.  It can be seen that our estimated values are very close to the actual modal frequencies, and the accuracy of the estimation is independent of the number of modes to be decomposed, although for VMD procedure there are cases in which the central frequencies of modes with small energy will deviate from the theoretical.

\begin{figure}[!htbp]
\centering
  \includegraphics[width=0.6\columnwidth]{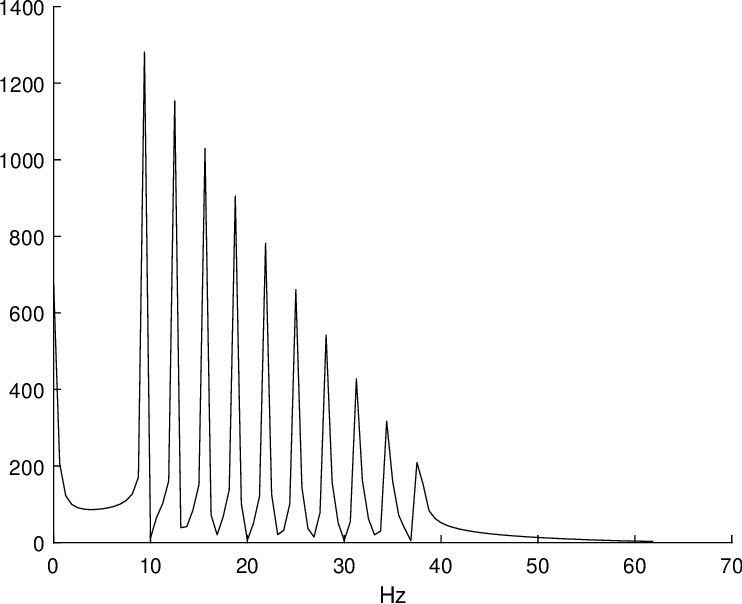}
  \caption{The original spectrums for signal $y(t)=6t + \displaystyle\sum\limits_{i=1}^{10} (13-i)\cos[(20+10i)\pi t]$.}
  \label{SVMD_spectrum_1_original}
\end{figure}

\begin{figure}[!htbp]
\centering
  \includegraphics[width=0.75\columnwidth]{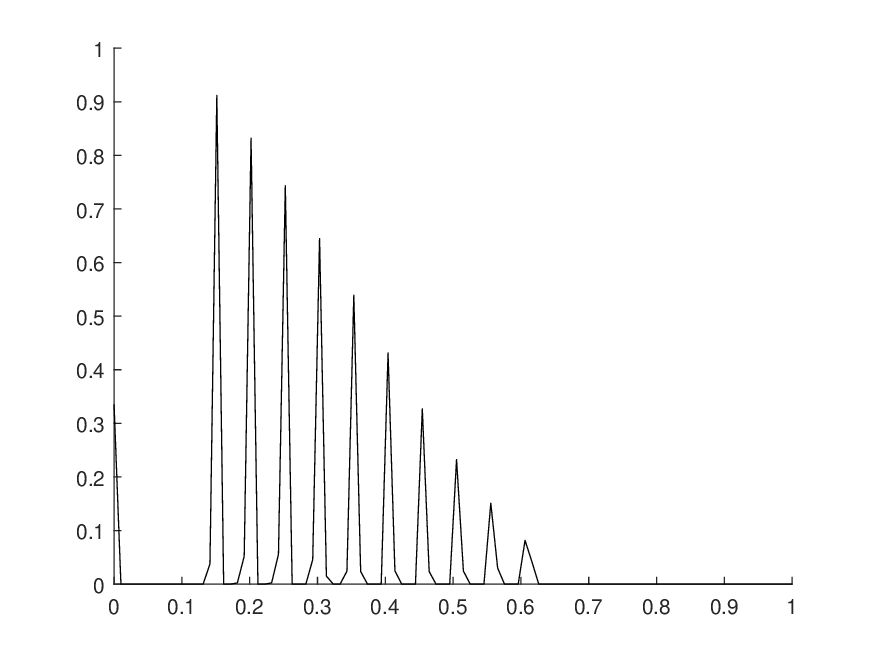}
  \caption{The decomposed spectrums for signal $y(t)=6t + \displaystyle\sum\limits_{i=1}^{10} (13-i)\cos[(20+10i)\pi t]$ by our method.}
  \label{SVMD_spectrum_1_ours}
\end{figure}

\subsection{Experiments Compared with Variational Mode Decomposition}

We noticed that, the Successive Variational Mode Decomposition(SVMD) method is another competitive method that can also automatically determine the number of IMFs for VMD\cite{Mojtaba} without any prior knowledge although it also lacks of rigorous proof on the convergence.  This method retrieves the modes one after another by recursively finding the $L$th mode from the residual, assuming the first $L-1$ modes are found after the last turn, until the final residual can be ignored with respect to some criterion set by the algorithm.  However, since the method is in essence recursive, the accuracy for the next mode to be evaluated is highly depended on the  residual determined by previous decomposition in all former turns, which will inevitably introduce accumulated error during the propagation.  We used official MATLAB implementation on SVMD \cite{Matlab} and observed that for narrow-banded signals wherein the center frequencies of different components are very close, the SVMD method may either introduce some redundant modes or drop some significant modes in their decomposition results, yet the phenomenon is hardly observed in our algorithm.

\begin{figure}[!htbp]
\centering
  \includegraphics[width=0.8\columnwidth]{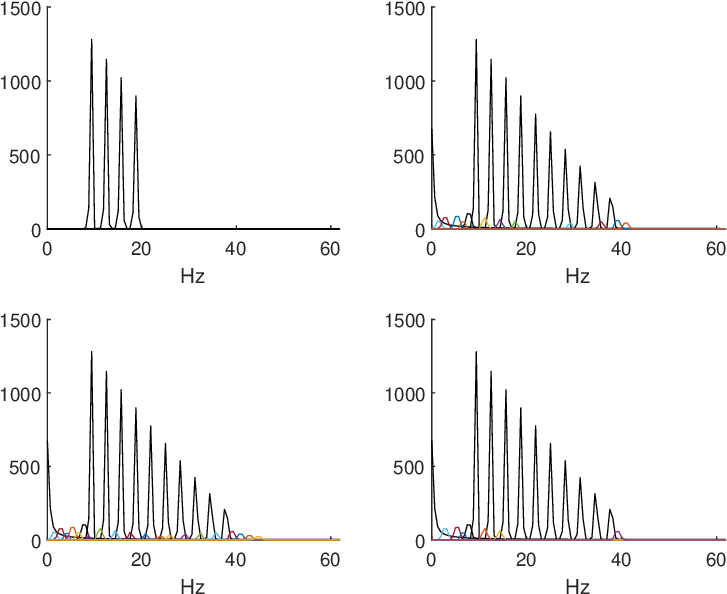}
  \caption{The decomposed spectrums for signal $y(t)=6t + \displaystyle\sum\limits_{i=1}^{10} (13-i)\cos[(20+10i)\pi t]$ by SVMD under different stopping criteria by considering noise(northwest), exact decomposition(northeast), Bayesian Estimation(southwest), power of the last mode(southeast).}
  \label{SVMD_spectrum_1}
\end{figure}

\begin{figure}[!htbp]
\centering
  \includegraphics[width=0.75\columnwidth]{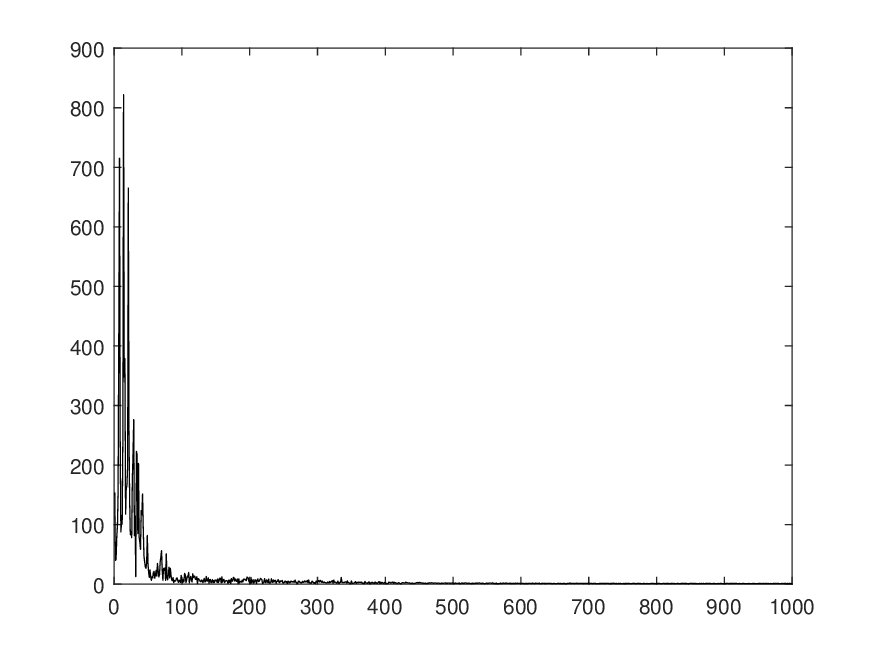}
  \caption{The original spectrums for signal of the 102th record in MIT-BIH Arrhythmia Database.}
  \label{SVMD_spectrum_2_original}
\end{figure}

In this section, we pick 2 typical signals for illustration.  The first one is $y(t)=6t + \displaystyle\sum\limits_{i=1}^{10} (13-i)\cos[(20+10i)\pi t]$, which we have investigated in the experiments for finding the Cutting Curve.  For the completeness of our comparison, we investigated all stopping criterion implemented in the source code of SVMD.  The method considering about eliminating noise(northwest) drops the base-frequency and a batch of high frequency, while method for extract reconstruction, method based on Beyesian Estimation, or method considering the power of last mode all result in some redundant modes with low energy on the decomposed spectrum (see colored spectrum at the bottom in each subplot), which are not the essential modes that should included in the decomposition result (Fig. \ref{SVMD_spectrum_1}).  In contrast, our method(Fig. \ref{SVMD_spectrum_1_ours}) accurately extracts all necessary modalities without any redundant ones.

\begin{figure}[!htbp]
\centering
  \includegraphics[width=0.7\columnwidth]{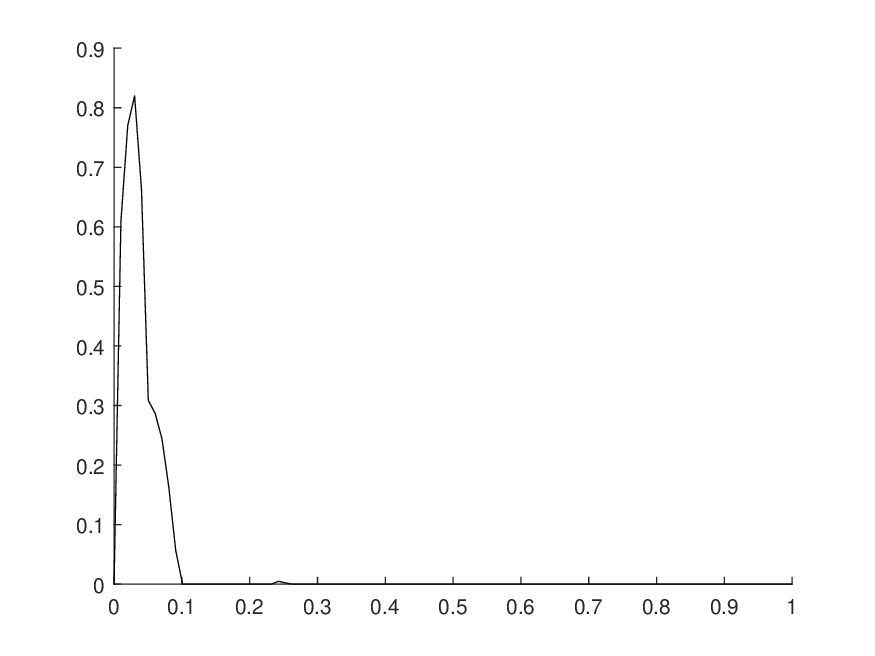}
  \caption{The decomposed spectrums for the V2 lead of the 102th record in MIT-BIH Arrhythmia Database by our method.}
  \label{SVMD_spectrum_2_ours}
\end{figure}

\begin{figure}[!htbp]
\centering
  \includegraphics[width=1.0\columnwidth]{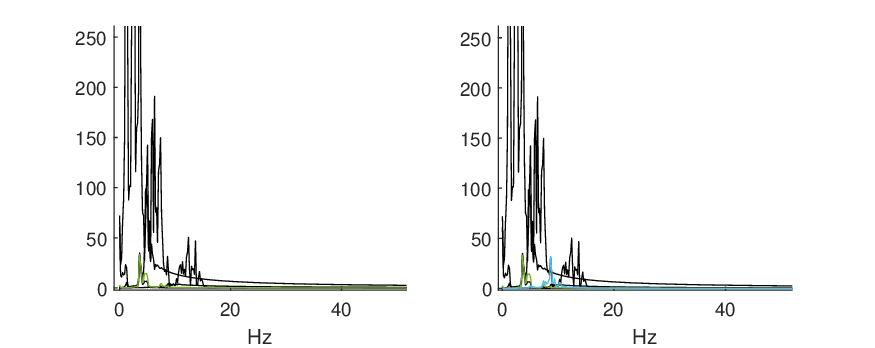}
  \caption{Local zoom of the decomposed spectrums for signal with respect to the V2 lead of the 102th record in MIT-BIH Arrhythmia Database, and decomposed by \protect\cite{Matlab} under different stopping criteria by considering Bayesian Estimation(left), power of the last mode(right).  Note that the figure for the one considering noise and the one with exact decomposition is missed since the SVMD failed to converge in 10 seconds.}
  \label{SVMD_spectrum_2}
\end{figure}

The second one is ECG signals from V2 lead of the 102th record in MIT-BIH Arrhythmia Database \protect\cite{MIT-BIH,Goldberger}.  The experiment result also shows similar phenomenon as stated in the former experiment, the method for exact decomposition fails to converge in more than 10 seconds so we cannot draw the spectrum of the corresponding figure (Fig. \ref{SVMD_spectrum_2}).  Our method in this case can also succeed in retrieving all the effective modes(Fig. \ref{SVMD_spectrum_2_ours}).

\clearpage
\subsection{Experiment on Real Signals}

To demonstrate the generality for different application domains, we also evaluate its performance on electrocardiogram (ECG) signals. ECG signals are representative examples of real-world non-stationary physiological signals composed of multiple characteristic waveforms (such as the P wave, QRS complex, and T wave), each occupying distinct frequency bands. Accurately decomposing an ECG signal into its constituent modes is therefore a challenging yet practically meaningful test for any adaptive mode decomposition framework.

We apply our method to selected records from the MIT-BIH Arrhythmia Database. Specifically, we extract 2000-sample segments from the MLII lead of Record 100, the V4 lead of Record 124, and the V2 lead of Record 102.  For each signal, our method automatically determines the number of intrinsic mode functions and their corresponding initial center frequencies, which are then applied to VMD for the final decomposition.

Fig.\ref{test_case_ecg1}-Fig.\ref{test_case_ecg3} shows the decomposition result of MLII lead of Record 100, the V4 lead of Record 124, and the V2 lead of Record 102.  In each image, the top subfigure represents the original signal, the bottom subfigure represents the residual signal, and the middle graphs represent the decomposed IMFs.  (b) stands for the spectrums of each decomposed IMF, (c) shows the final converged cutting curve for the spectral amplitude of each signal, and (d) shows the reconstructed signal (orange) with respect to the source signal (blue).  It is shown that our algorithm determines 4, 4, 2 modes for these three signals in order.  At the time domain we can see that in each decomposed IMFs in our chosen number of mode can reflect the basic characteristic that is helpful for the disease diagnosis.  For example, in Fig.\ref{test_case_ecg1}(a) we see the heartbeat rhythm from IMF 1-3 and IMF 4 amplifies the characteristic before and after the peaks.  The similar phenomenon also happens in Fig.\ref{test_case_ecg2}(a).  In Fig.\ref{test_case_ecg3}We see that each modality maintains the significant shape features of the original signal.  The Correlation Coefficients between the reconstructed and source signal is 0.9993, 0.9986, 0.9998 respectively.

Fig.\ref{vmd_ecg_performance} shows the performance comparison between our preferred mode number and center frequencies with respect to randomly chosen mode number from 2 to 10 for each investigated signal.  The first line is the result for MLII Lead of Record 100, the second line is the result for V4 Lead of Record 124 and the third line is the result for V2 Lead of Record 102.  (a),(c),(g) stands for the performance of orthogonality (average value of off-diagonal elements), (b),(e),(h) stands for the Correlation Coefficients performance and (c),(f),(i) stands for Power Ratio of Residual performance.  The red points stands for the performance under our chosen mode number and center frequencies, while the blue lines stands for the performance from randomly choosing mode number from 2 to 10.  The first 2 lines shows that our algorithm can choose a reasonable mode number protecting over-decomposition, while the third line shows that the V2 Lead of Record 102 seems under-decomposition.

However, Fig.\ref{ecg1_mode_7}-Fig.\ref{ecg3_mode_3} shows that by artificially increasing the mode number it seems increasing some redundant mode or similar mode.  In Fig.\ref{ecg1_mode_7} and Fig.\ref{ecg2_mode_7} the first 5 modes are similar in that they all reflect the heart beat rhythm, although they are almost orthogonal due to their frequency increasing relationship with each other.  However, this does not bring any improvement to the interpretability for the help of disease diagnosis.  Especially, in Fig.\ref{ecg3_mode_3} we see when setting mode number into 3, the decomposed modes can hardly reflect the basic characteristic of the source signal.  In contrast, the performance of orthogonality in turn improves with the increasing of the mode number as in Fig.\ref{ecg1_mode_7} and Fig.\ref{ecg2_mode_7} the former 5 modes are almost orthogonal due to their frequency increasing relationship with each other.  This indicates that evaluating the quality of signal decomposition cannot be solely based on performance indicators, but also on the physical meaning of the actual decomposed signal. To some extent, our modal determination method is naturally related to the geometric shape of the signal spectrum, and in most cases, the number of modes is determined based on the identifiability of the spectral peaks. The resulting modes can be proven valuable in engineering scenarios.

\begin{figure}
\centering
\includegraphics[scale=0.7]{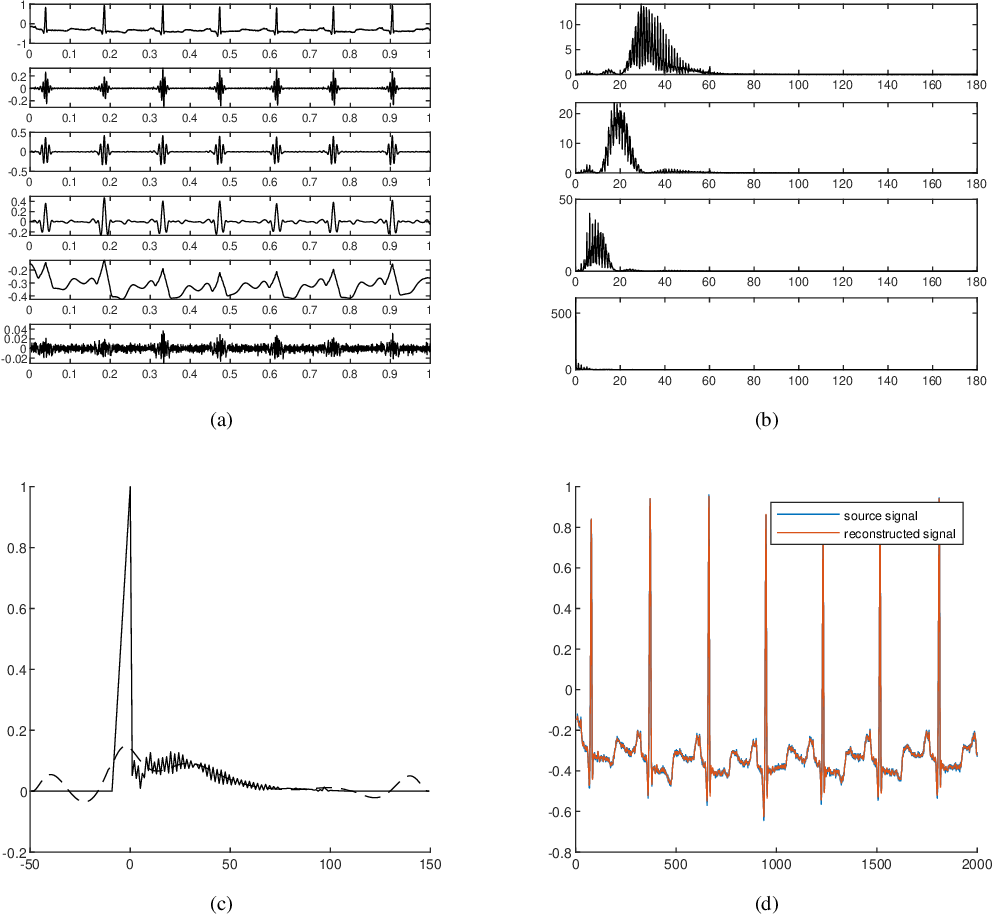}
\caption{The Decomposition result for MLII Lead of Record 100. (a) The original signal(top), the decomposed IMFs(middle 4), the residual(bottom), (b) The spectrum of each IMF, (c) The final converged cutting line, (d) The source signal and the reconstruction signal.}
\label{test_case_ecg1}
\end{figure}

\begin{figure}
\centering
\includegraphics[scale=0.7]{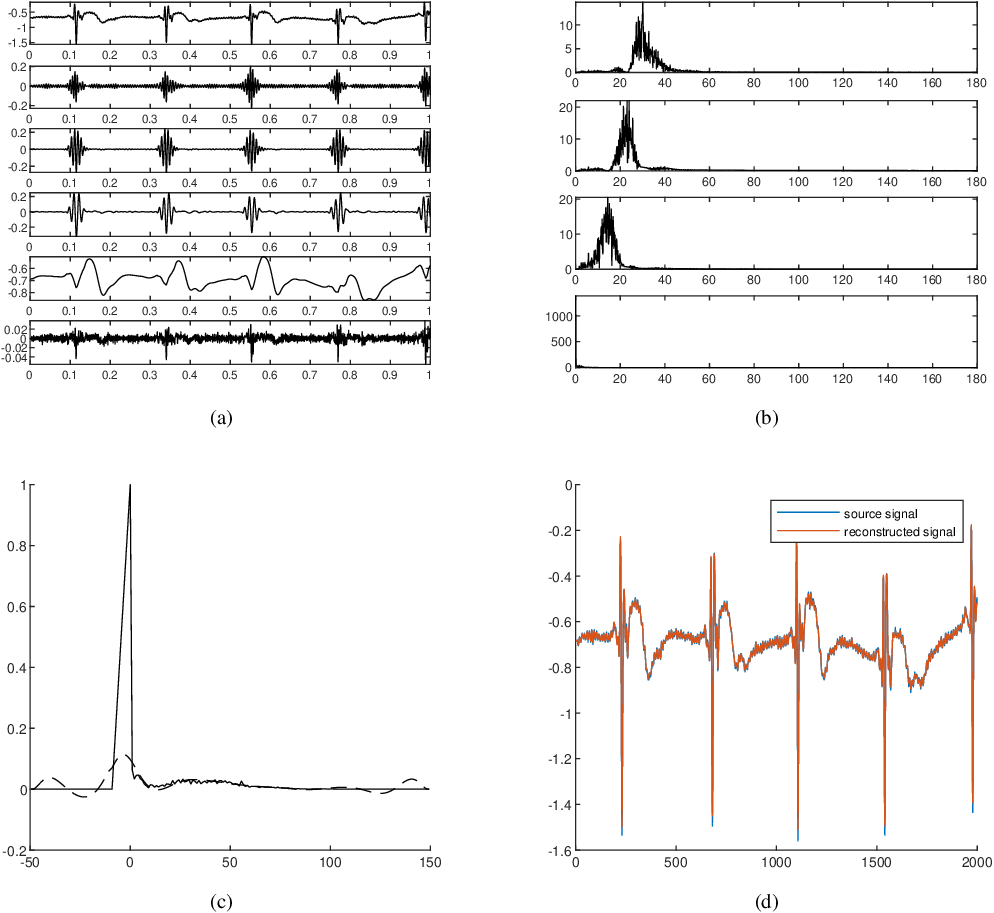}
\caption{The Decomposition result for V4 Lead of Record 124. (a) The original signal(top), the decomposed IMFs(middle 4), the residual(bottom), (b) The spectrum of each IMF, (c) The final converged cuttnig curve, (d) The source signal and the reconstruction signal.}
\label{test_case_ecg2}
\end{figure}

\begin{figure}
\centering
\includegraphics[scale=0.7]{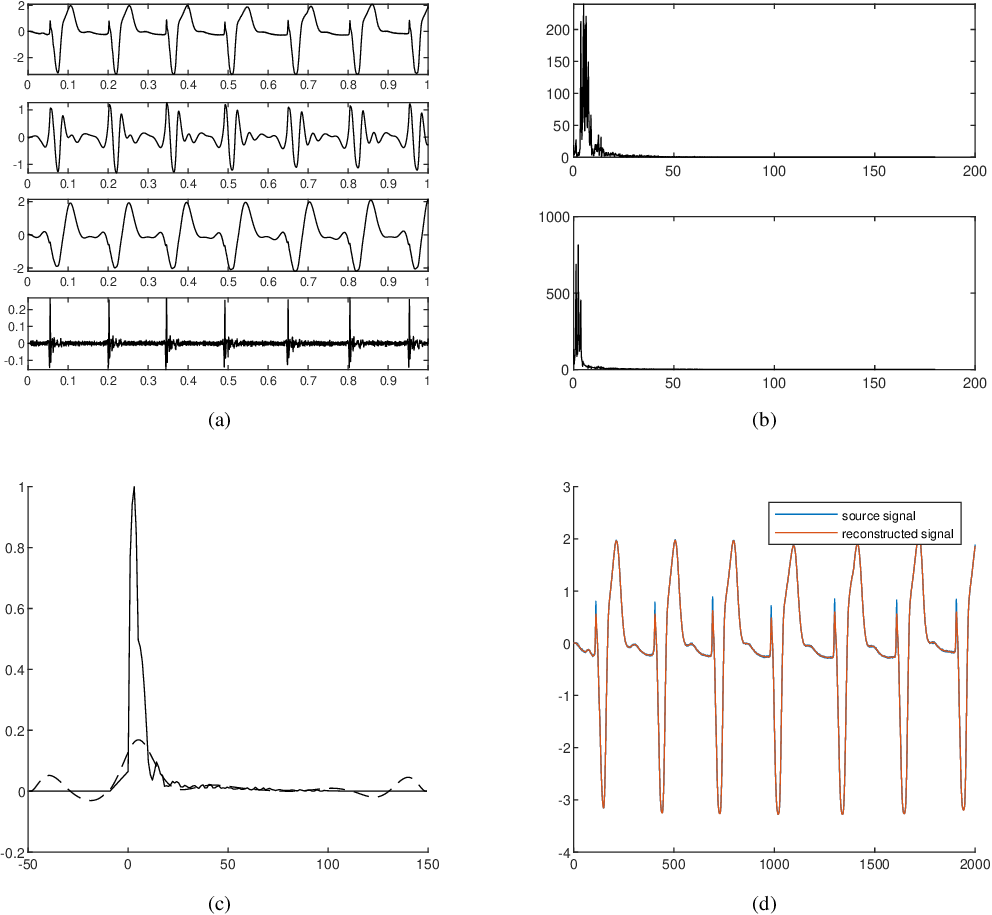}
\caption{The Decomposition result for V2 Lead of Record 102. (a) The original signal(top), the decomposed IMFs(middle 2), the residual(bottom), (b) The spectrum of each IMF, (c) The final converged cutting curve, (d) The source signal and the reconstruction signal.}
\label{test_case_ecg3}
\end{figure}

\begin{figure}
\centering
\includegraphics[scale=0.7]{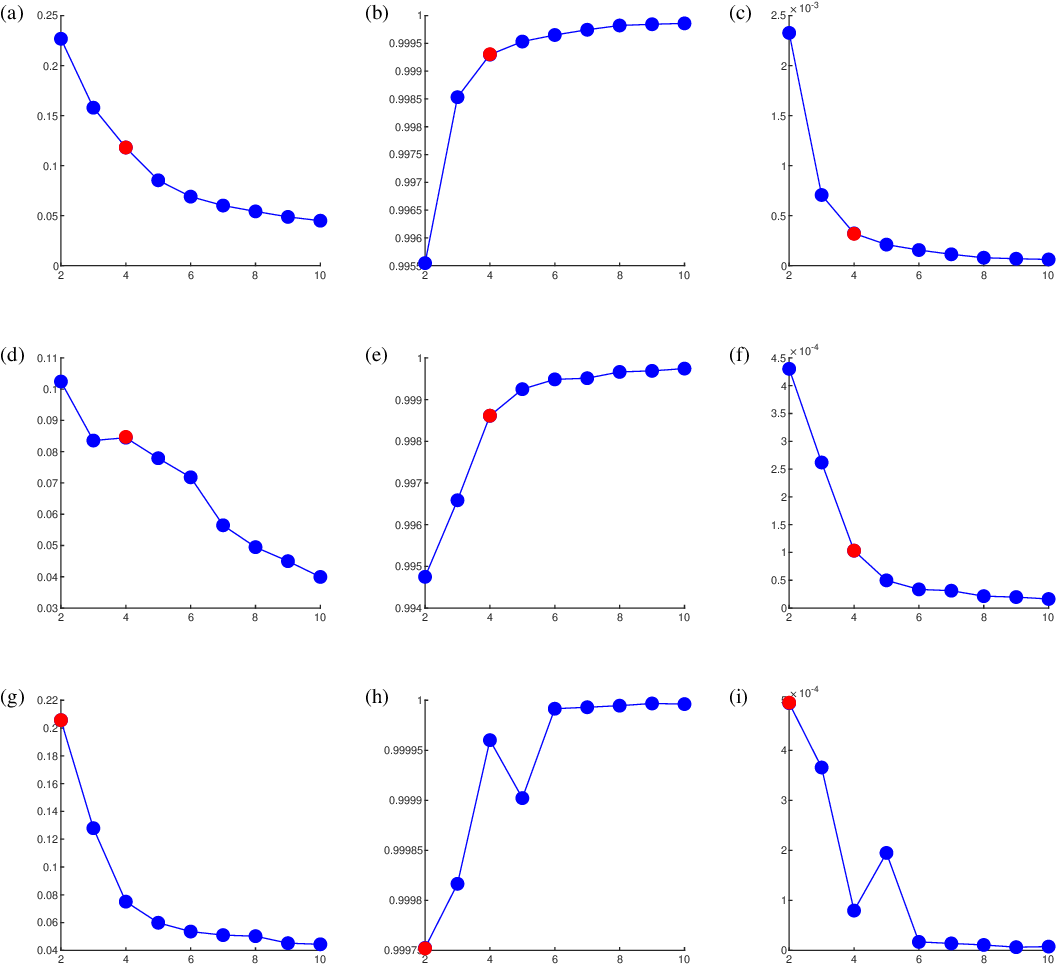}
\caption{The performance of decomposition of our chosen K and center frequencies for each ECG signal compared with random select K from 2 to 10 to proceed VMD.}
\label{vmd_ecg_performance}
\end{figure}

\begin{figure}
\centering
\includegraphics[scale=0.6]{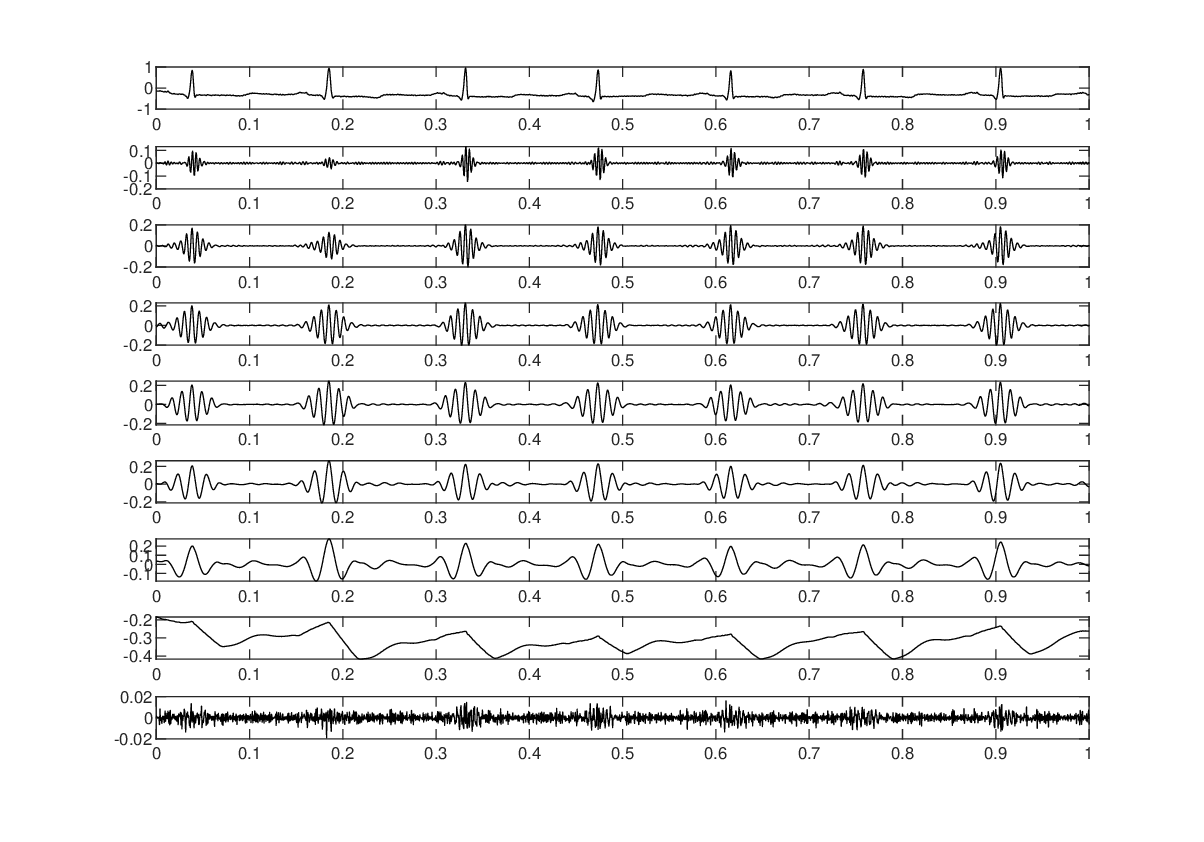}
\caption{The decomposition result on MLII Lead for Record 100 when K is set over our estimation (for example 7).}
\label{ecg1_mode_7}
\end{figure}

\begin{figure}
\centering
\includegraphics[scale=0.5]{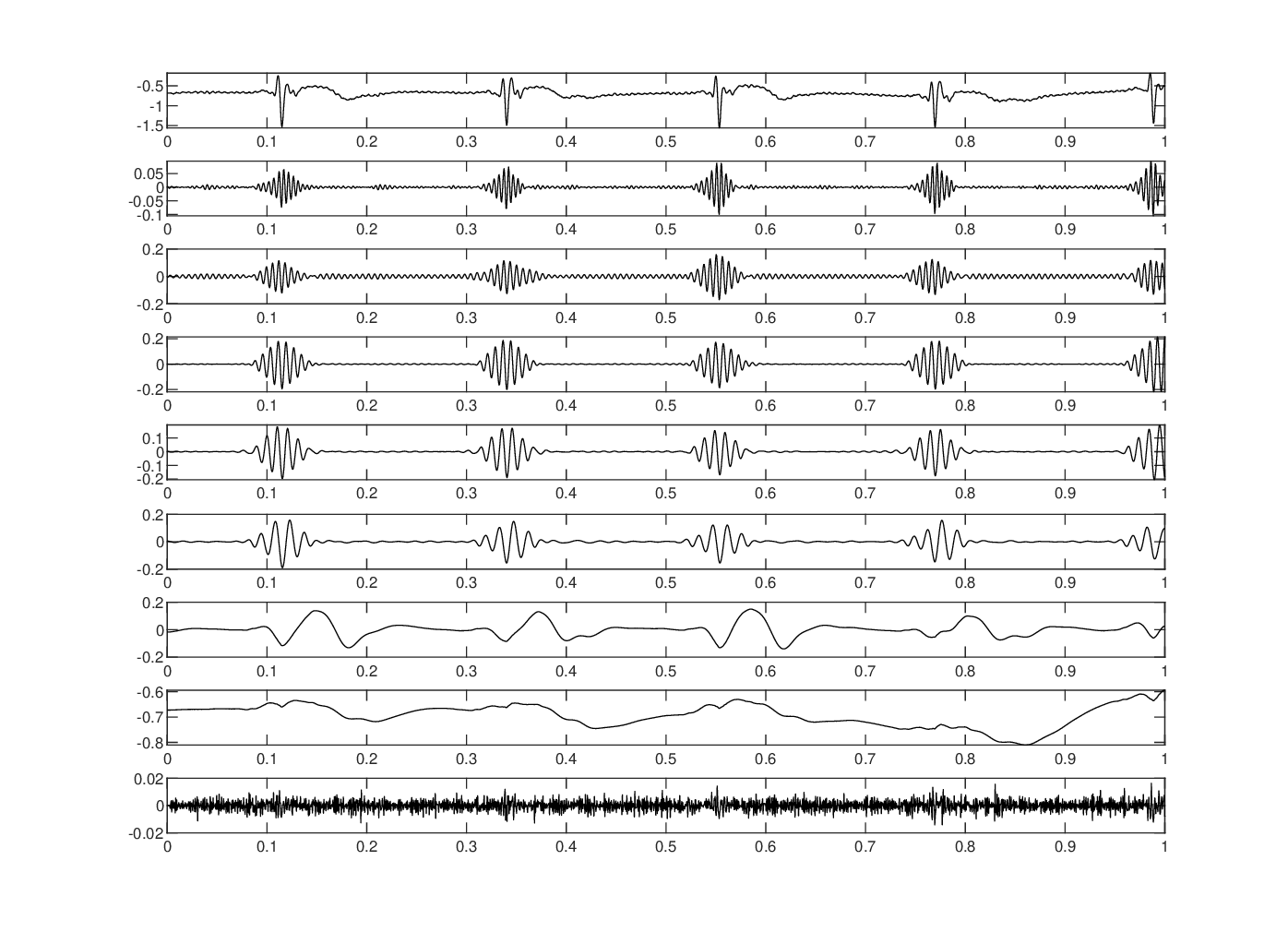}
\caption{The decomposition result on V4 Lead for Record 124 when K is set over our estimation (for example 7).}
\label{ecg2_mode_7}
\end{figure}

\begin{figure}
\centering
\includegraphics[scale=0.8]{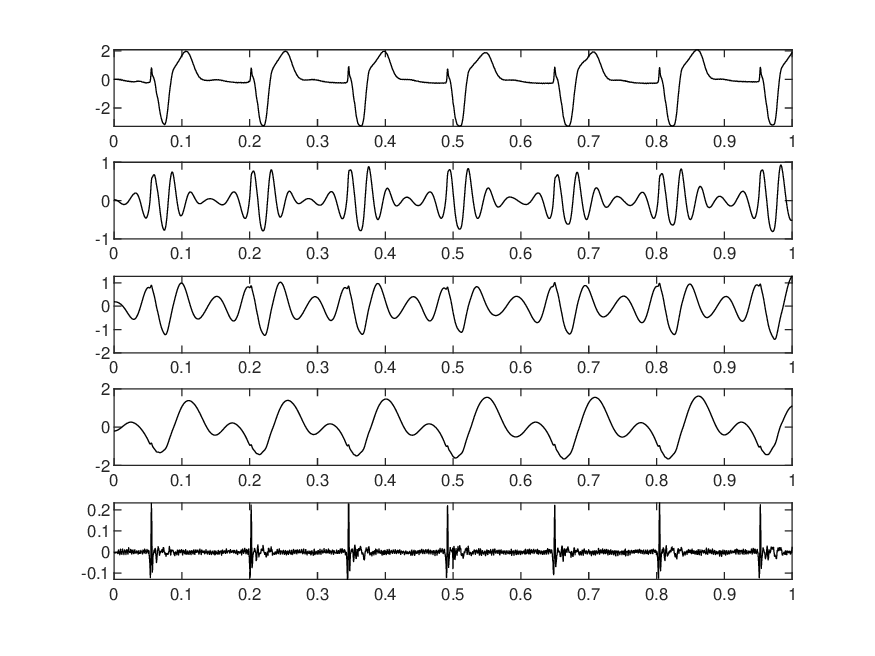}
\caption{The decomposition result on V2 Lead for Record 102 when K is set over our estimation (for example 3).}
\label{ecg3_mode_3}
\end{figure}

\clearpage
\section{Conclusion}
In this article, we address the long-standing open challenge of automatically determining the number of intrinsic mode functions (IMFs) and their corresponding center frequencies for Variational Mode Decomposition (VMD). We first establish a geometric insight into this modal initialization problem: the number of valid IMFs is inherently equivalent to the count of disconnected intervals of the spectrum above a critical cutting curve. Based on this insight, we formulate the problem within a discrete functional framework defined by the number of isolated intervals with respect to the cutting curve. To tackle the intractable non-continuous and non-differentiable nature of this discrete functional optimization, we introduce the cutting curve as a surrogate to transform the original problem into a continuous functional optimization problem. On this basis, we propose a globally convergent method for the automatic determination of IMF number and initial center frequencies, with rigorous mathematical proofs provided for the convexity, strong duality and global convergence of the proposed framework.

For practical engineering implementation, we further develop an efficient and high-precision numerical scheme for the proposed variational framework. Specifically, we convert the original optimal variational differential equation into a matrix form, and separate the unknown variables by introducing an extended Hadamard product with compatible broadcasting rules, which reformulates the problem into a standard system of linear equations. We also balance the magnitude order of elements in the coefficient matrix to reduce its condition number, which effectively avoids numerical ill-conditioning and enables efficient, high-precision numerical solving even for long-sequence signals.

We conduct comprehensive experiments on both synthetic narrowband signals and real-world electrocardiogram (ECG) signals. Experimental results demonstrate that our method can reliably extract physically meaningful IMFs from complex signals. It is validated that the proposed algorithm can select a reasonable number of modes to avoid over-decomposition while guaranteeing high signal reconstruction accuracy. Furthermore, the decomposed modal components derived from our method retain meaningful and discriminative features of the target waveform, which is critical for subsequent signal analysis and interpretation tasks.

This work provides a theoretically rigorous and practically robust initialization paradigm for adaptive VMD, with the potential to be extended to a wide range of signal processing scenarios.

\section{Acknowledgement}
This study is supported by the National Natural Science Foundation of China under Grant No. 61773290 and the Fundamental Research Funds for the Central Universities (22120230311) and Tongji University Medicine-X Interdisciplinary Research Initiative (Grant No. 2025-0554-YB-11).

\clearpage
\setcounter{section}{0}
\renewcommand{\thesection}{\arabic{section}}
\numberwithin{equation}{section}
\setcounter{equation}{0}
\renewcommand{\theequation}{S\arabic{section}.\arabic{equation}}
\setcounter{figure}{0}
\renewcommand{\thefigure}{S\arabic{section}.\arabic{figure}}
\setcounter{table}{0}
\renewcommand{\thetable}{S\arabic{section}.\arabic{table}}

\clearpage
\appendix
\section{Finite Discrete Expression and Their Matrices}
\label{Matrices}
Here we show finite discrete expression and the corresponding the conversion matrices with respect to (\ref{final}),(\ref{con}).   Since the finite difference expressions for $\alpha(x)$ and $g(x)$ are identical in form, the conversion matrices $\boldsymbol G^{(n)}$ and $\boldsymbol A^{(n)}$ also share the exactly the same structure 

\begin{equation}
g^{(1)}(x)\approx\cfrac {-g(x+2h)+8g(x+h)-8g(x-h)+g(x-2h)}{12h}
\end{equation}
\begin{align}
g^{(2)}(x) \approx \cfrac {-g(x+2h) + 16g(x+h) - 30g(x) + 16g(x-h) - g(x-2h)}{12h^2}
\end{align}
\begin{equation}
g^{(3)}(x)\approx \cfrac {g(x+2h)-2g(x+h)+2g(x-h)-g(x-2h)}{2h^3}
\end{equation}
\begin{align}
g^{(4)}(x)\approx \cfrac{g(x+2h)-4g(x+h)+6g(x)-4g(x-h)+g(x-2h)}{h^4} 
\end{align}

\begin{equation}
\boldsymbol{A}^{(1)}=\boldsymbol{G}^{(1)}=\cfrac 1 {12} \begin{bmatrix}
1 & -8 & 0 & 8 & -1 & 0 & \cdots & \cdots & \cdots & 0 \\
0 & 1 & -8 & 0 & 8 & -1 & 0 & \cdots & \cdots & 0 \\
\vdots &\vdots &\vdots &\vdots &\vdots &\vdots &\vdots &\vdots &\vdots &\vdots \\
0 & \cdots & \cdots & 0 & 1 & -8 & 0 & 8 & -1 & 0 \\
0 & 0 & \cdots & \cdots & 0 & 1 & -8 & 0 & 8 & -1 
\end{bmatrix}
\end{equation}
\begin{equation}
\boldsymbol{A}^{(2)}=\boldsymbol{G}^{(2)}=\cfrac 1 {12} \begin{bmatrix}
-1 & 16 & -30 & 16 & -1 & 0 & \cdots & \cdots & \cdots & 0 \\
0 & -1 & 16 & -30 & 16 & -1 & 0 & \cdots & \cdots & 0 \\
\vdots &\vdots &\vdots &\vdots &\vdots &\vdots &\vdots &\vdots &\vdots &\vdots \\
0 & \cdots & \cdots & 0 & -1 & 16 & -30 & 16 & -1 & 0 \\
0 & 0 & \cdots & \cdots & 0 & -1 & 16 & -30 & 16 & -1
\end{bmatrix}
\end{equation}
\begin{equation}
\boldsymbol{A}^{(3)}=\boldsymbol{G}^{(3)}=\cfrac 1 2 \begin{bmatrix}
-1 & 2 & 0 & -2 & 1 & 0 & \cdots & \cdots & \cdots & 0 \\
0 & -1 & 2 & 0 & -2 & 1 & 0 & \cdots & \cdots & 0 \\
\vdots &\vdots &\vdots &\vdots &\vdots &\vdots &\vdots &\vdots &\vdots &\vdots \\
0 & \cdots & \cdots & 0 & -1 & 2 & 0 & -2 & 1 & 0 \\
0 & 0 & \cdots & \cdots & 0 & -1 & 2 & 0 & -2 & 1 
\end{bmatrix}
\end{equation}
\begin{equation}
\boldsymbol{A}^{(4)}=\boldsymbol{G}^{(4)}=\begin{bmatrix}
1 & -4 & 6 & -4 & 1 & 0 & \cdots & \cdots & \cdots & 0 \\
0 & 1 & -4 & 6 & -4 & 1  & 0 & \cdots & \cdots & 0 \\
\vdots &\vdots &\vdots &\vdots &\vdots &\vdots &\vdots &\vdots &\vdots &\vdots \\
0 & \cdots & \cdots & 0 & 1 & -4 & 6 & -4 & 1  & 0 \\
0 & 0 & \cdots & \cdots & 0 & 1 & -4 & 6 & -4 & 1 
\end{bmatrix}
\end{equation} 
and after evaluating $(2\pmb A^{(2)}\pmb \alpha) \odot \pmb G^{(2)}+(4\pmb A^{(1)}\pmb \alpha)\odot \pmb G^{(3)} +2\pmb \alpha \odot \pmb G^{(4)}$, we introduce another 4 lines depicting the boundary value condition, noted as $\pmb B$:
\begin{equation}
\pmb B = \begin{bmatrix}
1 & 0 & 0 & \cdots & \cdots & 0 & 0 & 0 \\
0 & 0 & 0 & \cdots & \cdots & 0 & 0 & 1 \\
-1 & 1 & 0 & \cdots & \cdots & 0 & 0 & 0 \\
0 & 0 & 0 & \cdots & \cdots & 0 & -1 & 1
\end{bmatrix}
\end{equation}
and (\ref{final}) is then transformed into
\begin{equation}
\begin{bmatrix}
(2\pmb A^{(2)}\pmb \alpha) \odot \pmb G^{(2)}+(4\pmb A^{(1)}\pmb \alpha)\odot \pmb G^{(3)} +2\pmb \alpha \odot \pmb G^{(4)} \\
\pmb B
\end{bmatrix}\pmb g = \pmb c
\end{equation}
in which 
\begin{equation}
\pmb \alpha=\begin{bmatrix}
\alpha_0\\
\alpha_1\\
\vdots \\
\alpha_{N-1} \\
\alpha_N
\end{bmatrix},
\pmb g=\begin{bmatrix}
g_0\\
g_1\\
\vdots \\
g_{N-1} \\
g_N
\end{bmatrix} \text{and} \ \ \pmb c=\begin{bmatrix}
h^4c_2 \\
h^4c_3 \\
\vdots \\
h^4c_{N-3} \\
h^4c_{N-2} \\
c_0 \\
c_N \\
c_1-c_0 \\
c_N-c_{N-1}
\end{bmatrix}
\end{equation}
and $c_i=\beta_i-\lambda_i+\mu_i$.  In this way the condition number of the matrix to be inverted during evaluation is more gentle.

\section{Supplementary Material of the Proof}
\label{Proof}
\subsection{Proof of Convexity}
\label{Proof_of_convexity}
Let $\mathcal O(g(x))= \mathcal O_1(g(x))+\mathcal O_2(g(x))$,  where \\$\mathcal O_1(g(x))=\displaystyle \int_\Omega \alpha(x)g''^2(x)\text dx$ and $\mathcal O_2(g(x))=\displaystyle-\int_\Omega \beta(x)g(x) \text dx$.  For $\mathcal O_1(g(x))$, we have
\begin{align}
\lambda \mathcal O_1(g_1(x))+(1-\lambda)J_1(g_2(x))- \mathcal O_1[\lambda g_1(x)+(1-\lambda)g_2(x)]\\\notag=\lambda(1-\lambda)\int_\Omega \alpha(x)[g_1''(x)-g_2''(x)]^2\text dx \ge 0
\end{align}
and $\mathcal O_2(g(x))$ is obviously a linear funtional, so we have
\begin{equation}
\mathcal O[\lambda g_1(x) + (1-\lambda) g_2(x)] \le \lambda \mathcal O[g_1(x)] + (1-\lambda) \mathcal O[g_2(x)]
\label{convexity}
\end{equation}
so that our objective functional is convex.

\subsection{Concavity of Dual Functional $\mathcal G$}
\label{Proof_of_dual}
In this part we give a short proof to the concavity of $\mathcal G$.  According to the definition of $\mathcal G$, we have
\begin{equation}
\begin{split}
&\ \ \ \ \mathcal G(\theta\lambda_1(x)+(1-\theta)\lambda_2(x), \theta\mu_1(x)+(1-\theta)\mu_2(x))\\
&=\inf\limits_{g(x) \in \mathcal D} \mathcal L[g(x),\theta\lambda_1(x)+(1-\theta)\lambda_2(x),\theta\mu_1(x)+(1-\theta)\mu_2(x)] \\
&=\inf\limits_{g(x) \in \mathcal D} \mathcal L[\theta g(x) + (1-\theta)g(x),\theta\lambda_1(x)+(1-\theta)\lambda_2(x),\theta\mu_1(x)+(1-\theta)\mu_2(x)]\\
&=\inf\limits_{g(x) \in \mathcal D} \{\theta\mathcal L[g(x), \lambda_1(x), \mu_1(x)]+(1-\theta) \mathcal L[(g(x), \lambda_2(x), \mu_2(x)]\} \\
&\ge \theta \inf\limits_{g(x) \in \mathcal D} \mathcal L[g(x), \lambda_1(x), \mu_1(x)] + (1-\theta)\inf\limits_{g(x) \in \mathcal D} \mathcal L[g(x), \lambda_2(x), \mu_2(x)] \\
&=\theta\mathcal G(\lambda_1(x), \mu_1(x)) + (1-\theta)\mathcal G(\lambda_2(x), \mu_2(x))
\end{split}
\end{equation}
which concludes the concavity of $\mathcal G$.

\subsection{Proof of the Zero Duality Gap}

Now we begin to prove that the duality gap is zero.  Let 
\begin{equation}
\mathcal{O}(g(x))=\displaystyle \int_\Omega \alpha(x)g''^2(x)\text dx - \int_\Omega \beta(x)g(x) \text dx
\end{equation}
 be the objective functional, and construct Lagrangian functional
\begin{align}
\mathcal{L}(g(x), \lambda(x), \mu(x))=\displaystyle \int_\Omega \alpha(x)g''^2(x)\text dx &- \int_\Omega \beta(x)g(x) \text dx \notag \\&+ \int_\Omega \lambda(x)[g(x)-f(x)] \text dx -\int_\Omega \mu(x)g(x)\text dx
\end{align}
define dual functional as
\begin{equation}
\mathcal{G}(\lambda(x), \mu(x)) = \inf\limits_{g(x)\in\mathcal{D}}\mathcal{L}(g(x),\lambda(x),\mu(x))
\label{define}
\end{equation}
then on one side, we must have
\begin{align}
\mathcal{G}(\lambda^*(x),\mu^*(x)) =\inf\limits_{g(x)}\mathcal L(g(x), \lambda^*(x), \mu^*(x)) \le\inf\limits_{g(x)} \mathcal{O}(g(x))=\mathcal{O}(g^*(x))
\label{le_equation}
\end{align}
where $*$ stands for the optimum.  On the other hand, one can construct general sets
\begin{align}
\mathcal{S}=\{(p(x), q(x), o)|g(x)-f(x) \le p(x),-g(x) \le q(x), \mathcal O(g(x)) \le o\}
\label{pqo}
\end{align}
which forms an epi-graph to (\ref{std}).
thus
\begin{equation}
\forall \xi(x)=\lambda(p_1(x), q_1(x), o_1) + (1-\lambda)(p_2(x), q_2(x), o_2) 
\end{equation}
one can always find
\begin{equation}
\eta(x)=\lambda g_1(x)+(1-\lambda)g_2(x)
\end{equation}
that makes 
\begin{equation}
  \begin{cases}
  \eta(x)-f(x) \le \lambda p_1(x) + (1-\lambda) p_2(x) \\
  -\eta(x) \le \lambda q_1(x) + (1-\lambda) q_2(x) \\
  \mathcal{O}(\eta(x))\le\lambda\mathcal{O}(g_1(x))+(1-\lambda)\mathcal{O}(g_2(x))\le \lambda o_1 + (1-\lambda) o_2 \\
  \end{cases}
\end{equation}
hold and thus $\mathcal{S}$ is a convex set.  It is obvious that $(0, 0,\mathcal{O}(g^*(x))) \in \partial S$, so according to the Support Theorem of Convex Sets there must be a general hyperplane that on one side pass through $(0, 0,\mathcal{O}(g^*(x))) \in \partial S$ and on the other side ensures any point in $\mathcal S$ locates above the hyperplane.  That is, there exists $\lambda(x), \mu(x)$ that makes
\begin{equation}
\int_\Omega \lambda(x)p(x) \text dx +\int_\Omega \mu(x)q(x)\text dx \ge \mathcal{O}(g^*(x)) - o
\label{assert}
\end{equation} hold for all $(p(x),q(x),o)$ in $\mathcal S$.  Note that for $p(x), q(x) \ge 0$, there exist points in $\mathcal S$ that makes $o \le \mathcal O(g^*(x))$, so there exists points that makes the right side of (\ref{assert}) non-negative, this forces $\lambda(x) \ge 0$ and $\mu(x) \ge 0$, which indicates we can find a solution in dual space making the dual-gap vanish.

Next we prove the convexity of our feasible domain.  The feasible set is $0 \le g(x) \le f(x)$, so that $\forall g_1(x), g_2(x) \in [0,f(x)]$, we have $\lambda g_1(x)+(1-\lambda)g_2(x) \in [0,f(x)], \forall \lambda \in [0,1]$, which means the feasible set is also convex.

\begin{figure}[H]
  \centering
  \subfloat[]{
    \includegraphics[width=0.3\columnwidth]{converge_test0_1000.eps}
  }
  \subfloat[]{
    \includegraphics[width=0.3\columnwidth]{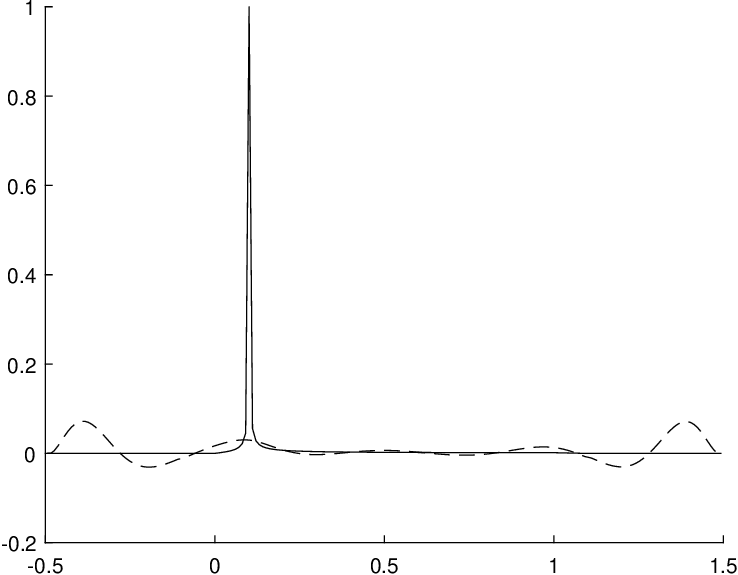}
  }
  \subfloat[]{
  	\includegraphics[width=0.3\columnwidth]{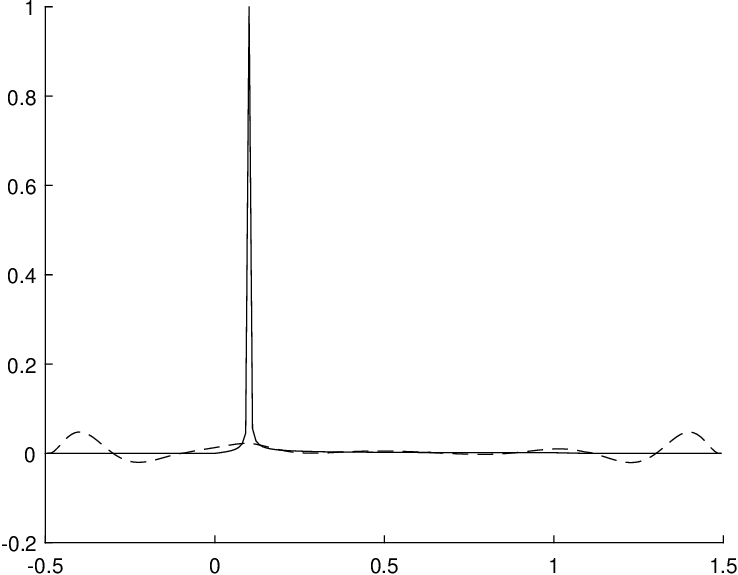}
  }
  \vspace{3pt}
  \subfloat[]{
  	\includegraphics[width=0.3\columnwidth]{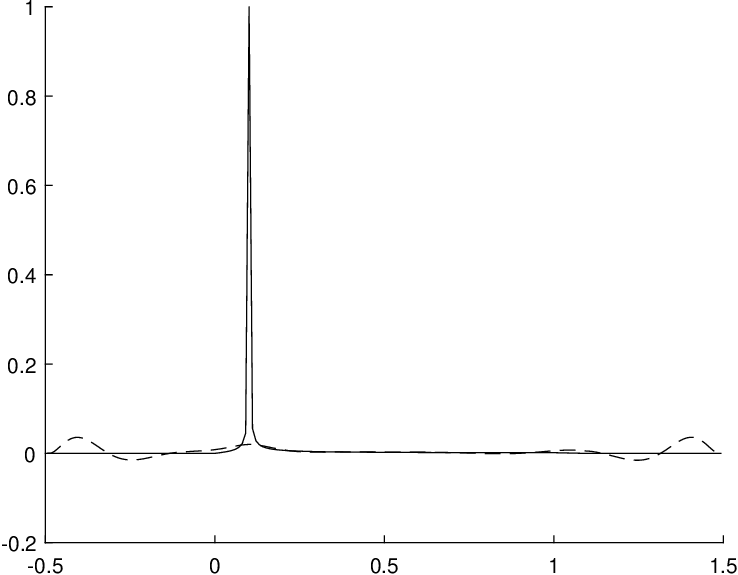}
  }
  \subfloat[]{
   \includegraphics[width=0.3\columnwidth]{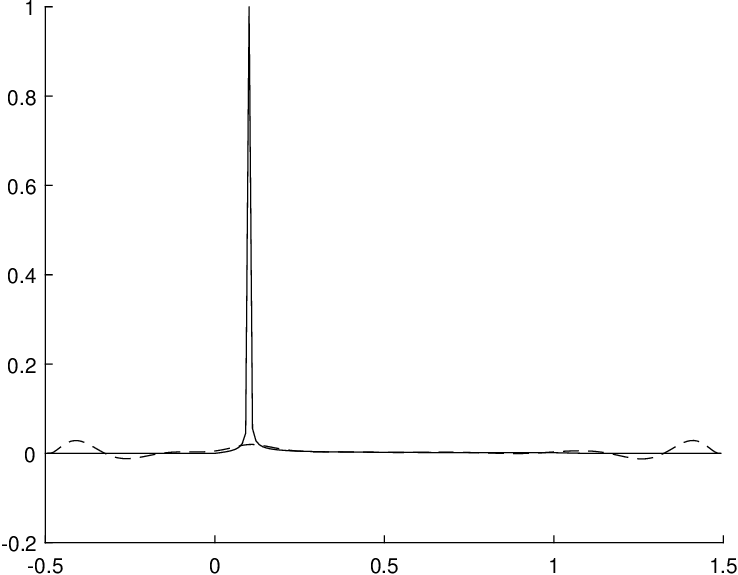}
  }
  \caption{The iteration for finding the cutting curve of FFT with respect to time series $y(t)=100\sin 20\pi t$.  The 1000th(a), 2000th(b), 3000th(c), 4000th(d), 5000th(e) iteration are respectively shown.}
  \label{diag4}
\end{figure}

\section{Experiment Detail}
In this section we begin to provide some details of our experiments as complementary.
\label{Experiment}

\begin{figure}[H]
  \centering
  \subfloat[]{
    \includegraphics[trim=0 0 0 3, clip,width=0.4\columnwidth]{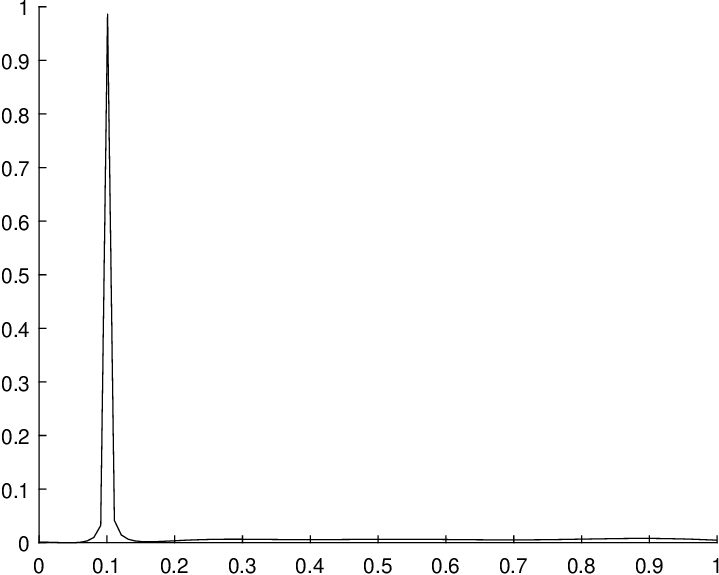}
  }
  \subfloat[]{
    \includegraphics[trim=0 0 0 3, clip,width=0.4\columnwidth]{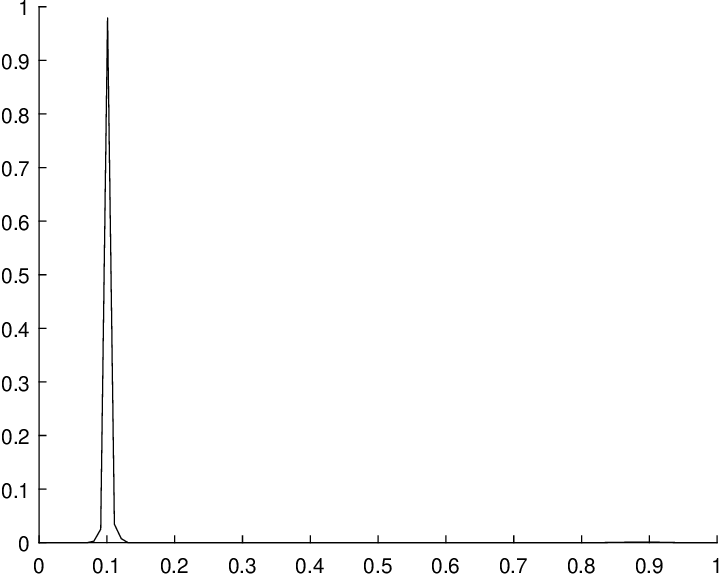}
  }
  \caption{The gap-elimination course for the found spectrum supporting line with series $y(t)=100\sin 20\pi t$.  (a) is the result of source spectrum subtracting found cutting curve, and one can see that there is some additional noise-gap in the bottom of line subtracted function. (b) is the result after our KDE process to eliminate the gap at the bottom .}
  \label{diag5}
\end{figure}
\vspace*{-25pt}
\begin{figure}[H]
\centering
	\includegraphics[width=0.7\columnwidth]
	{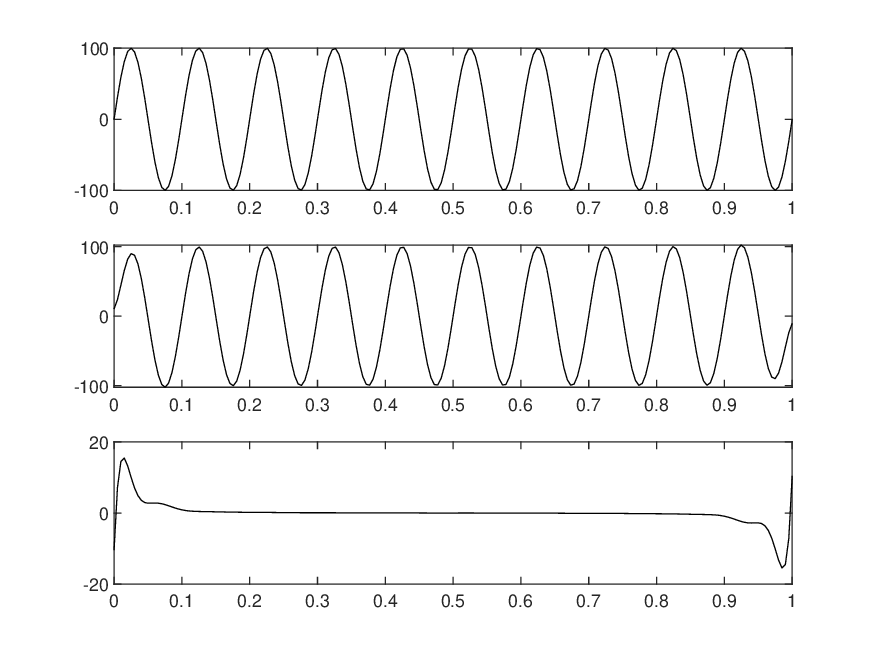}
\caption{The decomposition result of $y(t)=100\sin 20\pi t$ from the subsequent VMD, using the found number of modes and center frequency as initialization.  The subfigure from top to bottom are respective the source signal, the decomposed mode(only one) and the residual.}
\label{diag6}
\end{figure}

\begin{figure}[H]
  \centering
  \subfloat[]{
    \includegraphics[width=0.3\columnwidth]{converge_test1_1000.eps}
  }
  \subfloat[]{
    \includegraphics[width=0.3\columnwidth]{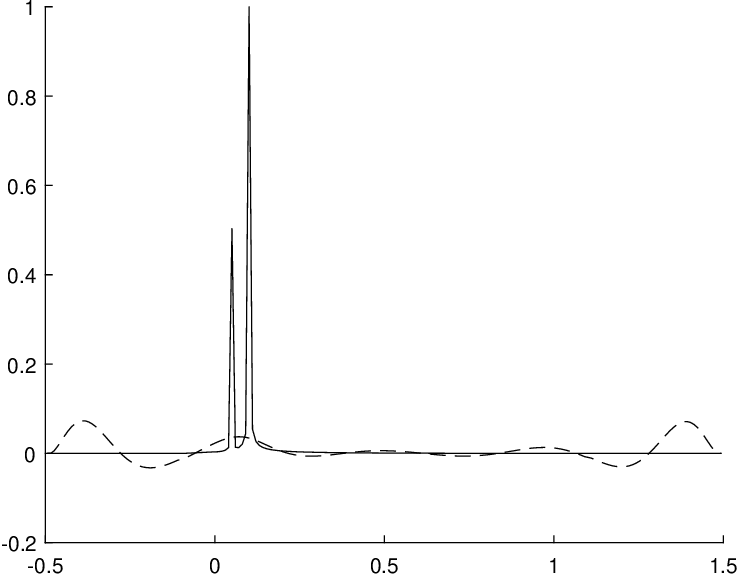}
  }
  \subfloat[]{
  	\includegraphics[width=0.3\columnwidth]{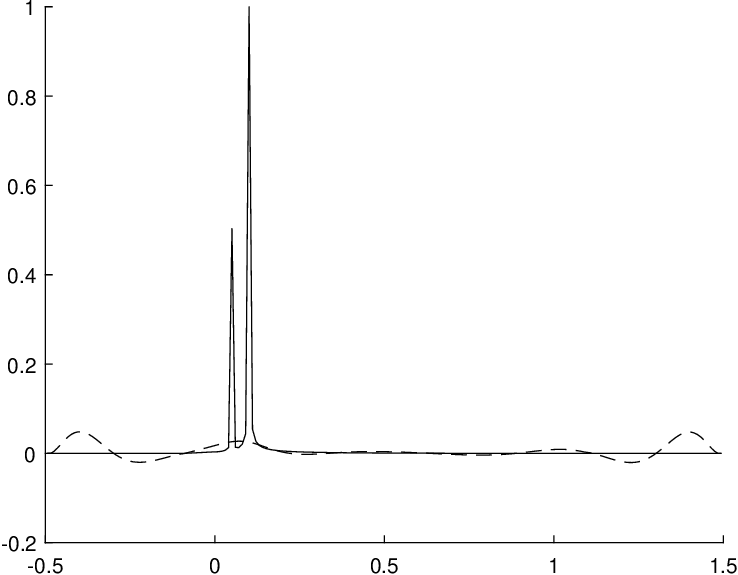}
  }
  \vspace{3pt}
  \subfloat[]{
  	\includegraphics[width=0.3\columnwidth]{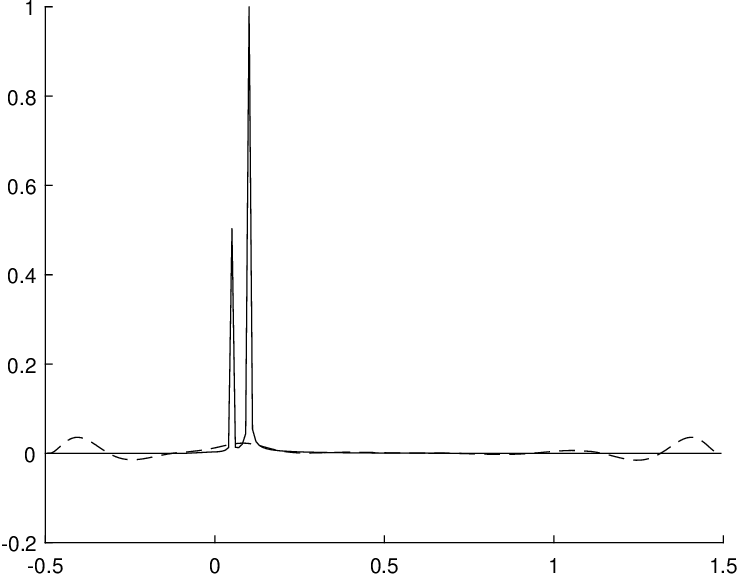}
  }
  \subfloat[]{
  	\includegraphics[width=0.3\columnwidth]{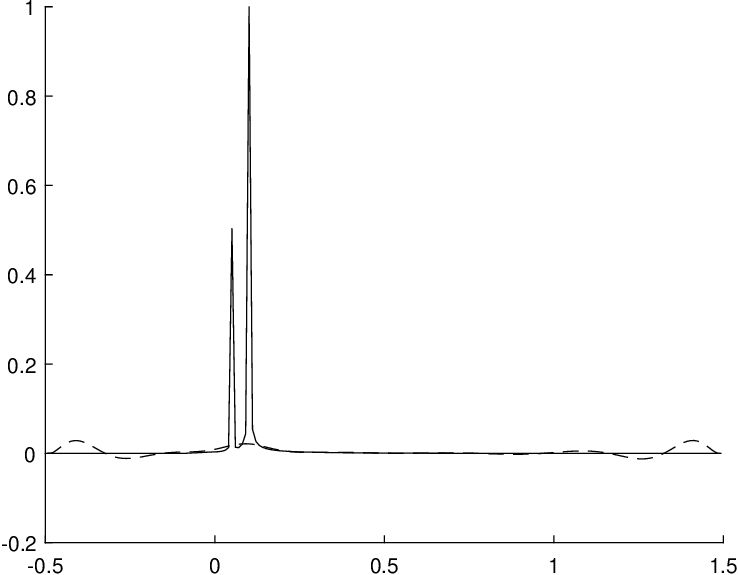}
  }
  \caption{The iteration for finding the cutting curve of FFT with respect to time series  $y(t)=10\cos(10\pi t)+20\sin(20\pi t)$.  The 1000th(a), 2000th(b), 3000th(c), 4000th(d), 5000th(e) iterations are respectively shown.}
  \label{diag7}
\end{figure}
\vspace*{-25pt}
\begin{figure}[H]
  \centering
  \subfloat[]{
    \includegraphics[width=0.45\columnwidth]{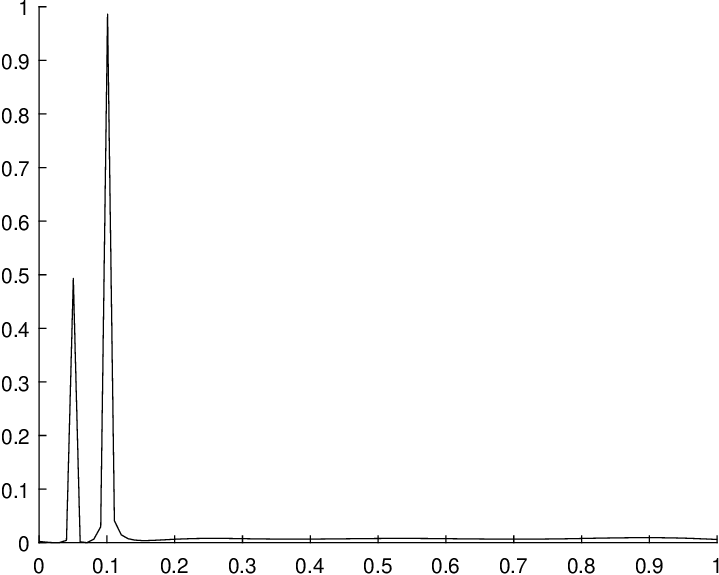}
  }
  \hfill
  \subfloat[]{
    \includegraphics[width=0.45\columnwidth]{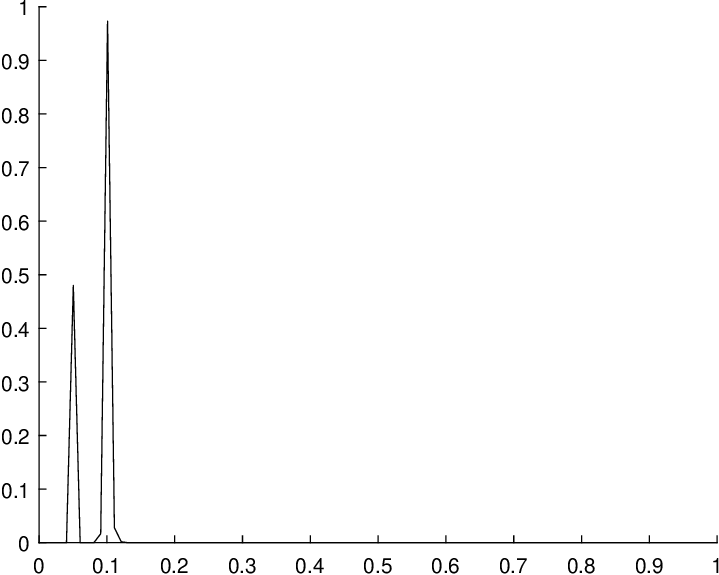}
  }

  \caption{The gap-elimination course for the found spectrum supporting line with series $y(t)=10\cos 10\pi t+20\sin 20\pi t$.  (a) is the result of source spectrum subtracting found cutting curve, and one can see that there is some additional noise-gap in the bottom of line subtracted function. (b) is the result after our KDE process to eliminate the gap at the bottom .}
  \label{diag8}
\end{figure}

\begin{figure}[H]
\centering
	\includegraphics[width=0.9\columnwidth]
	{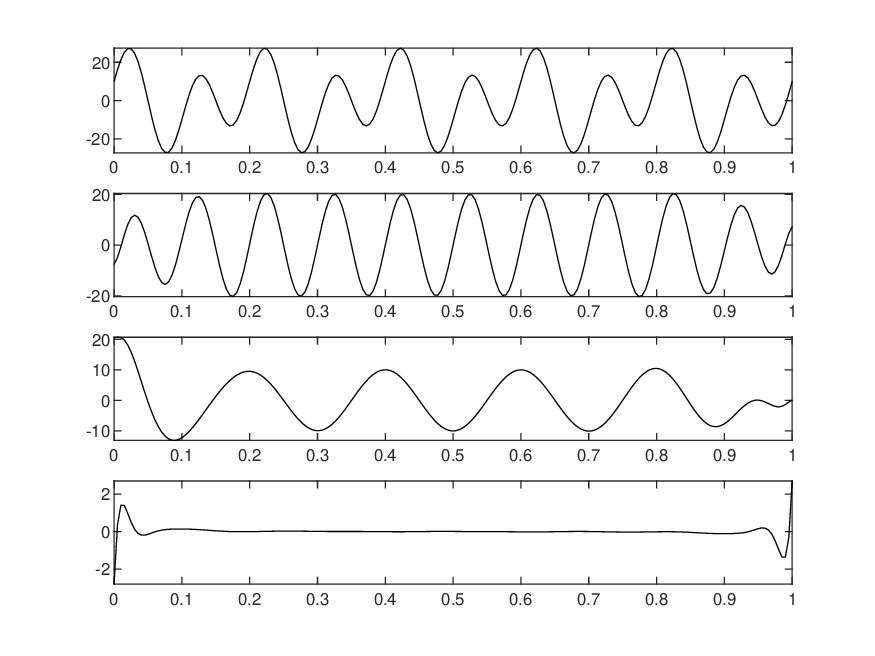}
\caption{The decomposition result of $y(t)=10\cos 10\pi t + 20\sin 20\pi t$ from the subsequent VMD, using the found number of modes and center frequency as initialization.  The sub-figure from top to bottom are respectively the source signal, the decomposed modes(there are two) and the residual.}
\label{diag9}
\end{figure}

\begin{figure}[H]
  \centering
  \subfloat[]{
    \includegraphics[width=0.3\columnwidth]{converge_test2_1000.eps}
  }
  \subfloat[]{
    \includegraphics[width=0.3\columnwidth]{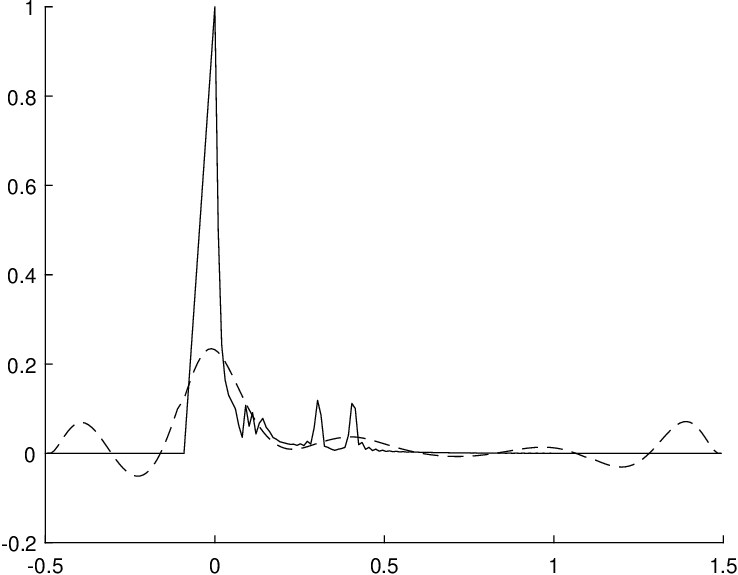}
  }
  \subfloat[]{
    \includegraphics[width=0.3\columnwidth]{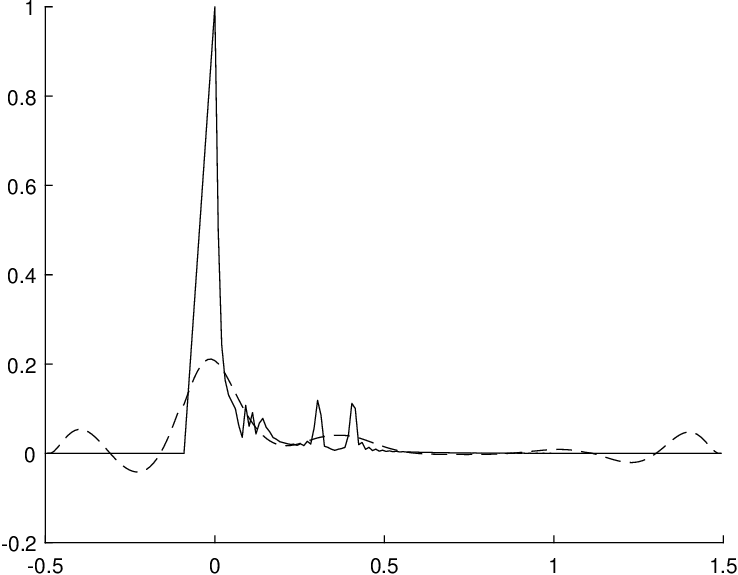}
  }
  \caption{The iteration for finding the cutting curve of FFT with respect to time series \eqref{seg_func}.  The 1000th(a), 2000th(b), 3000th(c) iterations are respectively shown.  Note that although the spectrum of second mode, $\cos 10\pi t$, is affected by that of the first mode, $6t^2$, the computed cutting curve precisely captured the trend of the lower bound so that all the modes can be filtered out.}
  \label{diag10}
\end{figure}
\begin{figure}[H]
  \centering
  \subfloat[]{
    \includegraphics[width=0.45\columnwidth]{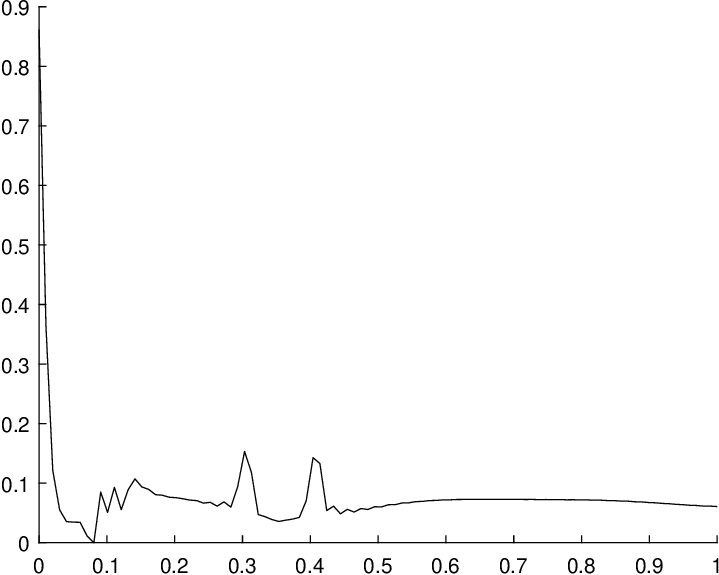}
  }
  \subfloat[]{
    \includegraphics[width=0.45\columnwidth]{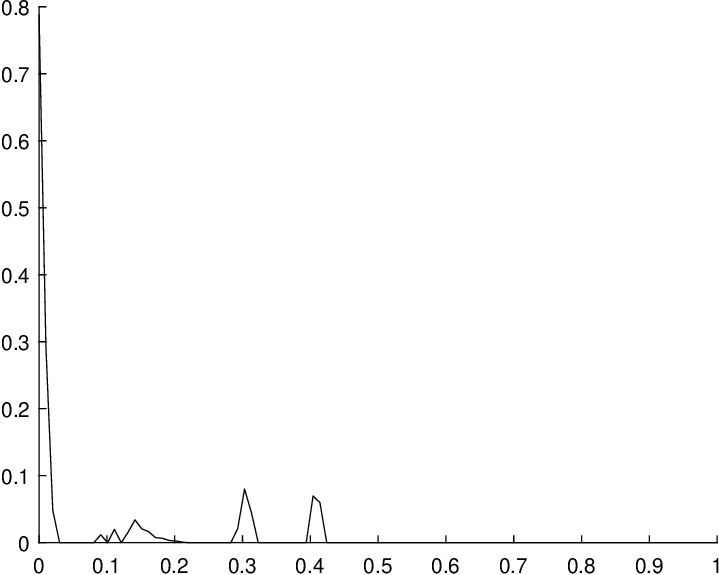}
  }

  \caption{The gap-elimination course for the found spectrum supporting line with series  \eqref{seg_func}.  (a) is the result of source spectrum subtracting found cutting curve, and one can see that there is some additional noise-gap in the bottom of line subtracted function. (b) is the result after our KDE process to eliminate the gap at the bottom.}
  \label{diag11}
\end{figure}
\vspace*{-50pt}

\begin{figure}[H]
\centering
	\includegraphics[width=\columnwidth]
	{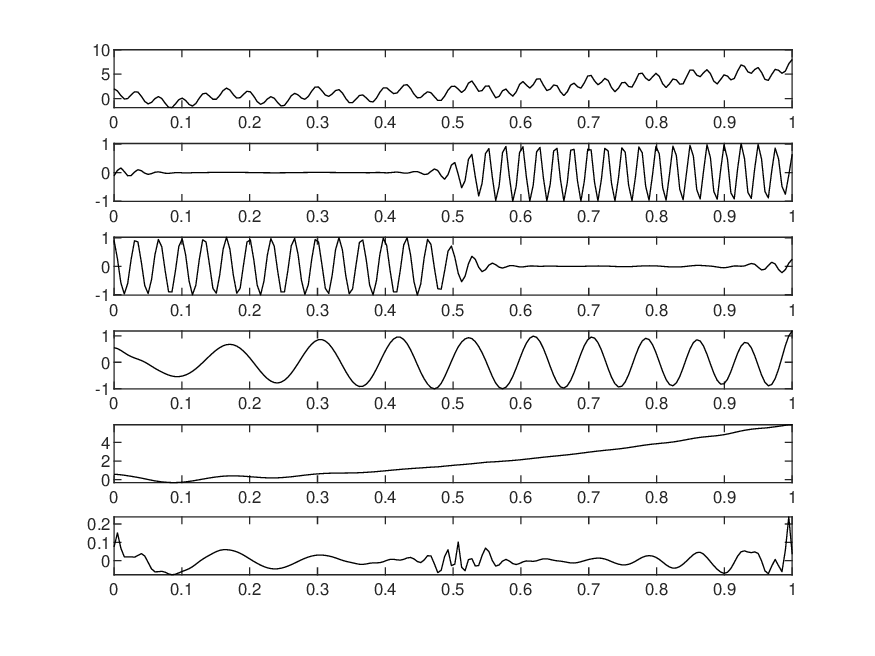}
\caption{The decomposition result of \eqref{seg_func} from the subsequent VMD, using the found number of modes and center frequency as initialization.  The sub-figure from top to bottom are respectively the source signal, the decomposed modes(there are four) and the residual.}
\label{diag12}
\end{figure}

\begin{figure}[H]
  \centering
  \subfloat[]{
    \includegraphics[width=0.3\columnwidth]{converge_test4_1000.eps}
  }
  \subfloat[]{
    \includegraphics[width=0.3\columnwidth]{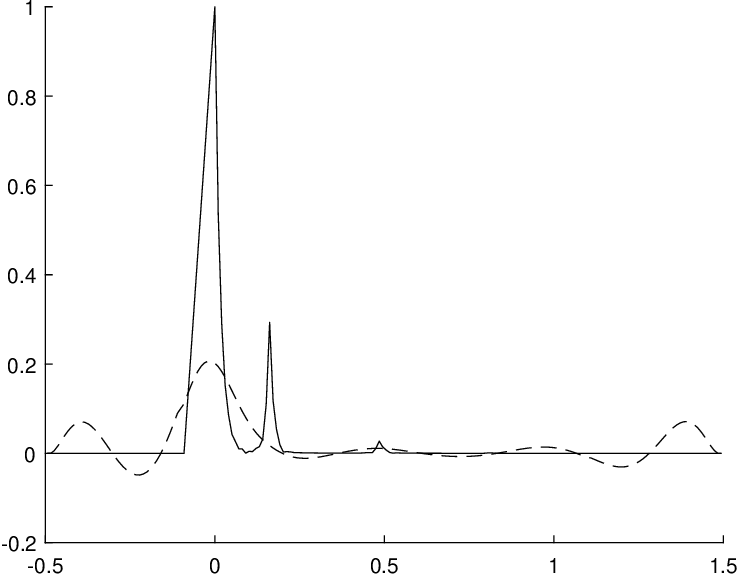}
  }
  \subfloat[]{
  	\includegraphics[width=0.3\columnwidth]{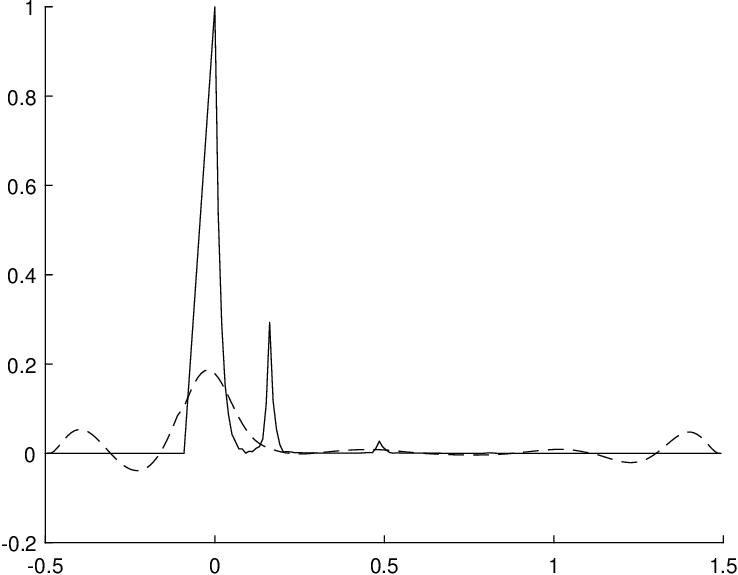}
  }
  \caption{The iterations for finding the cutting curve of FFT with respect to time series $\cfrac 1 {1.2+\cos 2\pi t} + \cfrac {\cos (32\pi t +0.2\cos 64\pi t)}{1.5+\sin(2\pi t)}$ are shown.  The 1000th(a), 2000th(b), 3000th(c) iterations are respectively shown.}
\end{figure}

\begin{figure}[H]
  \centering
  \subfloat[]{
    \includegraphics[width=0.45\columnwidth]{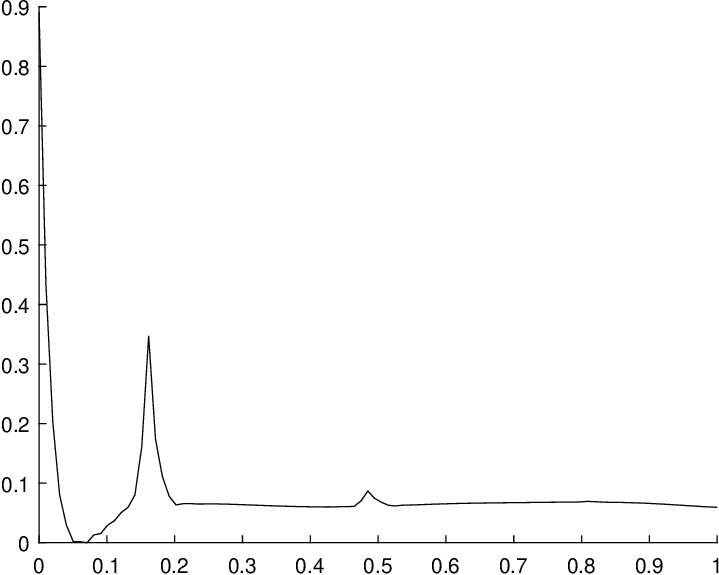}
  }
  \subfloat[]{
    \includegraphics[width=0.45\columnwidth]{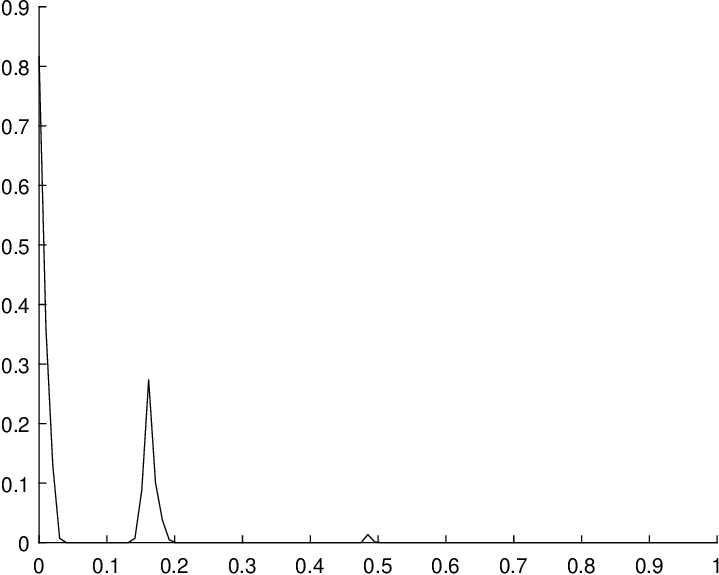}
  }

  \caption{The noise-elimination course for the found spectrum supporting line with series $\cfrac 1 {1.2+\cos 2\pi t} + \cfrac {\cos (32\pi t +0.2\cos 64\pi t)}{1.5+\sin(2\pi t)}$.  (a) is the result of source spectrum subtracting found cutting curve, and one can see that there is some additional noise-gap in the bottom of line subtracted function. (b) is the result after our KDE process to eliminate the gap at the bottom .}
  \label{diag14}
\end{figure}

\begin{figure}[H]
\centering
	\includegraphics[width=\columnwidth]
	{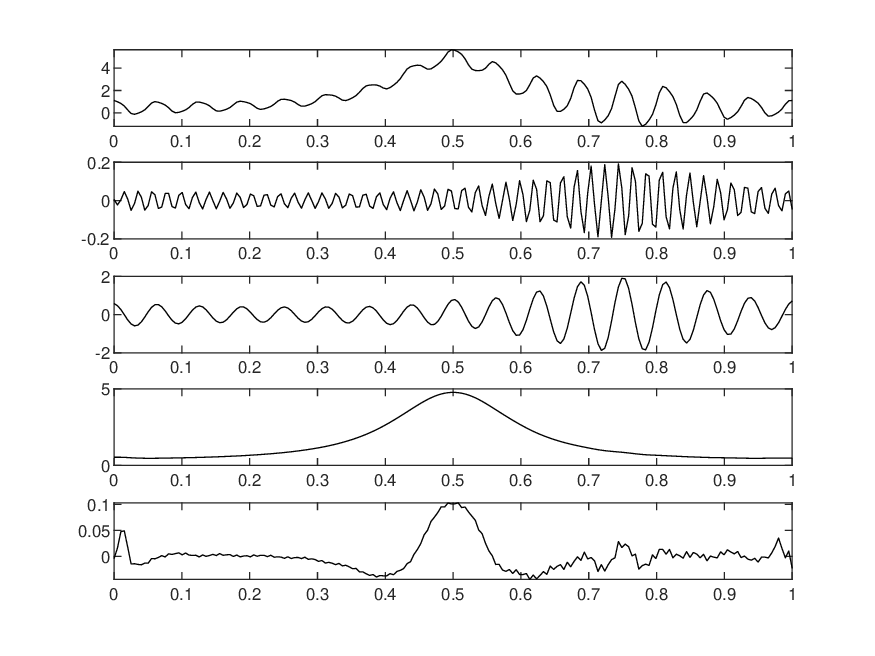}
\caption{The decomposition result of $\cfrac 1 {1.2+\cos 2\pi t} + \cfrac {\cos (32\pi t +0.2\cos 64\pi t)}{1.5+\sin(2\pi t)}$ from the subsequent VMD, using the found number of modes and center frequency as initialization.}
\label{diag15}
\end{figure}

\begin{figure}[H]
  \centering
  \subfloat[]{
    \includegraphics[width=0.3\columnwidth]{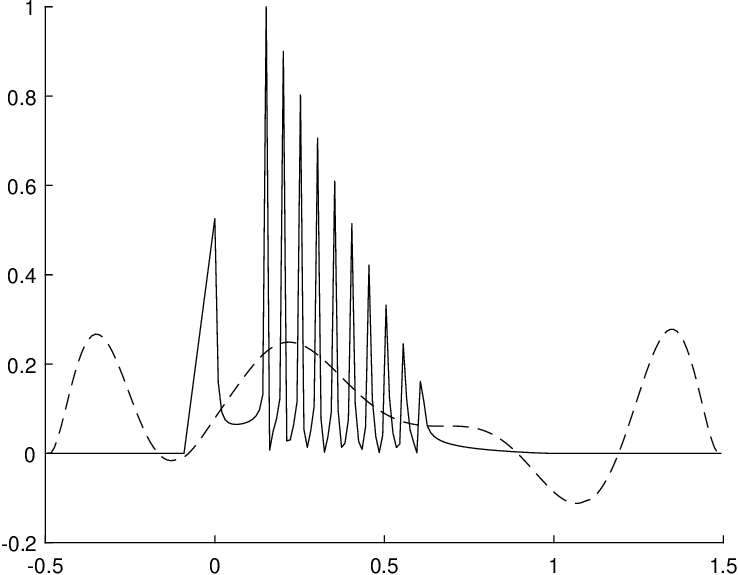}
  }
  \subfloat[]{
    \includegraphics[width=0.3\columnwidth]{converge_test3_1000.eps}
  }
  \subfloat[]{
  	\includegraphics[width=0.3\columnwidth]{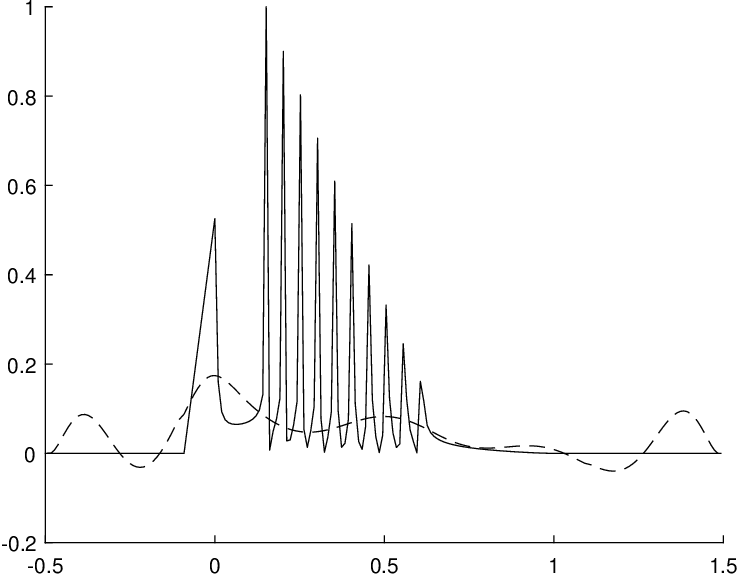}
  }
  \caption{The iteration for finding the cutting curve of function 
 $y(t)=6t + \displaystyle\sum\limits_{i=1}^{10} (13-i)\cos[(20+10i)\pi t]$, The 500th(a), 1000th(b), 1500th(c) iterations are respectively shown.}
  \label{diag16}

\end{figure}

\begin{figure}[H]
  \centering
  \subfloat[]{
    \includegraphics[width=0.4\columnwidth]{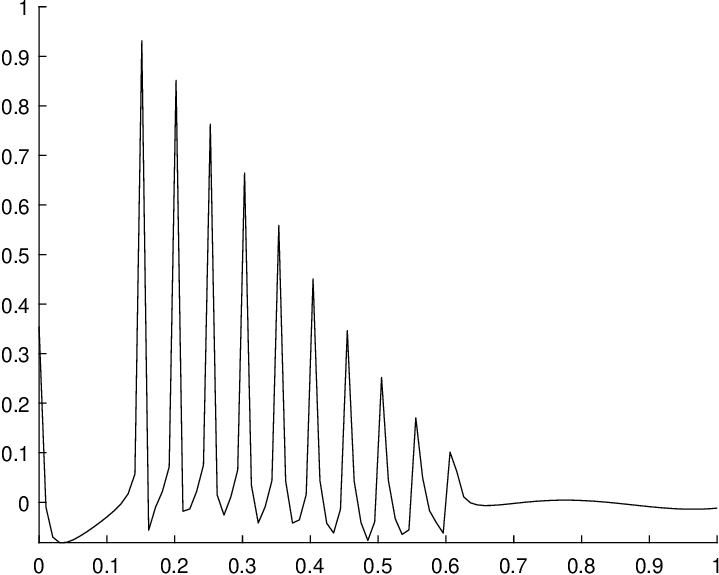}
  }
  \subfloat[]{
    \includegraphics[width=0.4\columnwidth]{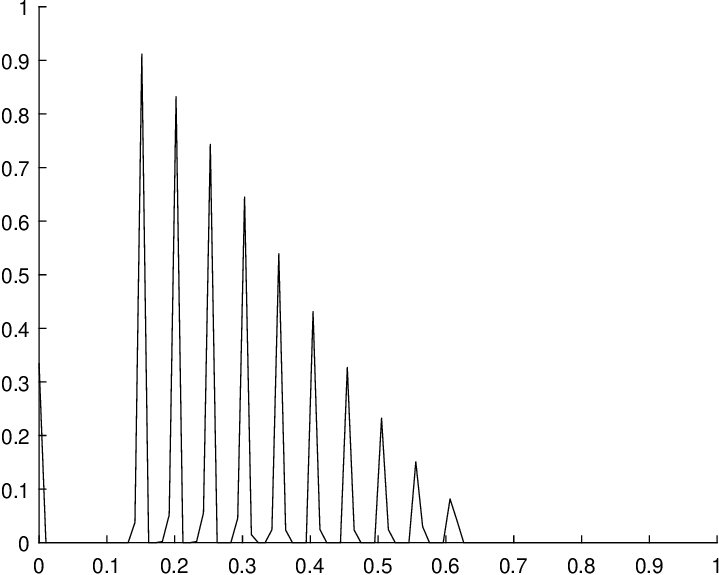}
  }

  \caption{The gap-elimination course for the found spectrum supporting line with series $y(t)=6t + \displaystyle\sum\limits_{i=1}^{10} (13-i)\cos[(20+10i)\pi t]$.  (a) is the result of source spectrum subtracting found cutting curve, and one can see that there is some additional noise-gap in the bottom of line subtracted function. (b) is the result after our KDE process to eliminate the gap at the bottom .}
  \label{diag17}
\end{figure}

\begin{figure}[H]
\centering
\makebox[\columnwidth][c]{
  \includegraphics[trim=70 80 60 60, clip,width=\columnwidth, height=18cm]{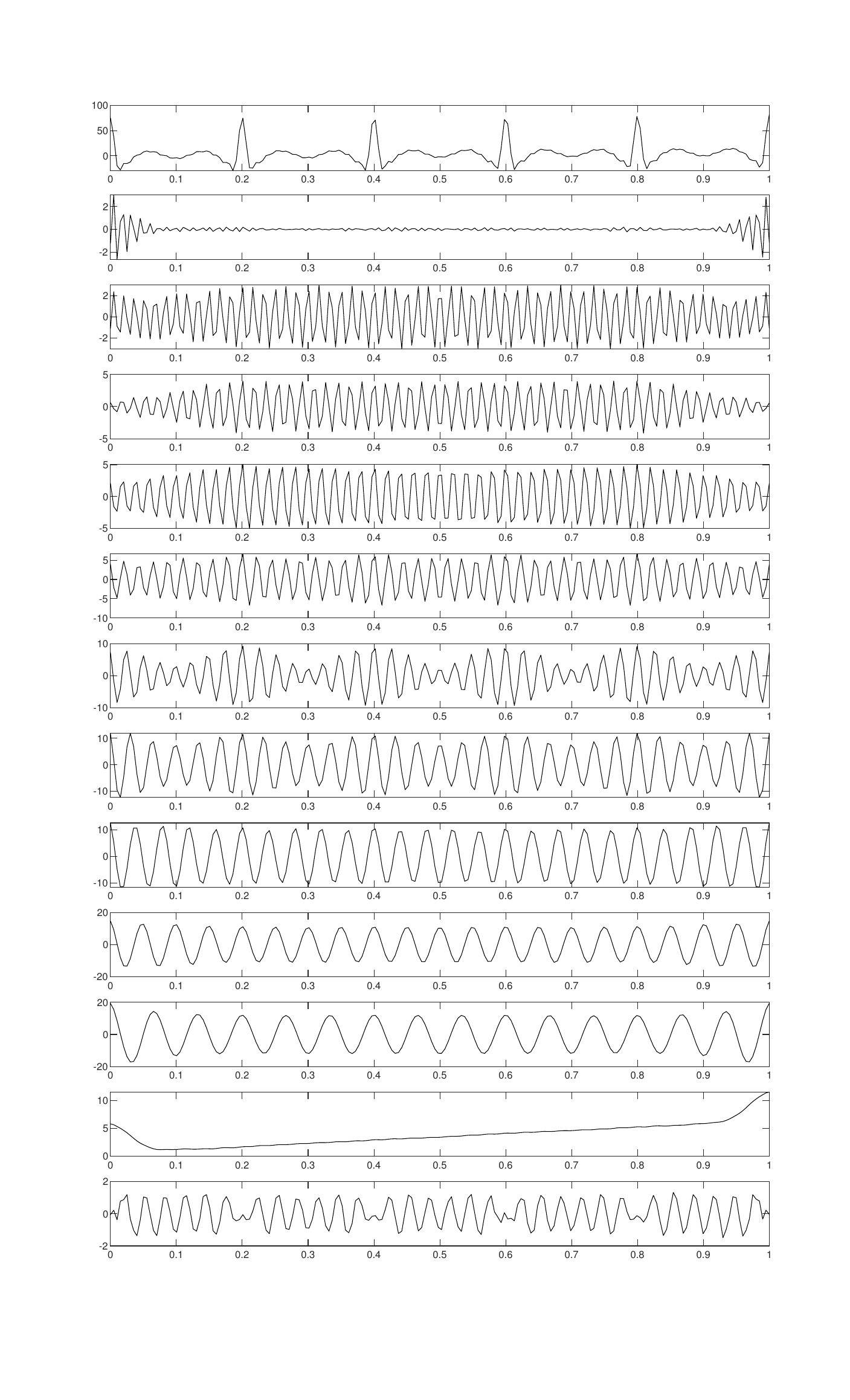}
}
\caption{The decomposition result of $y(t)=6t + \displaystyle\sum\limits_{i=1}^{10} (13-i)\cos[(20+10i)\pi t]$ from the subsequent VMD, using the found number of modes and center frequency as initialization.  The sub-figure from top to bottom are respectively the source signal, the decomposed modes(there are three) and the residual.}
	\label{diag18}
\end{figure}


\begin{thebibliography}{00}
\bibitem{Gianmarco}
{\sc G.~Baldini and F.~Bonavitacola}, {\em Channel identification with improved
  variational mode decomposition}, Physical Communication, 55 (2022),
  p.~101871,
  \url{https://doi.org/https://doi.org/10.1016/j.phycom.2022.101871},
  \url{https://www.sciencedirect.com/science/article/pii/S1874490722001495}.

\bibitem{BertsekasCOA}
{\sc D.~Bertsekas}, {\em Convex Optimization Algorithms}, Athena Scientific,
  2015, \url{https://books.google.com.sg/books?id=OwQ7EAAAQBAJ}.

\bibitem{bve}
{\sc W.~E. Boyce, R.~C. DiPrima, and D.~B. Meade}, {\em Boyce's elementary
  differential equations and boundary value problems}, 2017 - 2017.

\bibitem{Boyd}
{\sc S.~Boyd and L.~Vandenberghe}, {\em Convex Optimization}, Cambridge
  University Press, 2004.

\bibitem{Dragomiretskiy}
{\sc K.~Dragomiretskiy and D.~Zosso}, {\em Variational mode decomposition},
  IEEE Transactions on Signal Processing, 62 (2014), pp.~531--544,
  \url{https://doi.org/10.1109/TSP.2013.2288675}.

\bibitem{Feng}
{\sc Z.~Feng, D.~Zhang, and M.~J. Zuo}, {\em Adaptive mode decomposition
  methods and their applications in signal analysis for machinery fault
  diagnosis: A review with examples}, IEEE Access, 5 (2017), pp.~24301--24331,
  \url{https://doi.org/10.1109/ACCESS.2017.2766232}.

\bibitem{Kumaraguruparan}
{\sc G.~Kumaraguruparan and M.~K. Hota}, {\em Microseismic signal denoising
  based on variational mode decomposition with adaptive non-local means
  filtering}, Pure and Applied Geophysics, 180 (2023), pp.~1--23,
  \url{https://doi.org/10.1007/s00024-023-03365-0}.

\bibitem{Dazhong}
{\sc D.~Lao}, {\em Fundamentals of the Calculus of Variations(3rd Edition)},
  National Defense Industry Press, 2015.

\bibitem{Li}
{\sc C.~Li, Y.~Wu, H.~Lin, J.~Li, F.~Zhang, and Y.~Yang}, {\em Ecg denoising
  method based on an improved vmd algorithm}, IEEE Sensors Journal, 22 (2022),
  pp.~22725--22733, \url{https://doi.org/10.1109/JSEN.2022.3214239}.

\bibitem{Yunqian}
{\sc Y.~Li, D.~Huang, and Z.~Qin}, {\em A classification algorithm of fault
  modes-integrated lssvm and pso with parameters’ optimization of vmd},
  Mathematical problems in engineering, 2021 (2021), p.~6627367.

\bibitem{Lian}
{\sc J.~Lian, Z.~Liu, H.~Wang, and X.~Dong}, {\em Adaptive variational mode
  decomposition method for signal processing based on mode characteristic},
  Mechanical Systems and Signal Processing, 107 (2018), pp.~53--77,
  \url{https://doi.org/10.1016/j.ymssp.2018.01.019}.

\bibitem{Mojtaba}
{\sc M.~Nazari and S.~M. Sakhaei}, {\em Successive variational mode decomposition}, Signal Processing, 174 (2020), p.~107610,
  \url{https://doi.org/https://doi.org/10.1016/j.sigpro.2020.107610},
  \url{https://www.sciencedirect.com/science/article/pii/S0165168420301535}.

\bibitem{Rehman}
{\sc N.~u. Rehman and H.~Aftab}, {\em Multivariate variational mode
  decomposition}, IEEE Transactions on Signal Processing, 67 (2019),
  pp.~6039--6052, \url{https://doi.org/10.1109/TSP.2019.2951223}.

\bibitem{Sharma}
{\sc V.~Sharma}, {\em Gear fault detection based on instantaneous frequency
  estimation using variational mode decomposition and permutation entropy under
  real speed scenarios}, Wind Energy, 24 (2020),
  \url{https://doi.org/10.1002/we.2570}.

\bibitem{Group}
{\sc Theoretical and C.~B. Group}, {\em For more information on vmd and
  mdscope}, \url{https://www.ks.uiuc.edu/Research/vmd/vmd-1.3/ug/node5.html}.

\bibitem{Wen-Chao}
{\sc B.~Wen-Chao, S.~Liang-Duo, C.~Liang, and X.~Chu-Tian}, {\em Monthly runoff
  prediction based on variational modal decomposition combined with the dung
  beetle optimization algorithm for gated recurrent unit model}, Environmental
  Monitoring and Assessment, 195 (2023),
  \url{https://doi.org/10.1007/s10661-023-12102-y}.

\bibitem{Wu}
{\sc S.~Wu, F.~Feng, J.~Zhu, C.~Wu, and G.~Zhang}, {\em A method for
  determining intrinsic mode function number in variational mode decomposition
  and its application to bearing vibration signal processing}, Shock and
  Vibration, 2020 (2020), pp.~1--16,
  \url{https://doi.org/10.1155/2020/8304903}.

\bibitem{Xia}
{\sc Y.-k. Xia, W.-t. Wang, and X.-y. Li}, {\em Adaptive parameter selection
  variational mode decomposition based on bayesian optimization and its
  application to the detection of itsc in pmsm}, IEEE Access, PP (2024),
  pp.~1--1, \url{https://doi.org/10.1109/ACCESS.2024.3373880}.

\bibitem{Yang}
{\sc J.~Yang, E.~Stewart, J.~Ye, M.~Entezami, and C.~Roberts}, {\em An improved
  vmd method for use with acoustic impact response signals to detect corrosion
  at the underside of railway tracks}, Applied Sciences, 13 (2023), p.~942,
  \url{https://doi.org/10.3390/app13020942}.

\bibitem{Dehao}
{\sc D.~Yu and H.~Tang}, {\em Numberical Solutions of Differential Equations},
  Science Press, 2018.

\bibitem{ShangZhang}
{\sc S.~Zhang, G.~Liu, R.~Xiao, W.~Cui, J.~Cai, X.~Hu, Y.~Sun, J.~Qiu, and
  Y.~Qi}, {\em A combination of statistical parameters for epileptic seizure
  detection and classification using vmd and nltwsvm}, Biocybernetics and
  Biomedical Engineering, 42 (2022), pp.~258--272,
  \url{https://doi.org/https://doi.org/10.1016/j.bbe.2022.02.004},
  \url{https://www.sciencedirect.com/science/article/pii/S0208521622000079}.

\bibitem{Zheng}
{\sc Q.~Zheng, T.~Chen, L.~Xie, and H.~Su}, {\em Short-time variational mode
  decomposition: Algorithms, extensions and properties}, SSRN Electronic
  Journal,  (2022), \url{https://doi.org/10.2139/ssrn.4080800}.

\bibitem{MIT-BIH}
{\sc Moody GB, Mark RG. The impact of the MIT-BIH Arrhythmia Database. IEEE Eng in Med and Biol 20(3):45-50 (May-June 2001). (PMID: 11446209)}.

\bibitem{Goldberger}
{\sc Goldberger, A. L., Amaral, L. A., Glass, L., Hausdorff, J. M., Ivanov, P. C., Mark, R. G., Mietus, J. E., Moody, G. B., Peng, C. K., \& Stanley, H. E. (2000). PhysioBank, PhysioToolkit, and PhysioNet: components of a new research resource for complex physiologic signals. Circulation, 101(23), E215–E220. https://doi.org/10.1161/01.cir.101.23.e215}.

\bibitem{Matlab}
{Mojtaba Nazari (2026). Successive Variational Mode Decomposition (SVMD.m) \url{https://ww2.mathworks.cn/matlabcentral/fileexchange/98649-successive-variational-mode-decomposition-svmd-m}, MATLAB Central File Exchange}.

\bibitem{Collins2006}
{\sc Collins, P. J. (2006). {\it Differential and Integral Equations}. Oxford University Press, Oxford, UK.}

\bibitem{kress2013linear}
{\sc Kress, R. (2013). {\it Linear Integral Equations}. Applied Mathematical Sciences, Springer New York.}

\bibitem{keener2019principles}
{\sc Keener, J. P. (2019). {\it Principles Of Applied Mathematics: Transformation and Approximation}. CRC Press.}

\end{thebibliography}


\begin{thebibliography}{1}

\bibitem{Dragomiretskiy}
{\sc K.~Dragomiretskiy and D.~Zosso}, {\em Variational mode decomposition},
  IEEE Transactions on Signal Processing, 62 (2014), pp.~531--544,
  \url{https://doi.org/10.1109/TSP.2013.2288675}.

\bibitem{Mojtaba}
{\sc M.~Nazari and S.~M. Sakhaei}, {\em Successive variational mode
  decomposition}, Signal Processing, 174 (2020), p.~107610,
  \url{https://doi.org/https://doi.org/10.1016/j.sigpro.2020.107610},
  \url{https://www.sciencedirect.com/science/article/pii/S0165168420301535}.

\end{thebibliography}


\clearpage
\end{document}